\documentclass[11pt]{article}
\pdfoutput=1

\newif\ifnotes
\notestrue

\usepackage{hyperref}
\usepackage{fullpage}
\usepackage{amssymb}
\usepackage{amsmath}
\usepackage{mathtools}

\usepackage{tikz}
\usetikzlibrary{graphs, quotes}

\usepackage{amsthm}
\usepackage{amsfonts}
\usepackage{bm}
\usepackage{enumitem}
\usepackage{color}
\usepackage{comment}
\usepackage[capitalize]{cleveref}
\usepackage{xcolor}
\usepackage{float}
\usepackage[T1]{fontenc}
\usepackage{todonotes}
\usepackage{asymptote}
\usepackage{mdframed}
\usepackage[most]{tcolorbox}
\usepackage{hyperref}
\usepackage{enumitem}
\usepackage{framed}
\usepackage{mdframed}
\usepackage{scrextend}
\usepackage{bbm}

\usepackage{mathrsfs}

\usepackage[T1]{fontenc}

\definecolor{asparagus}{rgb}{0.53, 0.66, 0.42}
\definecolor{sacramentostategreen}{rgb}{0.0, 0.3, 0.15}
\definecolor{teal}{rgb}{0.0, 0.5, 0.5}
\definecolor{forestgreen}{rgb}{0.13, 0.6, 0.13}

\usepackage{hyperref}
\hypersetup{
    colorlinks=true,
    linkcolor=teal,
    filecolor=sacramentostategreen,
    citecolor=teal,
    urlcolor=teal,
}

\usepackage{thmtools}
\usepackage{thm-restate}
\usepackage{multicol}

\newtheorem{theorem}{Theorem}[section]
\newtheorem{definition}[theorem]{Definition}

\newtheorem{claim}[theorem]{Claim}

\newtheorem{corollary}[theorem]{Corollary}
\newtheorem{remark}[theorem]{Remark}

\newtheorem{mainproblem}{Problem}
\newtheorem{mainthm}{Theorem}
\newtheorem{maindefn}{Definition}
\newtheorem*{problem*}{Problem}

\theoremstyle{remark}

\Crefname{theorem}{Theorem}{Theorems}
\Crefname{claim}{Claim}{Claims}
\Crefname{lemma}{Lemma}{Lemmas}
\Crefname{proposition}{Proposition}{Propositions}
\Crefname{corollary}{Corollary}{Corollaries}
\Crefname{definition}{Definition}{Definitions}

\newcommand{\E}{\mathop{{}\mathbb{E}}}

\newcommand{\R}{\mathbb{R}}

\newcommand{\cA}{\mathcal{A}}

\newcommand{\cC}{\mathcal{C}}

\newcommand{\cE}{\mathcal{E}}

\newcommand{\cH}{\mathcal{H}}

\newcommand{\cO}{\mathcal{O}}
\newcommand{\cP}{\mathcal{P}}

\newcommand{\mA}{\mathbf{A}}

\newcommand{\mD}{\mathbf{D}}

\newcommand{\mG}{\mathbf{G}}

\newcommand{\mI}{\mathbf{I}}

\newcommand{\mM}{\mathbf{M}}

\newcommand{\mQ}{\mathbf{Q}}

\newcommand{\mS}{\mathbf{S}}
\newcommand{\mU}{\mathbf{U}}
\newcommand{\mV}{\mathbf{V}}
\newcommand{\mW}{\mathbf{W}}

\newcommand{\mZ}{\mathbf{Z}}

\newcommand{\mSigma}{\mathbf{\Sigma}}

\newcommand{\va}{\bm{a}}

\newcommand{\vc}{\bm{c}}

\newcommand{\ve}{\bm{e}}

\newcommand{\vm}{\bm{m}}

\newcommand{\vq}{\bm{q}}
\newcommand{\vr}{\bm{r}}

\newcommand{\vu}{\bm{u}}
\newcommand{\vv}{\bm{v}}
\newcommand{\vw}{\bm{w}}
\newcommand{\vx}{\bm{x}}
\newcommand{\vy}{\bm{y}}
\newcommand{\vz}{\bm{z}}

\def\to{\rightarrow}
\def\eps{\varepsilon}

\def\eps{\varepsilon}

\def\cal{\mathcal}

\usepackage{algorithm}
\usepackage[noend]{algpseudocode}

\newcommand{\defeq}{\coloneqq}

\newcommand{\vspan}[1]{\mathsf{aff}\inparen{#1}}
\newcommand{\lspan}[1]{\mathsf{span}\inparen{#1}}

\newcommand{\suchthat}{{\;\; : \;\;}}

\newcommand{\inparen}[1]{\left(#1\right)}             
\newcommand{\inbraces}[1]{\left\{#1\right\}}           
\newcommand{\insquare}[1]{\left[#1\right]}             
\newcommand{\inangle}[1]{\left\langle#1\right\rangle} 

\newcommand{\abs}[1]{\ensuremath{\left\lvert #1 \right\rvert}}

\newcommand{\norm}[1]{\ensuremath{\left\lVert #1 \right\rVert}}

\newcommand{\fnorm}[1]{\ensuremath{\left\lVert #1 \right\rVert_{F}}}

\newcommand{\ip}[1]{\left\langle #1 \right\rangle}

\usepackage{nicefrac}

\let\nfrac=\nicefrac

\newcommand{\Tr}{\operatorname{tr}}
\newcommand{\unif}{\mathsf{Unif}}

\newcommand{\logv}[1]{\log\inparen{#1}}
\newcommand{\diag}[1]{\mathsf{diag}\inparen{#1}}
\newcommand{\diam}[1]{\mathsf{diam}\inparen{#1}}

\newcommand{\so}[1]{\mathbb{O}\inparen{#1}}
\newcommand{\expv}[1]{\mathsf{exp}\inparen{#1}}

\makeatletter
\newcommand{\customlabel}[2]{%
   \protected@write \@auxout {}{\string \newlabel {#1}{{#2}{\thepage}{#2}{#1}{}} }%
   \hypertarget{#1}{#2}
}
\makeatother

\newcounter{casenum}

\newcounter{subcasenum}

\newcounter{casenump}

\newcommand{\casep}[2]{
    \ifthenelse{\equal{\value{casenump}}{0}}{
    \vskip.5\baselineskip\par\noindent
    }{}
    {\it Case \arabic{casenump}:} {\it #1}
    \vskip0.1\baselineskip
    \begin{addmargin}[1.5em]{1em}
    #2
    \end{addmargin}
    \addtocounter{casenump}{1}
}

\newcounter{subcasenump}

\usepackage{parskip}
\usepackage[
natbib=true,
style=trad-alpha,
maxalphanames=8,
sorting=nty,
backref=true
]{biblatex}
\usepackage{cleveref}

\usepackage{ulem}

\newcommand{\detv}[1]{\mathsf{det}\inparen{#1}}

\newcommand{\conv}[1]{\mathsf{conv}\inparen{#1}}

\renewcommand{\S}{\mathbb{S}}

\definecolor{ForestGreen}{RGB}{34, 139, 34}

\definecolor{forest}{RGB}{0,155,85}

\newcommand{\fixout}{\bgroup\markoverwith{\textcolor{forest}{\rule[0.5ex]{2pt}{0.4pt}}}\ULon}

\newcommand{\vol}{\operatorname{\mathsf{vol}}}
\newcommand{\opt}{\mathsf{OPT}}

\numberwithin{equation}{section}
\numberwithin{algorithm}{section}

\begin{document}

\title{Near-Optimal Streaming Ellipsoidal Rounding for General Convex Polytopes}
\author{Yury Makarychev
\thanks{Toyota Technological Institute at Chicago. Email: \texttt{yury@ttic.edu}. Supported by NSF Awards CCF-1955173, CCF-1934843, and ECCS-2216899.}
\and Naren Sarayu Manoj  
\thanks{Toyota Technological Institute at Chicago. Email: \texttt{nsm@ttic.edu}. Supported by NSF Graduate Research Fellowship.}
\and Max Ovsiankin
\thanks{Toyota Technological Institute at Chicago. Email: \texttt{maxov@ttic.edu}. Supported by NSF Award ECCS-2216899}}
\date{\today}

\sloppy
\maketitle
\thispagestyle{empty}
\begin{abstract}
    We give near-optimal algorithms for computing an ellipsoidal rounding of a convex polytope whose vertices are given in a stream.
    The approximation factor is linear in the dimension (as in John's theorem) and only loses an excess logarithmic factor in the aspect ratio of the polytope.
    Our algorithms are nearly optimal in two senses: first, their runtimes nearly match those of the most efficient known algorithms for the offline version of the problem. 
    Second, their approximation factors nearly match a lower bound we show against a natural class of geometric streaming algorithms.
    In contrast to existing works in the streaming setting that compute ellipsoidal roundings only for centrally symmetric convex polytopes, our algorithms apply to general convex polytopes.

    We also show how to use our algorithms to construct coresets from a stream of points that approximately preserve both the ellipsoidal rounding and the convex hull of the original set of points.
\end{abstract}
\newpage
\tableofcontents
\newpage
\pagenumbering{arabic}

\newcommand{\Gbar}{\overline{\mG}}
\newcommand{\cbar}{\overline{\vc}}


\section{Introduction}

We consider the problem of approximating convex polytopes in $\R^d$ with ``simpler'' convex bodies.
Consider a convex polytope $Z \subset \R^d$.
Our goal is to find a convex body $\widehat{Z}\subset \R^d$ from a given family of convex bodies, a translation vector $\vc \in \R^d$, and a scaling factor $\alpha \in (0, 1]$ such that
\begin{align}
    \vc + \alpha \cdot \widehat{Z} \subseteq Z \subseteq \vc + \widehat{Z}\label{eq:intro_approx_general}.
\end{align}
We say that $\widehat{Z}$ is a $\nfrac{1}{\alpha}$-approximation to $Z$; an algorithm that computes $\widehat{Z}$ is a $\nfrac{1}{\alpha}$-approximation algorithm. In this paper, we will be interested in approximating $Z$ with (a) ellipsoids and (b) polytopes defined by small number of vertices.

This problem has many applications in computational geometry, graphics, robotics, data analysis, and other fields (see \citep{agarwal2005geometric} for an overview of some applications).
It is particularly relevant when we are in the big-data regime and  storing polytope $Z$ requires too much memory. In this case, instead of storing $Z$, we find a reasonable approximation $\widehat{Z}$ with a succinct representation and then use it as a proxy for $Z$. In this setting, it is crucial that we use a \textit{low-memory} approximation algorithm to find $\widehat{Z}$.

In this paper, we study the problem of approximating convex polytopes in the streaming model. The streaming model is a canonical big-data setting that conveniently lends itself to the study of low-memory algorithms. 
We assume that $Z$ is the convex hull of points $\vz_1,\dots, \vz_n$: $Z = \conv{\{\vz_1,\dots,\vz_n\}}$; the stream of points $\{\vz_1,\dots,\vz_n\}$ contains all the vertices of $Z$ and additionally may contain other points from polytope $Z$.
In our streaming model, points $\vz_1,\dots, \vz_n$ arrive one at a time. At every timestep $t$, we must maintain an approximating body $\widehat{Z}_t$ and translate $\vc_t$ such that
\begin{align}
    \conv{\inbraces{\vz_1,\dots,\vz_t}} \subseteq \vc_t + \widehat{Z}_t.\label{eq:intro_approx_streaming}
\end{align}
Once a new point $\vz_{t+1}$ arrives, the algorithm must compute a new approximating body $\widehat{Z}_{t+1}$ and translation $\vc_{t+1}$ such that the guarantee \eqref{eq:intro_approx_streaming} holds for timestep $t+1$. Finally, after the algorithm has seen all $n$ points, we must have 
\begin{align}
    \vc_n + \alpha \cdot \widehat{Z}_n \subseteq \underbrace{\conv{\inbraces{\vz_1,\dots,\vz_n}}}_{Z} \subseteq \vc_n + \widehat{Z}_n.\label{eq:intro_approx_final}
\end{align}
for some $0 < \alpha \le 1$ (where $\nicefrac{1}{\alpha}$ is the approximation factor).
Note that the algorithm may not know the value of $n$ beforehand. We consider two types of approximation. 

\paragraph{Ellipsoidal roundings.} In one thrust, we aim to calculate an \textit{ellipsoidal rounding} of $Z$ -- we are looking for ellipsoidal approximation $\widehat Z=\cal E$. Formally, we would like to output an origin-centered ellipsoid $\cE$, a center/translate $\vc \in \R^d$, and a scaling parameter $0 < \alpha \le 1$ such that
\begin{align*}
    \vc + \alpha \cdot \cE \subseteq Z \subseteq \vc + \cE.
\end{align*}
Ellipsoidal roundings are convenient representations of convex sets. They have applications to preconditioning convex sets for efficient sampling and volume estimation \cite{he21}, algorithms for convex programming \cite{nesterov2008rounding}, robotics \cite{boyd97}, and other areas. They also require the storage of at most $\sim d^2$ floating point numbers, as every ellipsoid can be represented with a center $\vc$ and semiaxes $\vv_1,\dots,\vv_{d'}$ for $d' \le d$.

We note that by John's theorem \cite{john1948}, the minimum-volume outer ellipsoid for $Z$ achieves approximation $\nfrac{1}{\alpha} \le d$. Moreover, the upper bound of $d$ is tight, which is witnessed when $Z$ is a $d$-dimensional simplex (that is, the convex hull of $d+1$ points in general position).

We now formally state the streaming ellipsoidal rounding problem.

\begin{mainproblem}[Streaming ellipsoidal rounding]\label{prob:main}
Let \(Z = \conv{\inbraces{\vz_1, \ldots, \vz_n}} \subseteq \R^d\). 
A streaming algorithm \(\cA\)  receives points \(\vz_1, \ldots, \vz_n\) one at a time  and produces a sequence of ellipsoids \(\vc_t + \cE_t\) and scalings \(\alpha_t\).
The algorithm must satisfy the following guarantee at the end of the stream:
\[\vc_n + \alpha_n \cdot \cE_n \subseteq Z \subseteq c_n + \cE_n\]
We say that \(\vc_n + \cE_n\) is an ellipsoidal rounding of \(Z\) with approximation factor \(\nfrac{1}{\alpha_n}\).
\end{mainproblem}

We note that in the special case where $Z$ is centrally symmetric (i.e., $Z = -Z$), there are algorithms with nearly optimal approximation factors  $O(\sqrt{d\logv{n\kappa^{\mathsf{OL}}}})$ and $O(\sqrt{d\log \kappa})$ due to \citet{woodruff2022high} and \citet{mmo22}, respectively (here, $\kappa^{\mathsf{OL}}$ is the online condition number and $\kappa$ is the aspect ratio of the dataset). The running times of these algorithms nearly match those of the best-known offline solutions. However, these algorithms do not work with non-symmetric polytopes and we are not aware of any way to adapt them so that they do. We defer a more detailed discussion of the algorithms for the symmetric case to Section~\ref{sec:rw}. 

\paragraph{Convex hull approximation.} In another thrust, we want to find a translate $\vc \in \R^d$, subset $S \subseteq [n]$, and scale $\alpha$ such that
\begin{align*}
\conv{\inbraces{\vz_i: i \in S}} \subseteq \conv{\inbraces{\vz_1,\dots,\vz_n}} \subseteq  \vc + \frac{1}{\alpha} \cdot \conv{\inbraces{\vz_i - \vc : i \in S}}.
\end{align*}
Note that $\vc + \nfrac{1}{\alpha} \cdot \conv{\inbraces{\vz_i - \vc : i \in S}}$ is a $\nfrac{1}{\alpha}$-scaled copy of $\conv{\inbraces{\vz_i : i \in S}}$. In other words, we desire to find a \textit{coreset} $\inbraces{\vz_i : i \in S}$ that approximates $Z$. This approach has the advantage of yielding an interpretable solution -- one can think of a coreset as consisting of the most ``important'' datapoints of the input dataset.

We formally state the streaming convex hull approximation problem we study in Problem \ref{prob:coreset}.

\begin{mainproblem}[Streaming convex hull approximation]
\label{prob:coreset}
Let \(Z = \conv{\vz_1, \ldots, \vz_n} \subseteq \R^d\).
A streaming algorithm \(\cA\) receives points \(\vz_1, \ldots, \vz_n\) one at a time and produces a sequence of scalings \(\alpha_t\), centers \(\vc_t\), subsets \(S_t \subseteq [n]\) such that \(S_t \subseteq S_{t+1}\).
The algorithm must satisfy the following guarantee at the end of the stream.
\[\conv{\inbraces{\vz_i : i \in S_n}} \subseteq \conv{\inbraces{\vz_1,\dots,\vz_n}} \subseteq \vc_n + \frac{1}{\alpha} \cdot \conv{\inbraces{\vz_i - \vc_n : i \in S_n}}\]
We say that $\inbraces{\vz_i : i \in S_n}$ is a coreset of \(Z\) with approximation factor \(\nfrac{1}{\alpha_n}\). We will also call \(S_n\) a coreset.
\end{mainproblem}
Note that the model considered in Problem \ref{prob:coreset} is essentially the same as the \textit{online coreset model} studied by \citet{woodruff2022high}.
Similar to Problem \ref{prob:main}, Problem \ref{prob:coreset} has been studied in the case where $Z$ is centrally symmetric. In particular, \citet{woodruff2022high} obtain approximation factor $O(\sqrt{d\logv{n\kappa^{\mathsf{OL}}}})$ (where $\kappa^{\mathsf{OL}}$ is the same online condition number mentioned earlier).
However, whether analogous results for asymmetric polytopes hold was an important unresolved question.

\subsection{Our contributions}

In this section, we present our results for Problems \ref{prob:main} and \ref{prob:coreset}.

\subsubsection{Algorithmic results}

We start with defining several quantities that we need to state the results and describe their proofs.

\paragraph{Notation.} We will denote the linear span of a set of points $A$ by $\lspan{A}$. That is, $\lspan{A}$ is the minimal linear subspace that contains $A$. We denote the affine span of $A$ by $\vspan{A}$. That is, $\vspan{A}$ is the minimal affine subspace that contains $A$. Note that $\vspan{A} = \va + \lspan{A-\va}$ if $\va \in A$. Finally, we denote the unit ball centered at the origin by $B_2^d$.

\begin{maindefn}[Inradius]
\label{def:inradius}
Let $K \subset \R^d$ be a convex body. The \textit{inradius} $r(K)$ of $K$ is the largest $r$ such that there exists a point $\vc_I$ (called the \textit{incenter}) for which $\vc_I + r \cdot \inparen{B_2^d \cap \lspan{K - \vc_I}} \subseteq K$. 
\end{maindefn}

\begin{maindefn}[Circumradius]
\label{def:circumradius}
Let $K \subset \R^d$ be a convex body.  The \textit{circumradius} $R(K)$ of $K$ is the smallest $R$ such that there exists a point $\vc_C$ (called the \textit{circumcenter}) for which $K \subseteq \vc_C + R \cdot B_2^d$.
\end{maindefn}

\begin{maindefn}[Aspect Ratio]
\label{def:aspect-ratio}
Let $K \subset \R^d$ be a convex body. We say that $\kappa(K) \coloneqq \nfrac{R(K)}{r(K)}$ is the \textit{aspect ratio} of $K$.
\end{maindefn}

We now state Theorem \ref{thm:main_one}, which provides an algorithm for Problem \ref{prob:main}. In addition to the data stream of $z_1, \dots, z_n$, this algorithm needs a suitable initialization: a ball $\vc_0 + r_0 \cdot B_2^d$ inside $Z$.

\begin{mainthm}
\label{thm:main_one}
Consider the setting of Problem \ref{prob:main}. Suppose the algorithm is given an initial center $\vc_0$ and radius $r_0$ for which it is guaranteed that $\vc_0 + r_0 \cdot B_2^d \subseteq \conv{\inbraces{\vz_1,\dots,\vz_n}}$.
There exists an algorithm (Algorithm \ref{alg:main}) that, for every timestep $t$, maintains an origin-centered ellipsoid $\cE_t$, center $\vc_t$, and scaling factor $\alpha_t$ such that at every timestep $t$: $\conv{\inbraces{\vz_1,\dots,\vz_t}} \subseteq \vc_t + \cE_t$ 
and at timestep $n$: $\vc_n + \alpha_n \cdot \cE_n \subseteq Z \subseteq c_n + \cE_n$, where  \[\nfrac{1}{\alpha_n} = O\inparen{\min\inparen{ \nfrac{R(Z)}{r_0}, d\logv{\nfrac{R(Z)}{r_0}}}}\] The algorithm has runtime $\widetilde{O}(nd^2)$ and stores $O(d^2)$ floating point numbers.
\end{mainthm}

Note that 
the final approximation factor depends on the quality of the initialization $(\vc_0, r_0)$. If the radius $r_0$ of this ball is reasonably close to the inradius $r(Z)$ of $Z$, the algorithm gives an $O(\min (\kappa(Z),d\log \kappa(Z)))$ approximation.
In Theorem \ref{thm:main_two}, we adapt the algorithm form Theorem~\ref{thm:main_one} to the setting where the algorithm does not have the initialization information.
Note that the approximation guarantee of $O(\min (\kappa(Z),d\log \kappa(Z)))$ is a natural analogue of the bounds by~\cite{mmo22} and \cite{woodruff2022high} for the symmetric case (see Section~\ref{sec:rw}).


\begin{mainthm}
\label{thm:main_two}
Consider the setting of Problem \ref{prob:main}. There exists an algorithm (Algorithm \ref{alg:fully_online_approx}) that, for every timestep $t$, maintains an ellipsoid $\cE_t$, center $\vc_t$, and approximation factor $\alpha_t$ such that
\begin{align*}
    \vc_t + \alpha_t \cdot \cE_t \subseteq \conv{\inbraces{\vz_1,\dots,\vz_t}} \subseteq \vc_t + \cE_t.
\end{align*}
Additionally, let $r_t$ and $R_t$ be the largest and smallest parameters, respectively, for which there exists $\vc^{\star}_t$ such that
\begin{align*}
    \vc^{\star}_t + r_t \cdot \inparen{B_2^d \cap \lspan{\vz_1-\vc^{\star}_t,\dots,\vz_t-\vc^{\star}_t}} \subseteq \conv{\inbraces{\vz_1,\dots,\vz_t}} \subseteq \vc^{\star}_t + R_t \cdot B_2^d
\end{align*}
and $d_t \coloneqq \mathsf{dim}\inparen{\vspan{\vz_1,\dots,\vz_t}}$.
Then, for all timesteps $t$, we have
\begin{align*}
    \nfrac{1}{\alpha_t} = O\inparen{d_t\logv{d_t \cdot \max_{t' \le t} \frac{R_t}{r_{t'}}}}.
\end{align*}
The algorithm runs in time $\widetilde{O}(nd^2)$ and stores $O(d^2)$ floating point numbers.
\end{mainthm}
Let us now quickly compare the guarantees of Theorem~\ref{thm:main_one} and~\ref{thm:main_two}. Notice that the algorithm in Theorem~\ref{thm:main_two} does not require an initialization pair $(\vc_0, r_0)$. Additionally, the algorithm in Theorem~\ref{thm:main_two} outputs a per-timestep approximation as opposed to just an approximation at the end of the stream. However, these advantages come at a cost -- it is easy to check that the aspect ratio term seen in Theorem \ref{thm:main_two} can be larger than that in Theorem \ref{thm:main_one}, e.g., it is possible to have $\nfrac{R(Z)}{r_0} \le \max_{t' \le n} \nfrac{R_n}{r_{t'}}$.

However, when we impose the additional constraint that the points $\vz_t$ have coordinates that are integers in the range $[-N,N]$, we can improve over the guarantee in \Cref{thm:main_two} and obtain results that are independent of the aspect ratio. This is similar in spirit to the condition number-independent bound that \citet{woodruff2022high} obtain for the sums of online leverage scores. However, a key difference is that our results still remain independent of the length of the stream. See \Cref{thm:main_two_ints}.

\begin{mainthm}
\label{thm:main_two_ints}
Consider the setting of Problem \ref{prob:main}, where in addition, the points $\vz_1,\dots,\vz_n$ are such that their coordinates are integers in $\inbraces{-N, -N+1, \dots, N-1, N}$. There exists an algorithm (Algorithm \ref{alg:fully_online_approx}) that, for every timestep $t$, maintains an ellipsoid $\cE_t$, center $\vc_t$, and approximation factor $\alpha_t$ such that
\begin{align*}
    \vc_t + \alpha_t \cdot \cE_t \subseteq \conv{\inbraces{\vz_1,\dots,\vz_t}} \subseteq \vc_t + \cE_t.
\end{align*}
Let $d_t \coloneqq \mathsf{dim}\inparen{\vspan{\vz_1,\dots,\vz_t}}$. Then, for all timesteps $t$, we have
\begin{align*}
    \nfrac{1}{\alpha_t} = O\inparen{d_t\logv{d N}}.
\end{align*}
The algorithm runs in time $\widetilde{O}(nd^2)$ and stores $O(d^2)$ floating point numbers.
\end{mainthm}

We prove Theorems \ref{thm:main_one}, \ref{thm:main_two}, and \ref{thm:main_two_ints} in Section \ref{sec:main_alg_top}.
With Theorems \ref{thm:main_two} and \ref{thm:main_two_ints} in hand, obtaining results for Problem \ref{prob:coreset} becomes straightforward. We use the algorithm guaranteed by Theorem \ref{thm:main_two} along with a simple subset selection criterion to arrive at our result for Problem \ref{prob:coreset}. 

\begin{mainthm}
\label{thm:main_three}
Consider $Z = \conv{\inbraces{\vz_1,\dots,\vz_n}}$. For a subset $S \subseteq [n]$, let $Z\vert_S=\conv{\inbraces{\vz_i : i \in S}}$.
Consider the setting of Problem \ref{prob:coreset}. There exists a streaming algorithm (Algorithm \ref{alg:ellipse_to_coreset}) that, for every timestep $t$, maintains a subset $S_t$, center $\vc_t$, and scaling factor $\alpha_t$ such that
\begin{align*}
    Z\vert_{S_t} \subseteq \conv{\inbraces{\vz_1,\dots,\vz_t}} \subseteq \vc_t + \frac{1}{\alpha_t} \cdot \inparen{Z\vert_{S_t} - \vc_t}.
\end{align*}
Additionally, for $d_t$, $r_t$ and $R_t$ as defined in Theorem \ref{thm:main_two}, we have for all $t$ that
\begin{align*}
    \frac{1}{\alpha_t} &= O\inparen{d_t\logv{d_t \cdot \max_{t' \le t} \frac{R_t}{r_{t'}}}} & \text{ and }& & \abs{S_t} &= O\inparen{d_t\logv{\max_{t' \le t} \frac{R_t}{r_{t'}}}},
\end{align*}
and, if the $\vz_t$ have integer coordinates ranging in $\insquare{-N, N}$, then
\begin{align*}
    \frac{1}{\alpha_t} &= O\inparen{d_t\logv{dN}} & \text{ and }& & \abs{S_t} &= O\inparen{d_t\logv{dN}}.
\end{align*}
Each $S_t$ is either $S_{t-1}$ or $S_{t-1} \cup \{t\}$ (where $t\geq 1$ and $S_0 = \varnothing$).
The algorithm runs in time $\widetilde{O}(nd^2)$ and stores at most $O(d^2)$ floating point numbers.
\end{mainthm}

We prove Theorem \ref{thm:main_three} in Section \ref{sec:coreset}.

\subsubsection{Approximability lower bound}

Observe that the approximation factors obtained in Theorems \ref{thm:main_one}, \ref{thm:main_two}, and \ref{thm:main_three} all incur a mild dependence on (variants of) the aspect ratio of the dataset. A natural question is whether this dependence is necessary. In Theorem \ref{thm:main_four}, we conclude that the approximation factor from Theorem~\ref{thm:main_one} is in fact nearly optimal for a wide class of \textit{monotone} algorithms. We defer the discussion of the notion of a monotone algorithm to Section~\ref{sec:monotone}. Loosely speaking, a monotone algorithm commits to the choices it makes; namely, the outer ellipsoid may only increase over time $\vc_{t} + \cE_{t} \supseteq \vc_{t-1} + \cE_{t-1}$ and the inner ellipsoid $\vc_t + \alpha_t \cE_t$ satisfies a related but more technical condition $\vc_t + \alpha_t \cE_t \subseteq \conv{(\vc_{t-1} + \alpha_{t-1} \cdot \cE_{t-1}) \cup \{\vz_t\}}$.

\begin{mainthm}
\label{thm:main_four}
Consider the setting of Problem \ref{prob:main}. Let \(\cA\) be any monotone algorithm (see  Definition \ref{def:alg-invariant} in Section~\ref{sec:monotone}) that solves \Cref{prob:main} with approximation factor $\nfrac{1}{\alpha_n}$.
For every \(d \ge 2\), there exists a sequence of points \(\inbraces{\vz_1, \ldots, \vz_n} \subset \R^d\) 
such that algorithm \(\cA\) gets  an approximation factor of \(\nfrac{1}{\alpha_n} \geq \Omega\inparen{\frac{d \logv{\kappa(Z)}}{\log d}}\) on \(Z = \conv{\inbraces{z_1,\dots,z_n}}\).
\end{mainthm}

\subsection{Related work and open questions}
\label{sec:rw}

\paragraph{Streaming asymmetric ellipsoidal roundings.}

To our knowledge, the first paper to study ellipsoidal roundings in the streaming model is that of \citet{mukhopadhyayapproximate}. The authors consider the case where $d=2$ and prove that the approximation factor of the greedy algorithm (that which updates the ellipsoid to be the minimum volume ellipsoid containing the new point and the previous iterate) can be unbounded. Subsequent work by  \citet{mukhopadhyay2010approximate} generalizes this result to all $d \ge 2$.

\paragraph{Nearly-optimal streaming symmetric ellipsoidal roundings.} 
Recently, \citet{mmo22}, and \citet{woodruff2022high} gave the first positive results for streaming ellipsoidal roundings. Both \cite{mmo22} and \cite{woodruff2022high} considered the problem only in the \textit{symmetric setting} -- when the goal is to approximate the polytope $\conv{\inbraces{\pm \vz_1,\dots, \pm \vz_n}}$.  \cite{mmo22} and \cite{woodruff2022high} obtained $O(\sqrt{d\log \kappa(Z)})$ and $O(\sqrt{d\log n\kappa^{\mathsf{OL}}})$-approximations, respectively (here, $\kappa^{\mathsf{OL}}$ is the online condition number; see  \cite{woodruff2022high} for details). Their algorithms use only $\widetilde{O}(\mathsf{poly}(d))$ space, where the $\widetilde{O}$ suppresses $\log d$, $\log n$, and aspect ratio-like terms. 
Note that by John's theorem, the $\Omega(\sqrt{d})$ dependence is required in the symmetric setting even for offline algorithms.

A natural question is whether the techniques of \cite{mmo22} or \cite{woodruff2022high} extend to Problems \ref{prob:main} and~\ref{prob:coreset}. The update rule used in \cite{mmo22} essentially updates $\cE_{t+1}$ to be the minimum volume ellipsoid covering both $\cE_t$ and points $\pm\vz_{t+1}$. In the non-symmetric case, it would be natural to consider the minimum volume ellipsoid covering $\cE_t$ and point $\vz_{t+1}$. However, this approach does not give an $\tilde O(d)$ approximation.
The algorithm in \cite{woodruff2022high} maintains a quadratic form that consists of sums of outer products of ``important points'' (technically speaking, those with a constant online leverage score). Unfortunately, this approach does not suggest how to move the previous center $\vc_{t-1}$ to a new center $\vc_{t}$ in a way that allows the algorithm to maintain a good approximation factor. It is not hard to see that there exist example streams for which the center $\vc_{t-1}$ must be shifted in each iteration to maintain even a bounded approximation factor. This means that any nontrivial solution to Problems \ref{prob:main} and \ref{prob:coreset} must overcome this difficulty.

\paragraph{Offline ellipsoidal roundings for general convex polytopes.} \citet{nesterov2008rounding} gives an efficient \textit{offline} $O(d)$-approximation algorithm for the ellipsoidal rounding problem, with a runtime of $\widetilde{O}(nd^2)$. Observe that this is essentially the same runtime as those achieved by the algorithms we give (see Theorems \ref{thm:main_one} and \ref{thm:main_two}).

\paragraph{Streaming convex hull approximations.}  \citet{agarwal2010streaming} studied related problems of computing \textit{extent measures} of a convex hull in the streaming model, in particular finding coresets for the minimum enclosing ball, and obtained both positive and negative results.
\citet{blum2017approximate} showed that one cannot maintain an \textit{$\eps$-hull} in space proportional to the number of vertices belonging to the offline optimal solution (where a body $\widehat{Z}$ is an $\eps$-hull for $Z$ if every point in $\widehat{Z}$ is distance at most $\eps$ away from $Z$).

\paragraph{Coresets for the minimum volume enclosing ellipsoid problem (MVEE).} 
Let $\mathsf{MVEE}(K)$ denote the minimum volume enclosing ellipsoid for a convex body $K \subset \R^d$.
We say that a subset $S \subseteq [n]$ is an $\eps$-coreset for the MVEE problem if we have
\begin{align}
    \vol\inparen{\mathsf{MVEE}(Z)} \le \inparen{1+\eps}^d\vol\inparen{\mathsf{MVEE}(Z\vert_S)}\label{eq:mvee_coreset}.
\end{align}
There is extensive literature on coresets for the MVEE problem, and we refer the reader to papers by \citet{ky05}, \citet{ty07}, \citet{clarkson10}, \citet{bmv23}, and the book by \citet{todd16}.

Importantly, $\mathsf{MVEE}(Z\vert_S)$ may not be a good approximation for $\mathsf{MVEE}(Z)$ (for that reason, some authors refer to coresets satisfying~(\ref{eq:mvee_coreset}) as weak coresets for MVEE). Therefore, even though $\mathsf{MVEE}(Z)$ provides a good ellipsoidal rounding for $Z$, $\mathsf{MVEE}(Z\vert_S)$ generally speaking does not. Please see \cite[page 2]{ty07} and \cite[Section 2.1]{bmv23} for an extended discussion.

\section{Summary of Techniques}

In this section, we give an overview of the technical methods behind our results.

\subsection{Monotone algorithms}
\label{sec:monotone}
The algorithm we give in \Cref{thm:main_one} is of a certain class of \textit{monotone} algorithms, which we now define.

\begin{maindefn}[Monotone algorithm]
\label{def:alg-invariant}
Consider the setting of Problem \ref{prob:main}. Note the following invariants for every timestep \(t\).
\begin{align}
    \vc_t + \cE_t 
     &\supseteq 
     \conv{(\vc_{t-1} + \cE_{t-1}) \cup \{\vz_t\}}     
\label{item:def-invariant-1} \\
    \vc_t + \alpha_t \cE_t &\subseteq \conv{(\vc_{t-1} + \alpha_{t-1} \cdot \cE_{t-1}) \cup \{\vz_t\}} \label{item:def-invariant-2}
\end{align}

We say that an algorithm \(\cA\) is \textnormal{monotone} if for any initial \((\vc_0 + \cE_0, \alpha_0)\) and sequence of data points \(\vz_1, \ldots, \vz_n\), the resulting sequence \(\{(\vc_0 + \cE_0, \alpha_0), (\vc_1 + \cE_1, \alpha_1), \ldots, (\vc_n + \cE_n, \alpha_n)\}\) arising from applying \(\cA\) to the stream satisfies the two invariants \eqref{item:def-invariant-1} and \eqref{item:def-invariant-2}. Refer to Figure \ref{fig:update_step}.

We will sometimes consider how a monotone algorithm \(\cA\) makes a single update upon seeing a new point $\vx$. In this setting, we will call \(\cA\) a \textit{monotone update rule}.
\end{maindefn}

\begin{figure}[h]
\centering
\includegraphics[height=6cm]{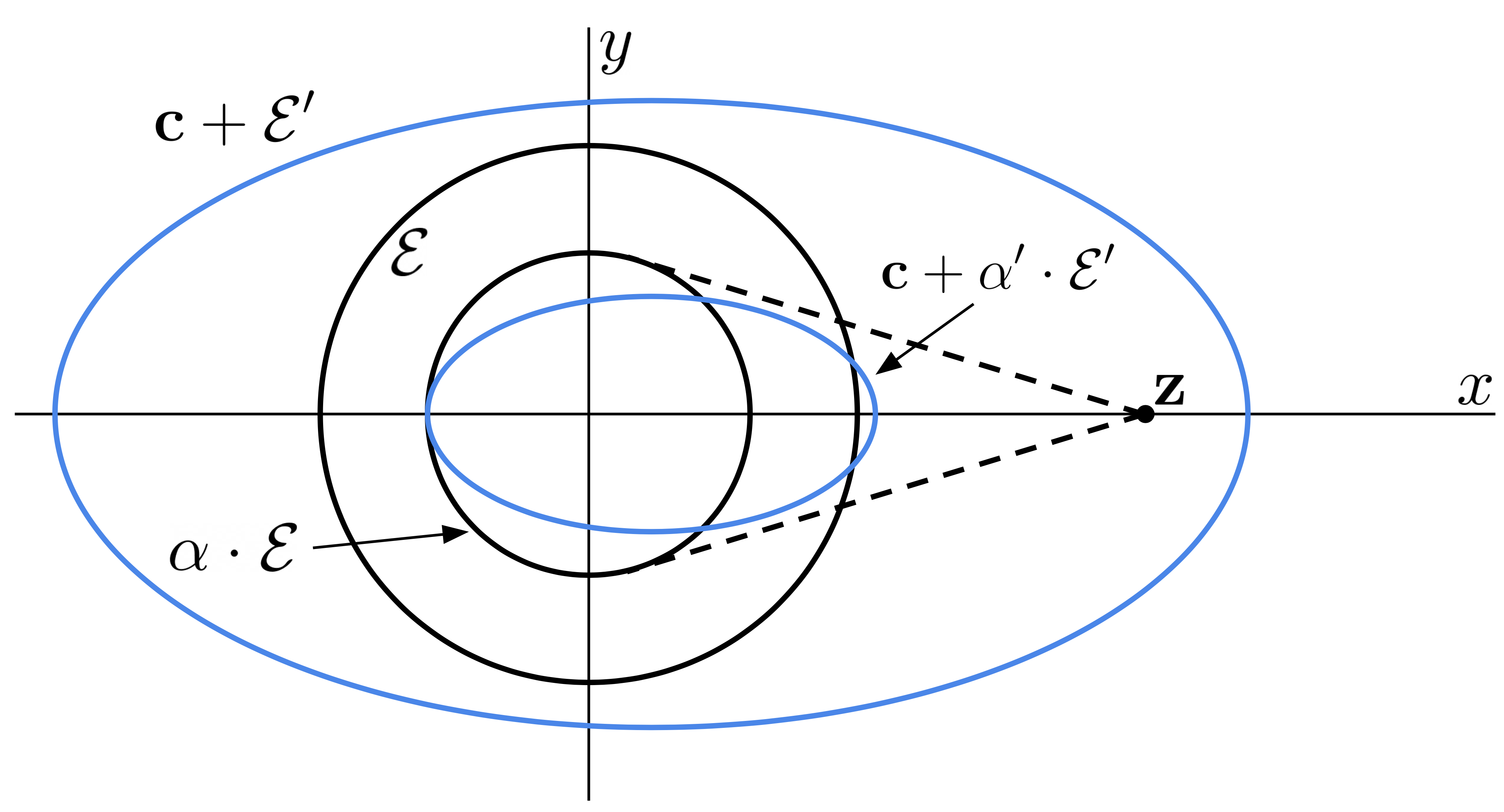}
\caption{A monotone update step. For brevity, we refer to \(\cE\) and \(\alpha \cdot \cE\) as the previous ellipsoids \(\cE_{t-1}, \alpha \cE_{t-1}\), and \(\cE'\) and \(\alpha' \cdot \cE'\) as the next ellipsoids \(\cE_{t}, \alpha' \cdot \cE_t\). \(\cE\) and \(\alpha \cE\) are, respectively, the larger and smaller black circles. 
\(c + \cE'\) and \(c + \alpha' \cE'\) are the larger and smaller blue ellipses. The dotted lines show \(\partial(\conv{\alpha \cE \cup \{\vz\}}) \setminus \partial(\alpha \cE)\), i.e. the the boundary of \(\conv{\alpha \cdot \cE \cup \{\vz\}}\) minus the boundary of \(\alpha \cE\). }
\label{fig:update_step}
\end{figure}

Here we will refer to \(\vc_t + \cE_t, \vc + \alpha_t \cE_t\) as the `next' ellipsoids and 
to \(\vc_{t-1} + \cE_{t-1}, \vc + \alpha_{t-1} \cE_{t-1}\) as the `previous' ellispoids.
The first condition we require is that 
\begin{equation}\vc_{t} + \cE_{t} \supseteq \vc_{t-1}+\cE_{t-1}.\label{item:def-invariant-1-alt}
\tag{\ref*{item:def-invariant-1}a}
\end{equation}
It ensures that each successive outer ellipsoid contains the previous outer ellipsoid. Thus once the algorithm decides  that some $\vz \in \vc_t + \cE_t$, it makes a commitment that $\vz \in \vc_{t'} + \cE_{t'}$ for all $t'\geq t$.
Note that (\ref{item:def-invariant-1-alt}) implies (\ref{item:def-invariant-1}), since $\vz_t$ must be in $\vc_t + \cE_t$ and $\vc_t + \cE_t$ is convex.
The second condition (\ref{item:def-invariant-2}) looks more complex but is also very natural. Assume that the algorithm only knows that (a) $\vc_{t-1} + \alpha_{t-1}\cE_{t-1}\subseteq Z$ (this is true from induction) and (b) $\vz_t\in Z$ (this is true by the definition of $Z$). Then it can only be certain that $A = \conv{(\vc_{t-1} + \alpha_{t-1} \cdot \cE_{t-1}) \cup \{\vz_t\}}$ lies in $Z$; as far as the algorithm is concerned, any point outside of $A$ may also be outside of $Z$.
Since the algorithm must ensure that $\vc_t + \alpha_t {\cal E}_t \subseteq Z$, it will also ensure that $\vc_t + \alpha_t {\cal E}_t\subseteq A$ and thus satisfy (\ref{item:def-invariant-2}).

\subsection{Streaming ellipsoidal rounding (Theorems \ref{thm:main_one}, \ref{thm:main_two}, and \ref{thm:main_two_ints})}

Now we describe the algorithm from Theorem~\ref{thm:main_one} in more detail. 
Our algorithm keeps track of the current ellipsoid $\cE_t$, center $\vc_t$, and scaling parameter $\alpha_t$. Initially, $\vc_0 + \cE_0$ is the ball of radius $r_0$ around $\vc_0$ ($r_0$ and $\vc_0$ are given to the algorithm),  and $\alpha_0=1$. Each time the algorithm gets a new point $\vz_t$, it updates $\cE_{t-1}$, $\vc_{t-1}$, $\alpha_{t-1}$ using a \textit{monotone} update rule (as defined in Definition~\ref{def:alg-invariant}) and obtains $\cE_{t}$, $\vc_{t}$, $\alpha_{t}$. 
The monotonicity condition is sufficient to guarantee that the algorithm gets a $1/\alpha_n$ approximation to $Z$.
Indeed, first using condition~(\ref{item:def-invariant-1}), we get
$$\vc_n + \cE_n \supseteq (\vc_{n-1} + \cE_{n-1}) \cup\{\vz_n\} \supseteq (\vc_{n-2} + \cE_{n-2}) \cup\{\vz_{n-1},\vz_n\}\supseteq \dots \supseteq \{\vz_1,\dots,\vz_n\}.$$
Thus, $\vc_n + \cE_n \supseteq Z$.
Then, using condition~(\ref{item:def-invariant-2}), we get
\begin{align*}
\vc_n + \alpha_n \cE_n &\subseteq
\conv{(\vc_{n-1} + \alpha_{n-1} \cE_{n-1}} \cup \{\vz_n\}) \subseteq \conv{(\vc_{n-2} + \alpha_{n-2} \cE_{n-2}} \cup \{\vz_{n-1},\vz_n\}) \\
&\subseteq \dots \subseteq \conv{(\vc_0 + \alpha_0 \cE_0) \cup \{\vz_1,\dots,\vz_n\}}.
\end{align*}
The initial ellipsoid $\vc_0 + \alpha_0 \cE_0 = \vc_0 + r_0 B_2^d$ is in $Z$ and therefore $\vc_n + \alpha_n \cE_n\subseteq \conv{\vz_1,\dots, \vz_n} = Z$.
We verified that the algorithm finds a $\nicefrac{1}{\alpha_n}$ approximation for $Z$.

Now, the main challenge is to design an update rule that ensures that $1/\alpha_n$ is small (as in the statement Theorem~\ref{thm:main_one}) and prove that the rule satisfies the monotonicity conditions/invariants from Definition~\ref{def:alg-invariant}.
We proceed as follows.

First, we design a monotone update rule that satisfies a particular evolution condition. This condition upper bounds the increase of the approximation factor $\nicefrac{1}{\alpha_t} - \nicefrac{1}{\alpha_{t-1}}$. Second, we prove that any monotone update rule satisfying the evolution condition yields the approximation we desire. These two parts imply Theorem \ref{thm:main_one}. Finally, we remove the initialization requirement from Theorem \ref{thm:main_one} and obtain Theorem \ref{thm:main_two}. 

\paragraph{Designing a monotone update rule.} Suppose that at the end of timestep $t-1$ our solution consists of a center $\vc_{t-1}$, ellipsoid $\cE_{t-1}$, and scaling parameter $\alpha_{t-1}$ for which the invariants in Definition \ref{def:alg-invariant} hold. We give a procedure that, given the next point $\vz_{t}$, computes $\vc_{t}, \cE_{t}, \alpha_{t}$ that still satisfy the invariants of Definition \ref{def:alg-invariant}. Further, we prove that the resulting update satisfies an evolution condition \eqref{eq:overview_evolution}:
\begin{align}
    \frac{\nfrac{1}{\alpha_{t}} - \nfrac{1}{\alpha_{t-1}}}{\log\mathsf{vol}(\cE_{t}) - \log\mathsf{vol}(\cE_{t-1})} \le C\label{eq:overview_evolution}
\end{align}
where $C$ is an absolute constant; $\vol{\cE}$ denotes the volume of ellipsoid $\cE$. 
While it is possible to find the optimal update using convex optimization (the update that satisfies the invariants and minimizes the ratio on the left of (\ref{eq:overview_evolution})), we instead provide an explicit formula for an update that readily satisfies (\ref{eq:overview_evolution}) and as we show is monotone. 

Now we describe how we get the formula for the update rule.
By applying an affine transformation, we may assume that ${\cal E}_{t-1}$ is a unit ball and $\vc_{t-1} = 0$. Further, we may assume that $\vz_t$ is colinear with $\ve_1$ (the first basis vector): $\vz_t = \|\vz_t\| \ve_1$. Importantly, affine transformations preserve (a) the invariants in Definition~\ref{def:alg-invariant} (if they hold for the original ellipsoids and points, then they also do for the transformed ones and vice versa) and (b) the value of the ratio in (\ref{eq:overview_evolution}), since they preserve the value of $\vol({\cE_t})/\vol({\cE_{t-1}})$.

Now consider the group $G = \so{d}_{\ve_1}\cong \so{d-1}$ of orthogonal transformations that map $\ve_1$ to itself: all of them map the unit ball ${\cal E}_{t-1}$ to itself and $\vz_{t}$ to itself. Thus, it is natural to search for an update $(\vc_t, {\cal E}_t)$ that is symmetric with respect to all these transformations.
It is easy to see that in this case ${\cal E}_t$ is defined by equation 
$(x_1/a)^2 + \sum_{i=2}^d (x_i/b)^2 =1$ where $a$ and $b$ are some parameters (equal to the semiaxes of ${\cal E}_{t}$) and $\vc_t = c \ve_1$ for some $c$. 
Since all ellipsoids and points appearing in the invariant conditions are symmetric w.r.t.\ $G$, it is sufficient now to restrict our attention to their sections by $2$d-plane $\lspan{\ve_1,\ve_2}$ and prove that the invariants hold in this plane. Hence, the problem reduces to a statement in two-dimensional Euclidean geometry (however, when we analyze~(\ref{eq:overview_evolution}), we still use that the volume of ${\cal E}_t$ is proportional to $ab^{d-1}$ and not $ab$).

Let us denote the coordinates corresponding to basis vectors $\ve_1$ and $\ve_2$ by $x$ and $y$. For brevity, let ${\cal E} = {\cal E}_{t-1}$, $\vz = \vz_t$, ${\cal E}' = {\cal E}_{t}$, $\vc = \vc_t =c\ve_1$, $\alpha = \alpha_{t-1}$, and $\alpha' = \alpha_t$.
We now need to choose parameters $a$, $b$, and $c$ so that invariants from Definition \ref{def:alg-invariant} and equation (\ref{eq:overview_evolution}) hold. See \Cref{fig:update_step}.
As shown in that figure, the new outer ellipse \(\vc + \cE'\) must contain the previous outer ellipse \(\cE\) and the newly received point \(\vz\). The new inner ellipse \(\vc + \alpha' \cE'\) must be contained within the convex hull of the previous inner ellipse \(\alpha \cE\) and \(\vz\).

It is instructive to consider what happens when point $\vz$ is at infinitesimal distance $\Delta$ from $\cE$: $\|\vz\| = 1 +\Delta$. We consider a minimal axis-parallel outer ellipse $\cE'$ that contains $\cE$ and $\vz$. It must go through $\vz = (1+\Delta,0)$ and touch $\cE$ at two points symmetric w.r.t.\ the $x$-axis, say, $(-\sin \varphi, \pm \cos \varphi)$. Angle $\varphi$ uniquely determines $\cE'$. Now we want to find the largest value of the scaling parameter $\alpha'$ so that $\alpha' \cE'$ fits inside the convex hull of $\cE$ and $\vz$. When $\Delta$ is infinitesimal, this condition splits into two lower bounds on $\alpha'$ --  loosely speaking, they say that $\cE$ does not extend out beyond the convex hull in the horizontal (one bound) and vertical directions (the other). The former bound becomes stronger (gives a smaller upper bound on $\alpha'$) when $\varphi$ increases, and the latter becomes stronger when $\varphi$ decreases. When $\varphi = \alpha/2 \pm O(\alpha^2)$, then all terms linear in $\alpha$ vanish in both bounds and then $\alpha' = \alpha - \Theta(\alpha^2\Delta)$ satisfies both of them; for other choices of $\varphi$, we have $\alpha' \leq \alpha - \Omega(\alpha\Delta)$. 
So we let $\varphi=\alpha/2$ and from the formula for $\alpha'$ get $1/\alpha' = 1/\alpha + O(\Delta)$. On the other hand, $\vol(\cE') \geq (1 + \Delta/2)\vol(\cE)$, since $\cE'$ covers $\vz = (1 + \Delta,0)$. It is easy to see now that the evolution condition (\ref{eq:overview_evolution}) holds: the numerator is $O(\Delta)$ and the denominator is $\Omega(\Delta)$ in (\ref{eq:overview_evolution}).

We remark that letting $\vc + {\cal E}'$ be the minimum volume ellipsoid that contains $\cal E$ and $\vz$ is a highly suboptimal choice (it corresponds to setting $\varphi=\Theta(1/d)$). To derive our specific update formulas for arbitrary $\vz$, we, loosely speaking,  represent an arbitrary update as a series of infinitesimal updates, get a differential equation on $a$, $b$, $c$, and $\alpha'$, solve it, and then simplify the solution (remove non-essential terms etc).  We get the following.

Our updates come from a family parameterized by \(\gamma \geq 0\). 
 Define \(\alpha'\) by \(\nfrac{1}{\alpha'} = \nfrac{1}{\alpha} + 2\gamma\). With this choice of $\alpha'$, define the new ellipses to be
\[\underbrace{\frac{1}{a^2}(x-c)^2 + \frac{1}{b^2} y^2 = 1}_{\vc + \cE'}, \qquad \underbrace{\frac{1}{a^2}(x-c)^2 + \frac{1}{b^2} y^2 = \alpha'^2}_{\vc + \alpha' \cE'}\]
where we use parameters
\begin{equation*}
\left.\begin{aligned}
a &= \expv{\gamma} \\
b &=1 + \frac{\alpha - \alpha'}{2} \\
c &= - \alpha + \alpha' \cdot a
\end{aligned}\qquad\right\}.
\end{equation*}
Choose $\gamma\approx \ln \|\vz\|$ so that $\vc + \cE'$ covers point $\vz$.
We use two-dimensional geometry to prove that $\cE'$, $\vc$, and $\alpha'$ satisfy the invariants (see Figure~\ref{fig:update_step}).
Now to prove the evolution condition, we observe two key properties: (1) the increase in the approximation factor is given by \(\frac{1}{\alpha'} - \frac{1}{\alpha} = 2 \gamma\) and (2) the length of the horizontal semiaxis of the new outer ellipse is \(\exp(\gamma)\). The length of the vertical semiaxis is at least $1$, so by the second property we have \(\log \vol(\cE') - \log \vol(\cE) \geq \gamma\). We combine this with the first property to prove that this update satisfies the evolution condition \eqref{eq:overview_evolution}.

Finally, we obtain an upper bound on $1/\alpha_n$ from the evolution equation. We have
\begin{align*}
    \nfrac{1}{\alpha_{n}} = \nfrac{1}{\alpha_0} + \sum_{t=1}^n
\left(\nfrac{1}{\alpha_{t}} - \nfrac{1}{\alpha_{t-1}}\right) \stackrel{\tiny(\text{by }\ref{eq:overview_evolution})}{\leq} 1+ C\sum_{t=1}^n(\log\mathsf{vol}(\cE_{t}) - \log\mathsf{vol}(\cE_{t-1})) = 1 + C\log \frac{\vol{\cE_n}}{\vol{\cE_0}}.
\end{align*}
It remains to get an upper bound on $\vol(\cE_n)$. We know that $\cE_n$ approximates $Z$, and $Z$, in turn, is contained in the ball of radius $R(Z)$. Loosely speaking, we get $\vol(\cE_n) \approx \vol(Z) \leq R(Z)^d \vol(B_2^d)$. Since $\cE_0$ is the ball of radius $r$, $\vol{\cE_0}= r^d \vol(B_2^d)$. We conclude that 
the approximation factor is at most $\nfrac{1}{\alpha_{n}} \lessapprox 1 + C\log \frac{R(Z)^d}{r^d} = 1 + O(d\log \frac{R(Z)}{r})$, as desired.

\paragraph{Removing the initialization assumption.} Once we have a monotone update rule and  guarantee on its approximation factor, we have to convert this to a guarantee where the algorithm does not have access to the initialization. 

One natural approach is as follows. Let $d' \le d$ be the largest timestep for which points $\vz_1, \dots,\vz_{d'+1}$ are in general position. We can compute the John ellipsoid for $\conv{\inbraces{\vz_1,\dots,\vz_{d'+1}}}$ and after that apply the monotone update rule guaranteed by Theorem \ref{thm:main_one} to obtain the rounding for every $t \ge d'+2$, so long as for every such timestep we have $\vz_t  \in \vspan{\vz_1,\dots,\vz_{t-1}}$.

The principal difficulty in this approach is designing an \textit{irregular update step} that will handle points $\vz_t$ outside of $\vspan{\vz_1,\dots,\vz_{t-1}}$; when we add these points the dimensionality of the affine hull increases by 1. We consider the special case where the new point $\vz_t$ is conveniently located with respect to our previous ellipsoid $\cE_{t-1}$ (see \Cref{fig:update_step_irreg} for a 2d-picture): $\cE_{t-1}$ is the unit ball in $\lspan{\ve_1,\dots,\ve_{d'}}$; point $z_t = (0,\dots, 0, \sqrt{1+2\alpha}),0,\dots)$, here only coordinate $d'+1$ is non-zero. We show that we can design an irregular update step for this special case that makes the new approximation factor $\nfrac{1}{\alpha_{t}}$ satisfy $\nfrac{1}{\alpha_t} = \nfrac{1}{\alpha_{t-1}}+1$.

\begin{figure}[h]
\centering
\includegraphics[height=6cm]{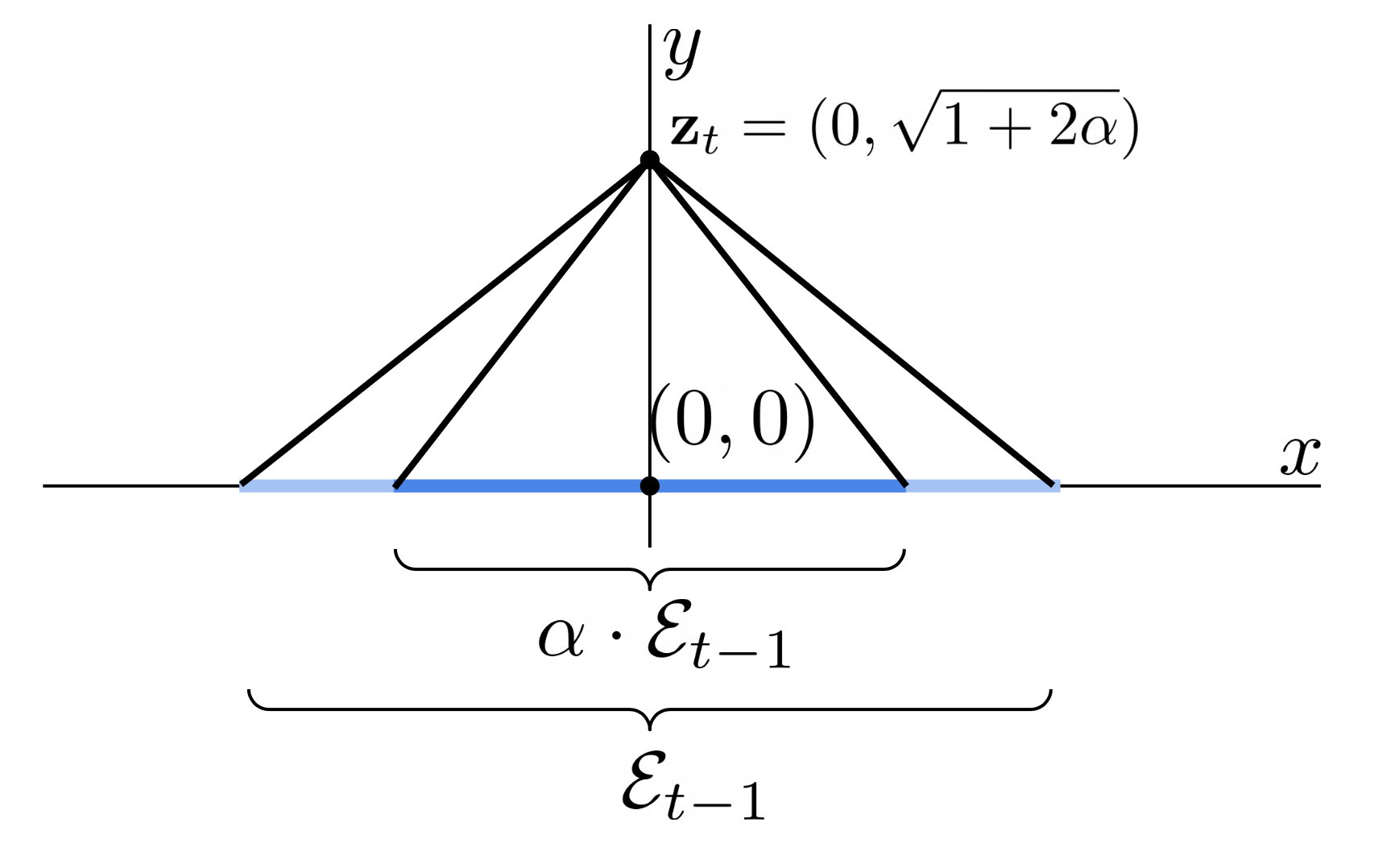}
\caption{Irregular update step. \(\cE_{t-1}\) and \(\alpha \cdot \cE_{t-1}\) are, respectively, the light blue strip on the \(x\)-axis and the dark blue strip on the \(x\)-axis. \(\vz_t = (0, \sqrt{1+2 \alpha})\) is the newly received point.
}
\label{fig:update_step_irreg}
\end{figure}

It turns out that it is sufficient to consider only this special case. To see this, note that we can choose an affine transformation that maps any new point $\vz_{t}$ and previous ellipsoid $\cE_{t-1}$ to the setting shown in \Cref{fig:update_step_irreg}.
Observe that there are at most $d-1$ irregular update steps. This means that the irregular update steps contribute at most an additive $d-1$ to the final approximation factor.

Finally, observe that the inradius of $\conv{\inbraces{\vz_1,\dots,\vz_t}}$ is not monotone in $t$. In particular, it can decrease after each irregular update step. Nonetheless, we can still give a bound on the radius of a ball that our convex body $\conv{\vz_1,\dots,\vz_t}$ contains for all $t$. This will give us everything we need to apply \Cref{thm:main_one} to this setting, and \Cref{thm:main_two} follows.

\paragraph{Improved bounds on lattices.} 
Finally, we briefly discuss how to remove the aspect ratio dependence in the setting where the input points $\vz_t$ have coordinates in $[-N,N]$. At a high level, this improvement follows from carefully tracking how the approximation factors of our solutions change after an irregular update step. Following \eqref{eq:overview_evolution}, recall that our goal is to analyze (where we write $\alpha_0 = 1$)
\begin{align*}
    \sum_{t \ge 1} \frac{1}{\alpha_t}-\frac{1}{\alpha_{t-1}}.
\end{align*}
By \eqref{eq:overview_evolution}, we see that for all ``regular'' updates, we have
\begin{align*}
    \frac{1}{\alpha_t}-\frac{1}{\alpha_{t-1}} \lesssim \logv{\frac{\vol_{d_t}\inparen{\cE_t}}{\vol_{d_t}\inparen{\cE_{t-1}}}},
\end{align*}
where $d_t = \mathsf{dim}\inparen{\vspan{\vz_1,\dots,\vz_t}}$. Furthermore, as previously mentioned, in our irregular update step, we get
\begin{align*}
    \frac{1}{\alpha_t}-\frac{1}{\alpha_{t-1}} = 1.
\end{align*}
In order to control the sum of the $\nfrac{1}{\alpha_{t}}-\nfrac{1}{\alpha_{t-1}}$, it remains to bound $\nfrac{\vol_{d_t}(\cE_t)}{\vol_{d_{t-1}}(\cE_{t-1})}$ for an irregular update step $t$. We will then get a telescoping upper bound whose last term is the ratio of the volume of the final ellipsoid to the Euclidean ball in the same affine span.

Similarly to the improvements of \citet{woodruff2022high} in the integer-valued case, it will turn out that we will be interested in the total product of these volume changes. By carefully tracking these, we will get that this product can be expressed as the determinant of a particular integer-valued matrix. Then, since this matrix has integer entries, the magnitude of its determinant must be at least $1$. We then observe that the volume of $\cE_n$ after normalizing by the volume of $\vol(B_2^{d_n})$ must be at most $(N\sqrt{d})^{d_n}$, since the length of any vector in this lattice is at most $N\sqrt{d}$. The desired result then follows.

\subsection{Coresets for convex hull \texorpdfstring{(\Cref{thm:main_three})}{(Theorem~\ref{thm:main_three})}}
\label{sec:overview_coreset}

We now outline our proof strategy for \Cref{thm:main_three}.
Our main task is to design an appropriate selection criterion for every new point -- in other words, we must check whether a new point $\vz_t$ is ``important enough'' to be added to our previous set of points $S_{t-1}$. We then have to show that this selection criterion yields the approximation guarantee promised by Theorem \ref{thm:main_three}.

To design the selection criterion, we run an instance of the algorithm in Theorem \ref{thm:main_two} on the stream. For every new point $\vz_t$, we ask two questions -- ``Does $\vz_t$ result in an irregular update step? Does it cause $\mathsf{vol}(\cE_{t})$ to be much larger than $\mathsf{vol}(\cE_{t-1})$?'' If the answer to any of these questions is affirmative, we add $\vz_t$ to the coreset. The first question is necessary to obtain even a bounded approximation factor (for example, imagine that the final point $\vz_n$ results in an irregular update step, then we must add it). The second question is quite natural, as it ensures that the algorithm adds ``important points'' -- those that necessitate a significant update.


We now observe that at every irregular update step $t_{d'}$ for $d' \le d$ and subsequent timestep $t \ge t_{d'}$ for which there are no irregular update steps in between $t_{d'}$ and $t$, there exists a translation $\vc_{d'}$ (which is the center for $\cE_{d'}$ that the algorithm maintains) and a value $r_{d'}$ for which we know
\begin{align*}
    \vc_{d'} + r_{d'} \cdot \inparen{B_2^d \cap \lspan{\vz_1-\vc_{d'},\dots,\vz_{d'}-\vc_{d'}}} \subseteq \conv{\vz_1,\dots,\vz_{t}} \subseteq \vc_C + R_{t} \cdot B_2^d,
\end{align*}
where $\vc_C$ is the circumcenter of $\conv{\inbraces{\vz_1,\dots,\vz_t}}$. The resulting bound on $\abs{S_t}$ follows easily from the above observation and a simple volume argument.

Finally, we obtain the approximation guarantee from noting that for all \(t\), the output of the algorithm from Theorem \ref{thm:main_two} given the first \(t\) points is the same as running it only on the points selected by \(S_t\).

\subsection{Lower bound (Theorem \ref{thm:main_four})}

Whereas in the upper bound we demonstrated a particular algorithm that satisfies the evolution condition \eqref{eq:overview_evolution}, 
for the lower bound it suffices to show that for any monotone algorithm, there exists an instance of the problem (a sequence of $\vz_1$,\dots, $\vz_n$) where the algorithm must satisfy the ``reverse evolution condition'', i.e.
\begin{align}
    \frac{\nfrac{1}{\alpha_{t}} - \nfrac{1}{\alpha_{t-1}}}{\log\mathsf{vol}(\cE_{t}) - \log\mathsf{vol}(\cE_{t-1})} \ge C\label{eq:lb_overview_evolution}
\end{align}
for some $C > 0$.
In analogy to the argument of the upper bound, showing this reverse evolution condition yields a lower bound of the form \(\frac{1}{\alpha_n} \geq \widetilde{\Omega}\inparen{d \log(\kappa)}\).
Given any monotone algorithm \(\cA\), the instance we use is produced by an adversary that repeatedly feeds \(\cA\) a point that is a constant factor away from the previous ellipsoid.

In order to simplify showing this reverse evolution condition, we use a symmetrization argument.
Specifically, by a particular sequence of Steiner symmetrizations, we see that the optimal response of \(\cA\) can be completely described in two dimensions. Thus, it is sufficient to only show this reverse evolution condition in the two-dimensional case where the previous outer ellipsoid is the unit ball.


This transformed two-dimensional setting is significantly simpler to analyze. Specifically, we can assume that the point given by the adversary is always \(2 \ve_1\). The rest of the argument proceeds by cases, again using two-dimensional Euclidean geometry. On a high level, the constraints placed on the new outer and inner ellipsoid by the monotonicity condition force the update of \(\cA\) to satisfy the reverse evolution condition.

\section{Preliminaries}
\label{section:preliminaries_notation}

\subsection{Notation} 

We denote the standard Euclidean norm of a vector $\vv$ by $\norm{\vv}$ and the Frobenius norm of a matrix $\mA$ by $\fnorm{\mA}$. We denote the singular values of a matrix $\mA \in \R^{d\times d}$ by $\sigma_1(\mA),\dots, \sigma_d(\mA)$. Let $\sigma_{\max}(\mA)$ and $\sigma_{\min}(\mA)$ be the largest and smallest singular values of $\mA$, respectively.
We write \(\diag{a_1, \ldots, a_d}\) to mean the \(d \times d\) diagonal matrix whose diagonal entries are \(a_1, \ldots, a_d\).
We use \(\mS_{++}^d\) to denote the set of \(d \times d\) positive definite matrices.
We use \(\ve_1, \ldots, \ve_d\) for the standard basis in \(\R^d\).

Denote the $\ell_2$-unit ball by $B_2^d = \inbraces{\vx \in \R^d \suchthat \norm{\vx}_2 \le 1}$,
and \(\S^{d-1} = \inbraces{\vx \in \R^d \colon \|\vx\|_2 = 1}\) the unit sphere. We use \(\partial S\) for the boundary of an arbitrary set \(S\). We use natural logarithms unless otherwise specified. 

In this paper, we will work extensively with ellipsoids. We will always assume that all ellipsoids and balls we consider are centered at the origin. We use the following representation of ellipsoids. For a non-singular matrix $\mA\in\R^{d\times d}$, let $\cE_\mA \coloneqq \inbraces{\vx \suchthat \norm{\mA\vx}\le 1}$. In other words, the matrix $\mA$ defines an bijective linear map satisfying $\mA \cE_\mA = B_2^d$. Every full-dimensional ellipsoid (centered at the origin) has such a representation. We note that this representation is not unique as matrices $\mA$ and $\mM\mA$ define the same ellipsoid if matrix $\mM$ is orthogonal (since $\norm{\mA\vv} = \norm{\mM\mA\vv}$ for every vector $\vv$). Sometimes, we will have to consider lower-dimensional ellipsoids within an ambient space of higher dimension; in this case, we will use the notation $\cE \cap H$ where $H$ is some linear or affine subspace -- note that $\cE \cap H$ is also an ellipsoid. 

Now consider the singular value decomposition of  $\mA$: $\mA = \mU\Sigma^{-1} \mV^T$ (it will be convenient for us to write $\Sigma^{-1}$ instead of standard $\Sigma$ in the decomposition). The diagonal entries of $\Sigma$ are exactly the semi-axes of $\cE_\mA$. As mentioned above, matrices $\mU\Sigma^{-1} \mV^T$ and $\mU'\Sigma^{-1} \mV^T$ define the same ellipsoid for any orthogonal \(\mU' \in \mathbb{R}^{d \times d}\); in particular, every ellipsoid can be represented by a matrix of the form $\mA=\Sigma^{-1} \mV^T$.

\subsection{Geometry}
We restate the well-known result that five points determine an ellipse.
This is usually phrased for conics, but for nondegenerate ellipses the usual condition that no three of the five points are collinear is vacuously true.
\begin{claim}[Five points determine an ellipse]\label{claim:conic_5pt}
Let \(\vc_1 + \partial \cE_1, \vc_2 + \partial \cE_2\) be two ellipses in \(\R^2\).
If they intersect at five distinct points, then \(\vc_1 + \partial \cE_1\) and \(\vc_2 + \partial \cE_2\) are the same.
\end{claim}

The following claim, that every full-rank ellipsoid (i.e. an ellipsoid whose span has full dimension) can be represented by a positive definite matrix, follows from looking at the SVD.
\begin{claim}\label{claim:prelim_ellips_sym}
    Let \(\cE \subseteq \R^d\) be a full-rank ellipsoid.
    Then there exists \(\mA \succ 0\) such that \(\cE = \cE_{\mA}\).
\end{claim}

We also have the standard result relating volume and determinants, which follows from observing \(\mA \cE_{\mA} = B_2^d\).
\begin{claim}\label{claim:prelim_vol}
Let \(\mA \succ 0\). Then 
\[\vol(\cE_{\mA}) = \det(\mA^{-1}) \vol(B_2^d)\]
\end{claim}

In order to give the reduction in the lower bound from the general case to the two-dimensional case, we use the technique of Steiner symmetrization (see e.g. \cite[Section 1.1.7]{artstein2015asymptotic}).
Given some unit vector \(\vu \in \R^d\) and convex body \(K \subseteq \R^d\), we write \(S_{\vu}(K)\) for the \textit{Steiner symmetrization} in the direction of \(\vu\).
Recall that the Steiner symmetrization is defined so that for any \(\vx \perp \vu\):
\[\vol((\vx + \R \vu) \cap K) = \vol((\vx + \R \vu) \cap S_{\vu}(K))\]
and so that \((\vx + \R \vu) \cap S_{\vu}(K)\) is an interval centered at \(\vx\).
Note that we will overload notation slightly as we will allow you \(\vu\) to be a vector of any non-zero length while Steiner symmetrization is usually defined with \(\vu\) being a unit vector, but we will simply take \(S_{\vu} = S_{\frac{\vu}{\|\vu\|_2}}\).

Importantly, Steiner symmetrization will preserve important properties of the update.
We have the key facts that \(\vol(S_{\vu}(K)) = \vol(K)\), \(S_{\vu}(K') \subseteq S_{\vu}(K)\) if \(K \subseteq K'\), and further the Steiner symmetrization preserves \(K\) being an ellipsoid:
\begin{claim}[\protect{\cite[Lemma 2]{bourgain2006estimates}}]\label{claim:steiner_ellipsoid}
    If \(c + \cE \subseteq \R^d\) is an ellipsoid, \(S_{\vu}(c + \cE)\) is still an ellipsoid.
\end{claim}

Further, if we apply Steiner symmetrization to a body that is a body of revolution about an axis,
it does not change the body if \(\vu\) is perpendicular to the axis of revolution.
\begin{claim}\label{claim:steiner_sym}
    Let \(K \subseteq \R^d\) be a body of revolution about the \(\ve_1\)-axis.
    Then if \(\vu \perp \ve_1\), \(S_{\vu}(K) = K\).
\end{claim}

\section{Streaming Ellipsoidal Rounding}
\label{sec:main_alg_top}

Our goal in this section is to prove Theorems \ref{thm:main_one} and \ref{thm:main_two}.

\subsection{Monotone algorithms solve \texorpdfstring{\Cref{prob:main}}{Theorem \ref{prob:main}}}

To design algorithms to solve the streaming ellipsoidal rounding problem,
we first show that any monotone algorithm gives a valid solution.
We let \(\vc_0 \in \R^d\) and \(r_0 \geq 0\) be given so that
\(\vc_0 + r_0 \cdot B_2^d \subseteq Z\),
and denote the initial ellipsoid as \(\cE_0 = r_0 \cdot B_2^d\).
Note that \(r_0\) need not be the inradius, although it is upper bounded by the inradius.

If we had for each intermediate step \(t\)
that \(\vc_t + \alpha_t \cdot \cE_t \subseteq \conv{\vz_1, \ldots \vz_t} \subseteq \vc_t + \cE_t\),
then clearly any algorithm that satisfies this would give a valid final solution as well.
However, in intermediate steps it is not clear that 
\(\vc_t + \alpha_t \cdot \cE_t \subseteq \conv{\vz_1, \ldots \vz_t}\),
due to the initialization of \(\vc_0 + \cE_0\) in our monotone algorithm framework.
Instead, we relax this invariant to \(\vc_t + \alpha_t \cdot \cE_t \subseteq \conv{\{\vz_1, \ldots \vz_t\} \cup (\vc_0 + \cE_0)}\),
which still suffices to produce a valid final solution.
\begin{claim}\label{claim:alg_invariant}
To solve \Cref{prob:main}, it suffices for 
the sequence of ellipsoids \(\vc_i + \cE_i\) and scalings \(\alpha_i\) to satisfy the invariants of \Cref{def:alg-invariant}.
\end{claim}

\begin{proof}
First, we argue that \(\conv{\vz_1, \ldots, \vz_n} \subseteq \vc_n + \cE_n\).
As \(\cE_n\) is an ellipsoid and therefore a convex set, it suffices to show \(\{\vz_1, \ldots, \vz_n\} \subseteq \vc_n + \cE_n\).
We actually argue by induction that \(\{\vz_1, \ldots, \vz_t\} \subseteq \vc_t + \cE_t\) for all \(0 \leq t \leq n\).
This is vacuously true for \(t = 0\). At each step \(t > 0\) the inductive hypothesis gives \(\{\vz_1, \ldots, \vz_{t-1}\} \subseteq \vc_{t-1} + \cE_{t-1}\), and thus by (\ref{item:def-invariant-1}) we have \(\{\vz_1, \ldots, \vz_{t}\} \subseteq \vc_t + \cE_t\).

Now, we argue that \(\vc_n + \alpha_n \cdot \cE_n \subseteq \conv{\vz_1, \ldots, \vz_n}\).
We show by induction that \(\vc_t + \alpha_t \cdot \cE_t \subseteq \conv{\{\vz_1, \ldots, \vz_t\} \cup (\vc_0 + \cE_0)}\)
for all \(0 \leq t \leq n\).
This is sufficient as \(\conv{\{\vz_1, \ldots, \vz_n\} \cup (\vc_0 + \cE_0)} = Z\).
The case for \(t = 0\) is trivial.
For \(t > 0\), the inductive hypothesis gives \(\vc_{t-1} + \alpha_{t-1} \cdot \cE_{t-1} \subseteq \conv{\{\vz_1, \ldots, \vz_{t-1}\} \cup (\vc_0+\cE_0)}\),
and by (\ref{item:def-invariant-2}) we have
\[c_t + \alpha_t \cdot \cE_t \subseteq \conv{(c_{t-1} + \alpha_{t-1} \cdot \cE_{t-1}) \cup \{\vz_i\}} \subseteq \conv{\{\vz_1, \ldots, \vz_t\} \cup (\vc_0 + \cE_0)}\]
as desired.
\end{proof}

\subsection{Special case}\label{sec:two_dim_update}

In light of \Cref{claim:alg_invariant}, our strategy is to design an algorithm that preserves the invariants given in \Cref{def:alg-invariant}. This algorithm can be thought of as an \textit{update rule} that,
given the previous outer and inner ellipsoids \(\vc_{t-1} + \cE_{t-1}, \vc_{t-1} + \alpha_{t-1} \cE_{t-1}\) and next point \(\vz_t\), produces the next outer and inner ellipsoids \(\vc_t + \cE_t, \vc_t + \alpha_t \cE_t\).

It is in fact sufficient to consider the simplified case where the previous outer ellipsoid is the unit ball, and the previous inner ellipsoid is some scaling of the unit ball;
we will show this in \Cref{sec:gen_high_dim}. 
We can further specialize by considering only the two-dimensional case \(d=2\).
We will later show that the high-dimensional case is not much different, as all the relevant sets \(\vc_{t_1} + \cE_{t-1}, \vc_t + \cE_{t}\) and \(\conv{\alpha \cdot \cE_{t-1} \cup \{\vz_t\}}\) form bodies of revolution about the axis through \(\vc_{t-1}\) and \(\vz_t\).

We now describe our two-dimensional update rule.
In order to simplify notation, we will let \(\alpha\) be the previous scaling \(\alpha_{t-1}\), and \(\alpha'\)
be the next scaling \(\alpha_t\). 
We will assume that \(\alpha \leq \nfrac{1}{2}\) to simplify the analysis
of our update rule; this will not affect the quality of our final approximation as this update rule
will only be used in `large approximation factor' regime.
We will also overload notation; writing \(c + \cE\) even when \(c\)
is a scalar to mean \((c, 0) + \cE\).
We can describe the previous outer ellipsoid \(\cE\) with the equation \(x^2 + y^2 \leq 1\), and the previous inner ellipsoid \(\alpha \cE\) with \(x^2 + y^2 \leq \alpha^2\).
We define the next outer and inner ellipsoids \(c + \cE'\), \(c + \alpha' \cE'\) as
\[\underbrace{\frac{1}{a^2}(x-c)^2 + \frac{1}{b^2} y^2 \leq 1}_{c + \cE'}, \qquad \underbrace{\frac{1}{a^2}(x-c)^2 + \frac{1}{b^2} y^2 \leq \alpha'^2}_{c + \alpha' \cE'}\]
where we use parameters
\begin{equation}\label{eqn:update_params}
\left.\begin{aligned}
a &= \expv{\gamma} \\
b &=1 + \frac{\alpha - \alpha'}{2} \\
c &= - \alpha + \alpha' \cdot a \\
\alpha'& = \frac{1}{\frac{1}{\alpha} + 2\gamma} \\
\end{aligned}\qquad\right\}
\end{equation}

We will let \(\vz\) be the rightmost point of \(c + \cE'\), so that \(\vz = (c+ a, 0)\). Eventually, we will choose \(\gamma\) so that \(\vz\) coincides with \(\vz_{t}\), the point received in the next iteration.
In \Cref{sec:gen_alg}, these parameters \(a(\gamma), b(\gamma), c(\gamma), \alpha'(\gamma)\)
will be used as functions of the parameter \(\gamma \geq 0\). However, we will not yet explicitly specify \(\gamma\), so in this section these parameters can be thought of as constants for some fixed \(\gamma\). 
This update rule is pictured in \Cref{fig:update_step}.

We first collect a few straightforward properties of this update rule.
\begin{claim}\label{claim:update_params}
The parameters in the setup (\ref{eqn:update_params}) satisfy
\begin{enumerate}
    \item\label{item:claim_update_params_1} \(\frac{1}{\alpha'} = \frac{1}{\alpha} + 2\gamma\)
    \item\label{item:claim_update_params_2} \(b \geq 1\)
    \item\label{item:claim_update_params_3} \(c \geq 0\)
    \item\label{item:claim_update_params_4} \(c + \alpha' \cdot a \geq \alpha\)
\end{enumerate}
\end{claim}
Before proving these properties, we provide geometric interpretations.
Intuitively, (\ref{item:claim_update_params_1}) means that \(\gamma\) is proportional to the increase in the approximation factor at this step, a fact that we will use when analyzing the general-case algorithm.
(\ref{item:claim_update_params_2}) means that the outer ellipsoid grows on every axis;
and (\ref{item:claim_update_params_3}) means that the centers of the next ellipsoids are to the right of the \(y\)-axis, i.e. the centers of the next ellipsoids are further towards \(\vv\) than those of the previous ellipsoids.
The rightmost point of \(c + \alpha' \cE'\) is \(c + \alpha' \cdot a\), so (\ref{item:claim_update_params_4}) shows that this point is to the right of the rightmost point of \(\alpha \cdot \cE\).
\begin{proof}
(\ref{item:claim_update_params_1}) is clear from rearranging the definition of \(\alpha'\).
From (\ref{item:claim_update_params_1}) we also have \(\alpha' \leq \alpha\), so that (\ref{item:claim_update_params_2}) follows immediately.

For (\ref{item:claim_update_params_3}), observe that \(\frac{\alpha}{\alpha'} = 1 + 2 \gamma \alpha\).
When \(\alpha \leq \nfrac{1}{2}\), this means \begin{equation}\label{eqn:pup} 
\frac{\alpha}{\alpha'} \leq 1 + \gamma \leq \expv{\gamma} = a
\end{equation}
using \(1 + x \leq e^x\), \Cref{claim:wk_e}-(\ref{item:wk_1}).
By definition of \(c\), \(\alpha/\alpha' \leq a\) is equivalent to \(c \geq 0\).

To show (\ref{item:claim_update_params_4}), by definition we have that \(c + \alpha' \cdot a = - \alpha + 2 \alpha' a\).
Thus showing \(c + \alpha' \cdot a \geq \alpha\) is equivalent to showing that \(\alpha' a \geq \alpha\), which is equivalent to the inequality in (\ref{eqn:pup}).
\end{proof}

As \Cref{fig:update_step} depicts, the update step we defined satisfies the invariants in \Cref{def:alg-invariant} and so is monotone;
in the rest of this section we make this picture formal.
To start, we consider the invariant concerning outer ellipsoids; we will show that \(\cE \subseteq c + \cE'\).
For now we can think of \(\vz\) as replacing \(\vz_t\)
, and clearly \(\vz \in c + \cE'\), so if we show that \(\cE \subseteq c + \cE'\), then \(\conv{\cE \cup \{\vz\}} \subseteq c + \cE'\) as well since \(c + \cE'\) is convex.

\begin{claim}\label{claim:update_step_outer}
    \(\cE \subseteq c + \cE'\)
\end{claim}
\begin{proof}
First, observe that \(\cE \subseteq \cE'\) because both axes of \(\cE'\) have greater length than those of \(\cE\): \(a \geq 1\) by definition, and \(b \geq 1\) from \Cref{claim:update_params}-(\ref{item:claim_update_params_2}).
Now, we translate \(\cE'\) to the right until it touches \(\cE\) at two points. 
We call this translated ellipse \(c_r + \cE'\), as shown in \Cref{fig:update_step_outer}.
Observe that as long as \(c \leq c_r\), we have \(\cE \subseteq c + \cE'\).
We now determine \(c_r\).

\begin{figure}[ht]
\centering
\includegraphics[height=6cm]{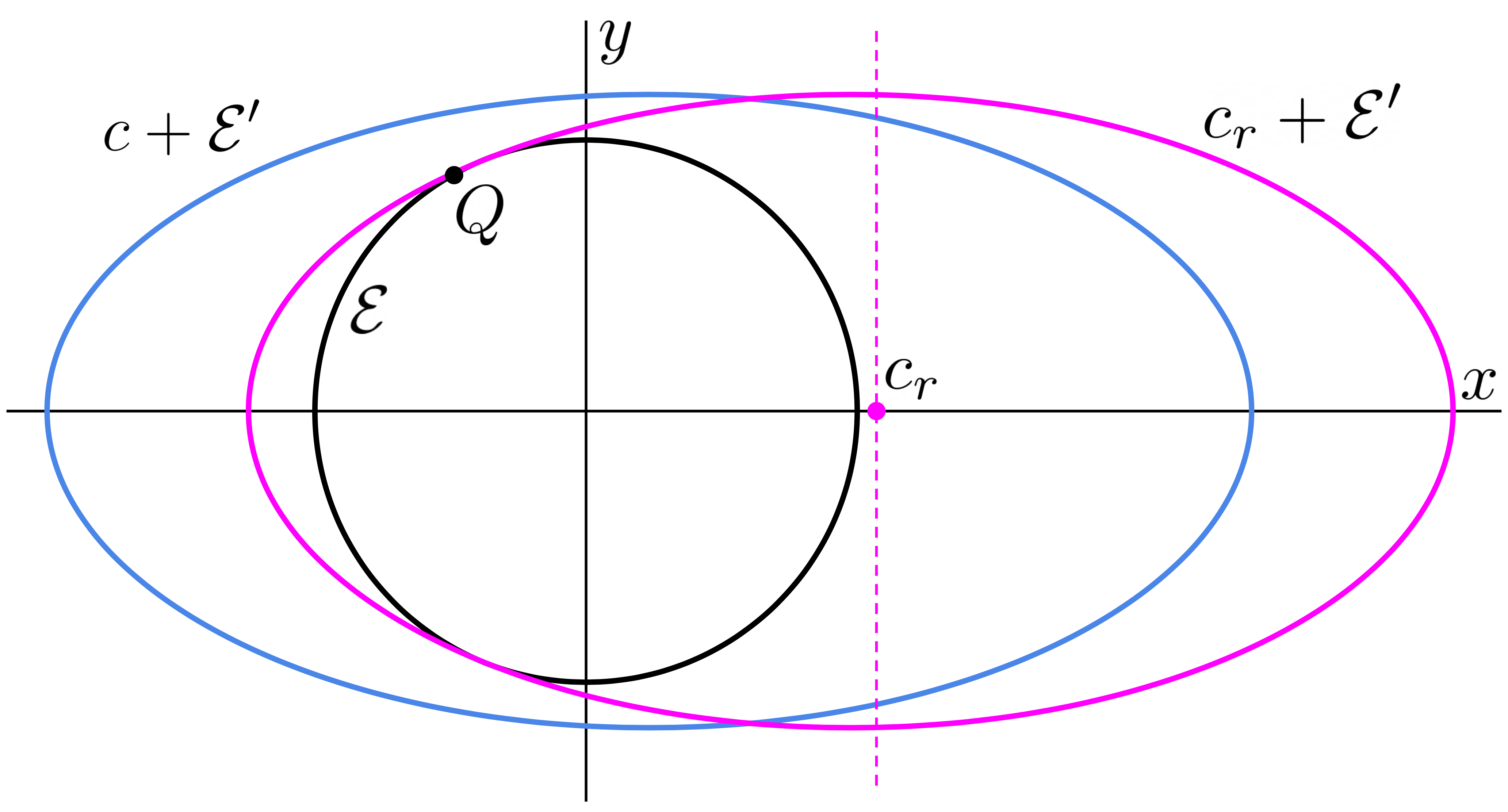}
\caption{Outer ellipses of the update step. As before, \(\cE\) is the black circle and \(c + \cE'\) is the blue ellipse.
\(c_r + \cE'\) is the magenta ellipse, with its center at \(c_r\) and the dotted magenta line showing the position of \(c_r\) along the \(x\)-axis.
\(c_r\) is defined so \(c_r + \cE'\) and \(\cE\) are tangent at two points. \(Q\) is one of these two tangent points.}
\label{fig:update_step_outer}
\end{figure}

First, note points on the boundary of \(c_r + \mathcal{E}'\) are described by the equation \begin{equation}\label{eqn:def_cr}
\frac{(x - c_r)^2}{a^2} + \frac{y^2}{b^2} = 1
\end{equation}
Let \(Q = (x', y')\) be the point of intersection between \(\cE\) and \(c_r + \cE'\) where \(y' > 0\).
Since \(Q\) is on the boundary of both ellipses, the vectors \(\left(\frac{2(x'-c_r)}{a^2}, \frac{2 y'}{b^2}\right)\) and \(\left(2x', 2y'\right)\), which are the normal vectors at \(Q\) of \(c_r + \cE'\) and \(\cE\) respectively, must be parallel.
Thus \(\frac{4 (x' - c_r)}{a^2} \cdot y' = \frac{4 y' x'}{b^2}\), which simplifies to \begin{equation}\label{eqn:x_cr}
x' = \frac{c_r}{1 - \frac{a^2}{b^2}}
\end{equation}
At this point we have a system of three equations relating \((x', y')\) and \(c_r\): (\ref{eqn:x_cr}), \(Q\) lying on \(\cE\), and \(Q\) satisfying (\ref{eqn:def_cr}).
We now solve this system to find \(c_r\).
To start, we expand (\ref{eqn:def_cr}) into
\(x'^2 - 2 x' c_r + c_r^2 + y'^2 \frac{a^2}{b^2} = a^2\),
which we rewrite into \(x'^2 \frac{a^2}{b^2} + x'^2 \left(1 - \frac{a^2}{b^2}\right) - 2 x' c_r + c_r^2 + y'^2 \frac{a^2}{b^2} = a^2\).
As \(Q\) lies on \(\cE\), this becomes \(x'^2 \left(1 - \frac{a^2}{b^2}\right) - 2 x' c_r + c_r^2 +  \frac{a^2}{b^2} = a^2\).
Substituting in (\ref{eqn:x_cr}), we get
\[\frac{c_r^2}{1 - \frac{a^2}{b^2}} - 2 \frac{c_r^2}{1 - \frac{a^2}{b^2}} + c_r^2 + \frac{a^2}{b^2} = a^2\]
Simplifying, we have \(c_r^2 \left(1 - \frac{b^2}{b^2 - a^2}\right) = a^2 \left(1 - \frac{1}{b^2}\right)\), i.e.
\[c_r^2 = \frac{b^2 - 1}{b^2} (a^2 - b^2)\]
To complete the claim it suffices to show \(c^2 \leq \frac{b^2 - 1}{b^2} (a^2 - b^2)\), which we do in \Cref{claim:pty_outer}.
\end{proof}

Now, we move on to the inner ellipsoid invariant of \Cref{def:alg-invariant}.
In particular, we will argue that \(c + \alpha' \cE' \subseteq \conv{\alpha\cE \cup \{\vz\}}\).
On a high level, we show this by arguing that the boundary of \(c + \alpha' \cE'\) does not intersect the boundary of \(\conv{\alpha\cE \cup \{\vz\}}\), except at points of tangency.

We can split the boundary of \(\conv{\alpha\cE \cup \{\vz\}}\) into two pieces: the part that intersects with the boundary of \(\alpha \cE\), which is an arc of the boundary of \(\alpha \cE\); and the remainder, which can described as two line segments connecting \(\vz\) to that arc. In particular, there are two lines that go through \(\vz\) and are tangent to \(\alpha \cE\), one of which we call line \(L\), and the other line is the reflection of \(L\) across the \(x\)-axis.
We define \(P_1\) and \(P_2\) as the tangent points of these lines to \(\alpha \cE\). 
Then, the boundary of \(\conv{\alpha\cE \cup \{\vz\}}\) consists of an arc \(P_1 P_2\) and the line segments \(\overline{P_1 \vz}, \overline{P_2 \vz}\). This is illustrated in \Cref{fig:update_step_inner}.
Note that at this point it is possible a priori for the arc \(P_1 P_2\) that coincides with the boundary of \(\conv{\alpha\cE \cup \{\vz\}}\)
to be either the major or minor arc; we will later show it must be the major arc.
We will take \(L\) to be the line whose tangent point to \(\alpha \cE\), \(P_1\), is above the \(x\)-axis, though this choice is arbitrary due to symmetry across the \(x\)-axis.

\begin{figure}[h]
\centering
\includegraphics[height=6cm]{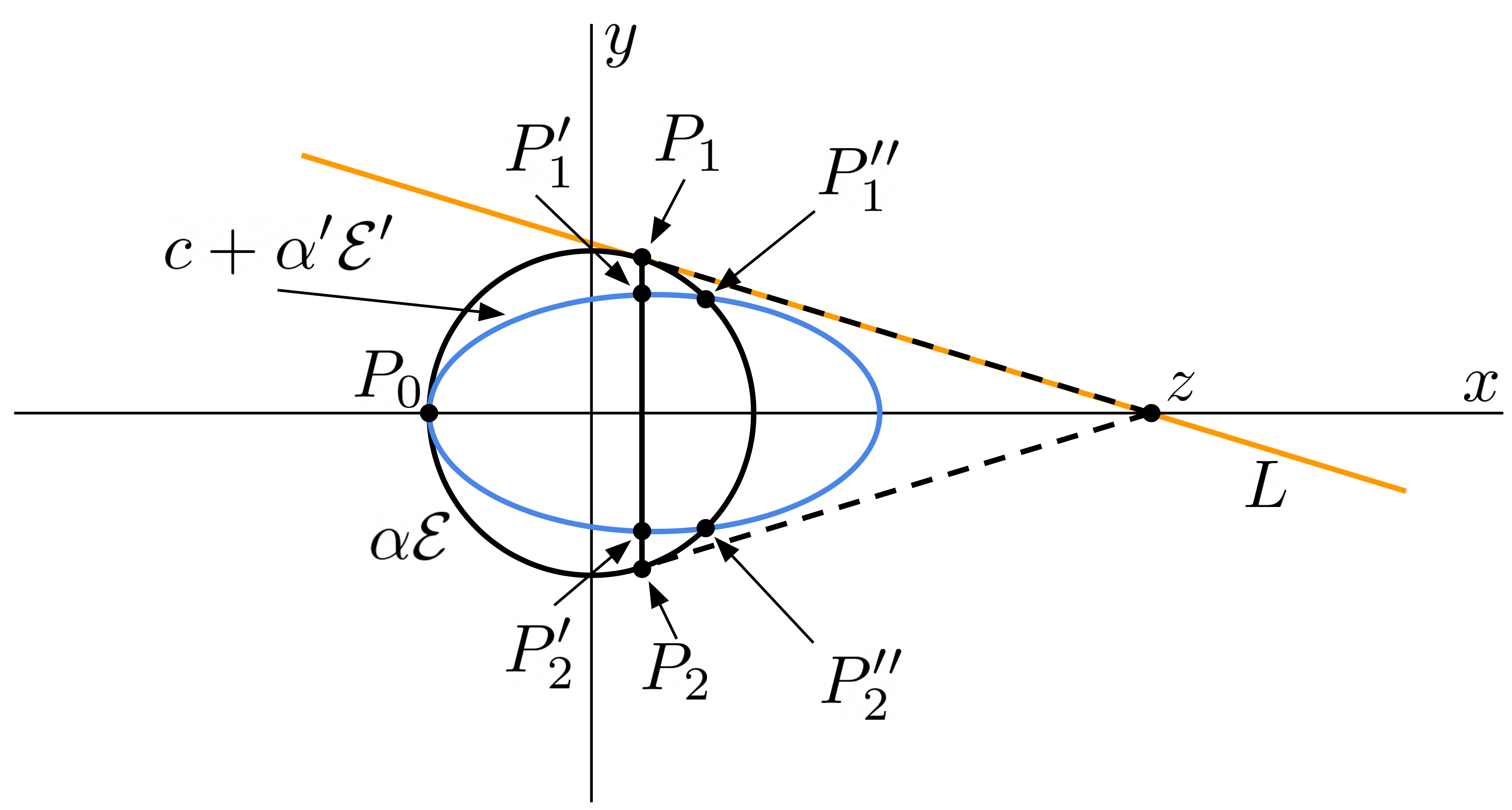}
\caption{Inner ellipses of the update step. As before, \(\alpha \cE\) is the black circle and \(c + \alpha' \cE\) is the blue ellipse. 
\(P_0\) is the shared leftmost point of \(\alpha \cE\) and \(c + \alpha' \cE'\).
There are two lines through \(\vv\) that are tangent to \(\alpha \cE\), one of which we call \(L\) and pictured in orange. We call the tangent points \(P_1\) and \(P_2\).
The line segments \(\overline{P_1 \vz}, \overline{P_2 \vz}\) are the dotted black lines.
\(P_1'\) and \(P_2'\) are the two points of intersection between \(c + \alpha' \cE\) and the line segment \(\overline{P_1 P_2}\). \(P_1''\) and \(P_2''\)
are the two points of intersection between \(\partial (c + \alpha' \cE')\) and \(\partial \alpha \cE\) to the right of the \(y\)-axis. Note that \(P_2, P_2', P_2''\) are the reflections of \(P_1, P_1', P_1''\) across the \(x\)-axis.}
\label{fig:update_step_inner}
\end{figure}

We first show that \(c + \alpha' \cE'\) does not intersect \(\overline{P_1 \vz}\) and \(\overline{P_2 \vz}\), except possibly at points of tangency.
In fact, we show a slightly stronger statement, in similar fashion to \Cref{claim:update_step_outer}.

\begin{claim}\label{claim:translate_ellipse_angle}
 \(c + \alpha' \cE'\) lies inside the angle \(\angle P_1 \vz P_2\).
\end{claim}
\begin{proof}
We translate \(c + \alpha' \cE'\) to the right until it touches \(L\) (and, by symmetry, \(\overline{P_2 \vz}\)).
We call this translated ellipse \(c_+ + \alpha' \cE'\), as shown in \Cref{fig:update_step_inner_translate}. 
(Formally, the center \(c_+\) can be described not as a translation from some other ellipse,
but as \(c_+\) such that \(c_+ + \alpha' \cE'\) intersects \(L\) at one point).
Observe that 
if \(c \leq c_{+}\), then \(c + \alpha' \cE'\) lies inside the angle \(\angle P_1 \vz P_2\).
We now determine \(c_{+}\).

\begin{figure}[h]
\centering
\includegraphics[height=6cm]{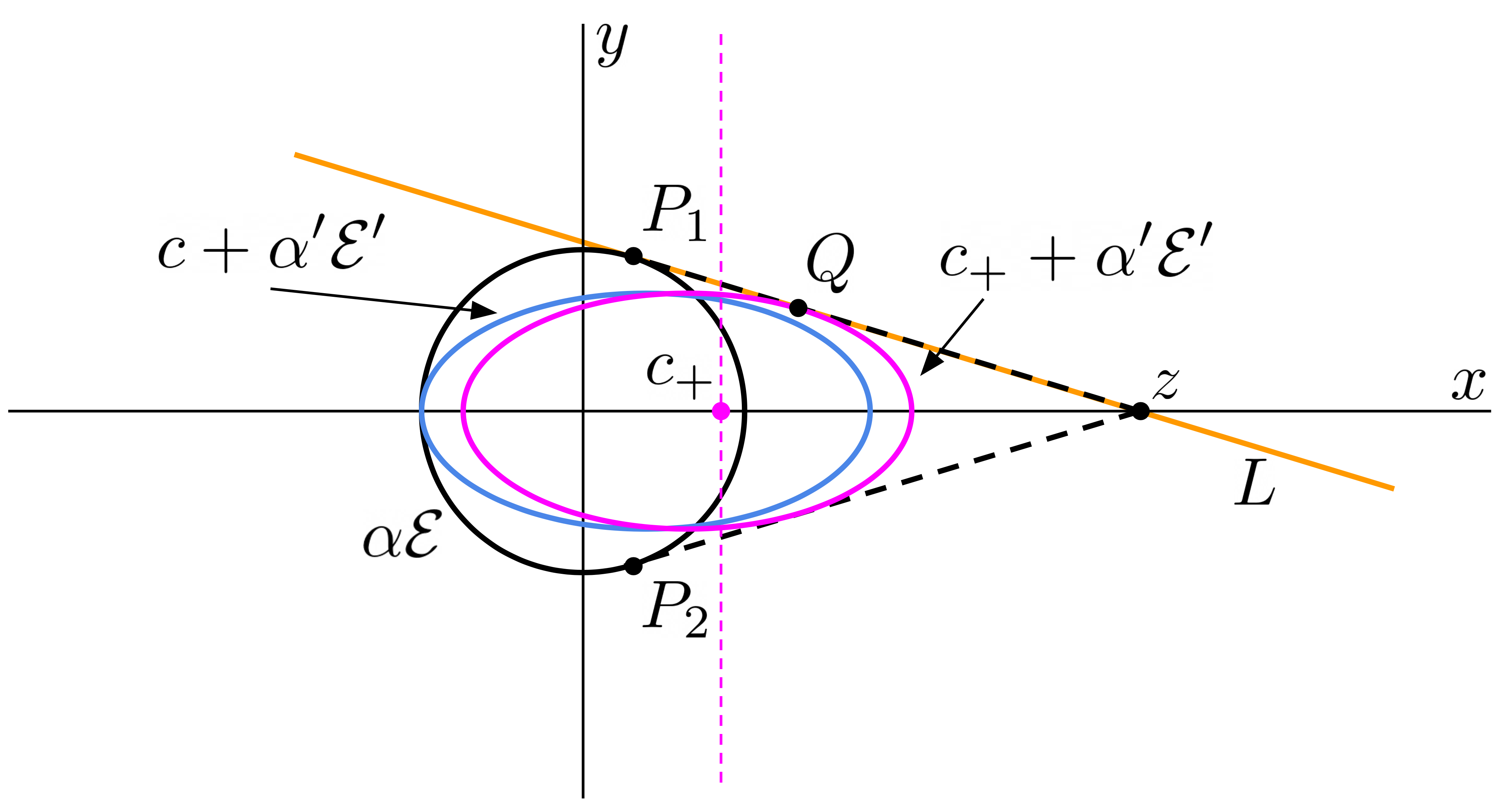}
\caption{Inner ellipses of the update step. As before, \(\alpha \cE\) is the black circle, \(c + \alpha' \cE\) is the blue ellipse, \(L\) is the orange line through \(\vz\) and tangent to \(\alpha \cE\), \(P_1\) and \(P_2\) are the tangent points on the lines through \(\vz\) tangent to \(\alpha \cE\), and \(\overline{P_1 \vz}, \overline{P_2 \vz}\) are the dotted black lines.
\(c_{+} + \alpha' \cE'\) is the magenta ellipse, with its center at \(c_{+}\) and magenta dotted line showing its position on the \(x\)-axis.
\(c_{+}\) is defined so that \(c_{+} + \alpha' \cE'\) is tangent to \(\overline{P_1 \vz}\) and \(\overline{P_2 \vz}\), with \(Q\) as the tangent point of \(c_{+} + \alpha' \cE'\) and \(\overline{P_1 \vz}\).
}
\label{fig:update_step_inner_translate}
\end{figure}

The equation of \(L\) is \[\underbrace{\frac{1}{c+a}}_{\ell_1} \cdot x + \underbrace{\sqrt{\frac{1}{\alpha^2} - \frac{1}{(c+a)^2}}}_{\ell_2} \cdot y = 1\]
where we define \(\ell_1, \ell_2\) as the coefficents for \(x\) and \(y\).
Observe that \(\vz\) is on \(L\), and \(L\) is tangent to \(\alpha \cE\) at \(P_1\), which has coordinates \begin{equation}\label{eqn:usit_p1} 
P_1 = \left(\frac{\alpha^2}{c+a}, \alpha^2 \sqrt{\frac{1}{\alpha^2} - \frac{1}{(c+a)^2}} \right).
\end{equation}
Tangency can be confirmed by checking that \(P_1\) is parallel to \((\ell_1, \ell_2)\), the normal vector definining \(L\).

Let \(Q = (x', y')\) be the point of intersection of \(L\) and \(c_{+} + \alpha' \cE\), there are three properties that define \(Q\). First it lies on the boundary of \(c_{+} + \alpha' \cE\), so it satisfies
\begin{equation}\label{eqn:pty_inner_oe} 
    \frac{(x'-c_{+})^2}{a^2} + \frac{y'^2}{b^2} = \alpha'^2.
\end{equation}

Second, at \(Q\) the normal vectors for the equations defining \(c_{+} + \alpha' \cE\) and \(L\) are parallel, i.e.
\((\frac{2(x-c_{+})}{a^2},\frac{2y}{b^2})\) is parallel to \((\ell_1, \ell_2)\). So
\begin{equation}\label{eqn:pty_inner_gm} 
    \frac{(x'-c_{+})}{a^2} \ell_2 = \frac{y'}{b^2} \ell_1.
\end{equation}
Finally, \(Q\) lies on \(L\), so we have \(\ell_1 x' + \ell_2 y' = 1\).
Solving this for \(y'\), we get
\begin{equation}\label{eqn:pty_inner_ol} 
    y' = \frac{1-\ell_1 x'}{\ell_2}.
\end{equation}
These three equations form a system for \(x', y'\) and \(c_{+}\), which we now solve to find \(c_{+}\).
Taking the square of (\ref{eqn:pty_inner_gm}) and rearranging gives \(\frac{y'^2}{b^2} = \frac{b^2 (x'-c_{+})^2 \ell_2^2}{a^4 \ell_1^2} \).
Substituting this into (\ref{eqn:pty_inner_oe}), we get
\(\frac{(x'-c_{+})^2}{a^2} + \frac{b^2 (x'-c_{+})^2 \ell_2^2}{a^4 \ell_1^2} = \alpha'^2\).
Now, defining \(r \defeq\frac{a^2 \ell_1^2}{b^2 \ell_2^2}\), we group the terms of this equation into the form
\begin{equation}\label{eqn:pty_inner_1}
    (x'-c_{+})^2 \cdot \frac{1}{a^2} \left(1 + \frac{1}{r}\right) = \alpha'^2.
\end{equation}
We substitute (\ref{eqn:pty_inner_ol}) into (\ref{eqn:pty_inner_gm})
to get \(\frac{x'-c_{+}}{a^2} \ell_2 = \frac{\ell_1}{b^2} \frac{1 - x' \ell_1}{\ell_2}\).
Grouping for \(x'\) and rearranging yields
\begin{equation}\label{eqn:pty_inner_2}
    x' - c_{+} = \frac{r}{1+r} \left(\frac{1}{\ell_1} - c_{+}\right).
\end{equation}
Next, we substitute (\ref{eqn:pty_inner_2}) into (\ref{eqn:pty_inner_1}), and get after some cancellation
\[ \left(\frac{1}{\ell_1} - c_{+}\right)^2 = \alpha'^2 a^2 \cdot \frac{1+r}{r}.\]
Observe on the left hand side that \(\frac{1}{\ell_1} - c_{+} = c+a-c_{+}\).
Clearly the center \(c_{+}\) must be to the left of \(\vz\), so this must be non-negative.
Hence after taking the positive square root, we obtain
\[c_{+} = c + a - \alpha' \cdot a \sqrt{\frac{1+r}{r}}\]
It remains to show that \(c \leq c_{+}\), or equivalently that
\[a - \alpha' \cdot a \sqrt{\frac{1+r}{r}} \geq 0\]
which we do in \Cref{claim:pty_inner}.
\end{proof}

Now, we build on the previous claim to show the inner ellipsoid invariant.
\begin{claim}\label{claim:update_step_inner}
    \(c + \alpha' \cdot \cE' \subseteq \conv{\alpha\cdot\cE \cup \{\vz\}}\)
\end{claim}
\begin{proof}
We will argue that the boundary of \(c + \alpha' \cE'\) does not intersect the boundary of \(\conv{\alpha\cE \cup \{\vz\}}\), except at points of tangency.
This is sufficient to establish the claim, as \Cref{claim:translate_ellipse_angle} shows that \(c + \alpha' \cE'\) is internal to \(\angle P_1 \vz P_2\), and so if \(c + \alpha' \cE'\) does not intersect the boundary of \(\conv{\alpha\cE \cup \{\vz\}}\), \(c + \alpha' \cE'\) must lie inside of, or be disjoint from \(\conv{\alpha\cE \cup \{\vz\}}\).
Since the leftmost points of \(\alpha \cE\) and \(c + \alpha' \cE'\) coincide, \(c + \alpha' \cE'\) must then lie inside of  \(\conv{\alpha\cE \cup \{\vz\}}\).
Recall that the boundary of \(\conv{\alpha\cE \cup \{\vz\}}\) consists of the arc \(P_1 P_2\) and the line segments \(\overline{P_1 \vz}, \overline{P_2 \vz}\).
\Cref{claim:translate_ellipse_angle} already shows that the boundary of \(c + \alpha' \cE'\) does not intersect \(\overline{P_1 \vz}\) and \( \overline{P_2 \vz}\), so we only need to show that the boundary of \(c + \alpha' \cE'\) does not intersect the arc \(P_1 P_2\).

To do this, we start by enumerating the points of intersection of \(\partial \alpha \cE\) and \(\partial(c + \alpha' \cE')\), recalling that \(P_1 P_2\) is an arc of \(\partial \alpha \cE\).
Observe that the leftmost points of \(\alpha \cE\) and \(c + \alpha' \cE'\) coincide, as the leftmost point of \(c + \alpha' \cE'\) is \(c - \alpha' \cdot a = -\alpha\) by definition; we call this point \(P_0\). \(P_0\) is a point of tangency and hence has intersection multiplicity 2, because the centers of \(\alpha \cE\) and \(c + \alpha' \cdot \cE'\) both lie on the \(x\)-axis.

Next, we argue for the existence of two more distinct intersection points \(P_1'', P_2''\) as depicted in \Cref{fig:update_step_inner}.
The leftmost point of \(c + \alpha' \cE'\) is \((- \alpha, 0)\), and the rightmost point is \(c + \alpha'\), which by \Cref{claim:update_params}-(\ref{item:claim_update_params_4}) is to the right of \((\alpha, 0)\), the rightmost point of \(\alpha \cE\).
Thus, by lying on \(\partial \alpha \cE\), \(P_1, P_2\) lie between the leftmost and rightmost points of \(c + \alpha' \cE'\), and so \(c + \alpha' \cE'\) intersects the line through \(P_1\) and \(P_2\).
Further, by \Cref{claim:translate_ellipse_angle}, as \(c + \alpha' \cE'\) lies in the angle \(\angle P_1 \vv P_2\), \(c + \alpha' \cE'\) actually intersects the line segment \(\overline{P_1 P_2}\).
Observe that this intersection happens at two distinct points, which we call \(P_1'\) and \(P_2'\).
Both points are inside of \(\alpha \cE\), yet \(\partial (c + \alpha' \cE')\) is a continuous path that connects both to
the rightmost point of \(c + \alpha' \cE'\), which is outside of \(\alpha \cE\).
Thus \(\partial (c + \alpha' \cE')\) intersects \(\partial \alpha \cE\) at two more distinct points, which we call \(P_1''\) and \(P_2''\).

Now, we argue that \(P_1''\) and \(P_2''\) lie on the minor arc \(P_1 P_2\).
First, observe that the arc \(P_1 P_2\) containing \(P_0\) is the major arc. This is because \(P_1\) lies to the right of the \(y\)-axis, as determined in (\ref{eqn:usit_p1}); and by symmetry so does \(P_2\).
This also implies that major arc \(P_1 P_2\) is the arc with which the boundary of \(\conv{\alpha \cE \cup \{\vz\}}\) coincides.
\(P_1'\) and \(P_2'\) are collinear with \(P_1\) and \(P_2\), and as \(P_1''\) and \(P_2''\) are to the right of \(P_1'\) and \(P_2'\), this implies that they must lie on the minor arc \(P_1 P_2\).

Counting all the intersection points of \(\partial \alpha \cE\) and \(\partial(c + \alpha' \cE')\), we have \(P_0\) (with multiplicity \(2\))
and \(P_1''\) and \(P_2''\) (both with multiplicity 1); with total multiplicity 4.
Using \Cref{claim:conic_5pt}, it is impossible for them to have another intersection point without both ellipses being the same.
Thus \(\partial(c + \alpha' \cE')\) cannot intersect the major arc \(P_1 P_2\) except at \(P_0\), and so except at points of tangency the boundary of \(c + \alpha' \cE'\) does not intersect the boundary of \(\conv{\alpha\cE \cup \{\vz\}}\).
\end{proof}

\subsection{Generalizing to high dimension and arbitrary previous ellipsoid}
\label{sec:gen_high_dim}
Now that we have demonstrated the invariants of \Cref{def:alg-invariant} for the special two-dimensional case where the previous ellipsoid is the unit ball, we generalize slightly to high dimension, although we first still assume the previous ellipsoid is the unit ball.

Using the parameters as defined in (\ref{eqn:update_params}), we will let \(\cE = B_2^d\), and define the boundary of \(\cE'\) as
\[\frac{1}{a^2} (\vx_1 - c)^2 + \frac{1}{b^2} \vx_2^2 + \ldots + \frac{1}{b^2} \vx_d^2 = 1\]
Observe that we can also write \(\cE' = \cE_{\mD}\) where \(\mD = \diag{\frac{1}{a^2}, \frac{1}{b^2}, \ldots, \frac{1}{b^2}}\).
Similarly to before, we let \(\vz = (c+a, 0, 0, \ldots, 0) \in \R^d\),
the furthest point of \(c + \cE'\) in the positive direction of the \(x_1\)-axis.

Now, we argue that the invariants of \cref{def:alg-invariant}
still hold in this setting.
\begin{claim}\label{claim:i_gen_dim1}
The inner and outer ellipsoid invariants hold in this setting:
\begin{enumerate}
    \item \(\cE \subseteq c \cdot \ve_1 + \cE'\)
    \item \(c \cdot \ve_1 + \alpha' \cE' \subseteq \conv{\alpha\cE \cup \{\vz\}}\)
\end{enumerate}
\end{claim}
\begin{proof}
Observe that \(\cE\), \(c \cdot \ve_1 + \cE'\), \(c \cdot \ve_1 + \alpha' \cE'\), and \(\conv{\alpha \cE \cup \{\vz\}}\) are all bodies of revolution about the \(x_1\)-axis,
with their cross-sections given by their counterparts in \Cref{sec:two_dim_update}.
As \Cref{claim:update_step_outer} and \Cref{claim:update_step_inner}
hold for these cross sections, the set containments hold for
the bodies of revolution as well.
\end{proof}

We further generalize to the case where the previous ellipsoid is arbitrary.
In particular, let \(\vc^{\circ} + \cE\) be the previous ellipsoid, with a vector \(\vc^{\circ} \in \R^d\) and \(\cE = \cE_{\mA}\) for non-singular matrix \(\mA \in \R^{d \times d}\).
Let \(\vz^{\circ} \in \R^d\) be an arbitrary vector, representing the next point received.
We let \(\vu = \mA(\vz^{\circ} - \vc^{\circ})\), and \(\mW \in \R^{d \times d}\) be an orthogonal matrix with \(\vw = \frac{\vu}{\|\vu\|}\) as its first column (e.g. by using as its columns an orthonormal basis containing \(\vw\)).
We define the next outer ellipsoid as \(\vc^{\circ}+c \mA^{-1} \vw + \cE' \) for \(\cE' = \cE_{\mW \mD \mW^\top \mA}\), with \(\mD = \diag{\frac{1}{a^2}, \frac{1}{b^2}, \ldots, \frac{1}{b^2}}\) as before.
Observe that \(\vz = \vc^{\circ}+(c+a) \mA^{-1} \vw\) is the furthest point of \(\vc^{\circ}+c \mA^{-1} \vw + \cE'\) from the previous center \(\vc^{\circ}\) towards \(\vz^{\circ}\).

This setup works to preserve the key invariants, as we see in the next claim.
\begin{claim}\label{claim:i_gen_us}
   The inner and outer ellipsoid invariants hold in this setting:
    \begin{enumerate}
        \item \(\vc^{\circ} + \cE \subseteq \vc^{\circ}+c \mA^{-1} \vw + \cE' \)
        \item \(\vc^{\circ}+c \mA^{-1} \vw + \alpha' \cE' \subseteq \conv{(\vc^{\circ} + \alpha\cE) \cup \{\vz\}}\)
    \end{enumerate}
\end{claim}
\begin{proof}
    We translate both set inclusions by \(- \vc^{\circ}\), then apply the nonsingular linear transformation \(\mW^\top \mA\).
    Observe that the set inclusions we wish to prove hold if and only if the transformed ones do.
    Noting that \(\mW^\top \mA \cE' = \cE_{\mW \mD}\), the transformed set inclusions are
    \(\cE_{\mW} \subseteq c \cdot \ve_1 + \cE_{\mW \mD}\)
    and \(c \cdot \ve_1 + \alpha' \cE_{\mW \mD} \subseteq \conv{\alpha \cE_{\mW} \cup \{(c + a) \cdot \ve_1\}}\).
    However, since \(\mW\) is an orthogonal matrix,
    \(\cE_{\mW} = B_2^d\) and \(\cE_{\mW \mD} = \cE_{\mD}\),
    and so the inclusions are exactly those shown in \Cref{claim:i_gen_dim1}.
\end{proof}

Choosing \(\gamma\) correctly in (\ref{eqn:update_params})
ensures that \(\vz \in \vc^{\circ}+c \mA^{-1} \vw + \cE'\) coincides with \(\vz^{\circ}\),
as stated in the upcoming claim. This can be seen by looking at the definition of \(\vz\).
\begin{claim}\label{claim:i_gen_v}
    If \(\gamma\) is chosen so that \(c + a = \|\vu\|\), then \(\vz = \vz^{\circ}\).
\end{claim}

\subsection{General algorithm}
\label{sec:gen_alg}
The goal of this section is to give and analyze a full algorithm that solves the streaming ellipsoid approximation problem, building on the analysis of the update rule from the previous sections.

Before we describe the complete algorithm, we give pseudocode in \Cref{alg:full_update} for its primary primitive, which is an update step like the one we analyzed in the previous section, \Cref{sec:gen_high_dim}.



\begin{algorithm}[h]
\caption{Full update step \(\cA^{\mathsf{full}}\)}\label{alg:full_update}

\textbf{input}: \(\mA_{t-1} \in \R^{d \times d}, \vc_{t-1} \in \R^d, \alpha_{t-1} \in [0, \frac{1}{2}], \vz_t \in \R^d\) \\
\textbf{output}: \(\mA_{t} \in \R^{d \times d}, \vc_{t} \in \R^{d}, \alpha_{t} \in [0, \alpha_{t-1}]\)
\begin{algorithmic}[1]
    \State Let \(\vu = \mA_{t-1}(\vz_t - \vc_{t-1})\), \(\vw = \frac{\vu}{\|\vu\|}\)
    \If{\(\|\vu\| > 1\)}
        \State Let \(\gamma^{\star}_t\) be such that \(a(\gamma^{\star}_t) + c(\gamma^{\star}_t) = \|\vu\|\)\label{alg_line:compute_update_t}
        \State \(\hat{\mA} = \frac{1}{b(\gamma^{\star}_t)} \mI_d + \left(\frac{1}{a(\gamma^{\star}_t)} - \frac{1}{b(\gamma^{\star}_t)}\right) \vw \vw^\top\)\label{alg_line:ellipsoid_update}
        \State \Return \(\mA_t = \hat{\mA} \cdot \mA_{t-1}, \ \vc_t = \vc_{t-1} + c(\gamma^{\star}_t) \mA_{t-1}^{-1} \vw, \ \alpha_t = \alpha'(\gamma^{\star}_t)\) \label{alg_line:update_mat}
    \Else
        \State \Return \(\mA_t = \mA_{t-1}, \vc_t = \vc_{t-1}, \alpha_i = \alpha_{t-1}\)
        \label{alg_line:no_update}
    \EndIf
\end{algorithmic}
\end{algorithm}

In Lines \ref{alg_line:compute_update_t}, \ref{alg_line:ellipsoid_update} and \ref{alg_line:update_mat}, we use the definition of \(a(\gamma), b(\gamma), c(\gamma), \alpha'(\gamma)\) from (\ref{eqn:update_params}),
substituting \(\alpha_{t-1}\) for \(\alpha\).
Although the update step does not explicitly mention ellipsoids,
we use \(\cE_t = \cE_{\mA_t}\) so that at iteration \(t\) the next outer and inner ellipsoids are \(\vc_t + \cE_{\mA_t}\)
and \(\vc_t + \alpha_t \cE_{\mA_t}\), respectively.
If at this iteration \(\|\vu\| \leq 1\), we will refer to this as the case where the ellipsoids are not updated, as is clear from Line \ref{alg_line:no_update}.

Observe also that if in iteration \(t\) we let \(\mW \in \R^{d \times d}\) be an orthogonal matrix with \(\vw\) as its first column,
we can write
\begin{equation}\label{eqn:alg_mat_update}
\hat{\mA} = \mW \cdot \diag{\frac{1}{a(\gamma^{\star}_t)}, \frac{1}{b(\gamma^{\star}_t)}, \cdots, \frac{1}{b(\gamma^{\star}_t)}} \cdot \mW^\top
\end{equation}

Now, we argue that this algorithm satisfies the invariants defined in \Cref{def:alg-invariant}.
This argument is essentially the observation that the update step in the algorithm is the one analyzed in \Cref{claim:i_gen_us}.
\begin{claim}\label{claim:alg_correctness}
\Cref{alg:full_update} is a monotone update; i.e. satisfies the invariants in \Cref{def:alg-invariant}.
\end{claim}
\begin{proof}
If \(\|\vu\| \leq 1\), then \(\vz_i \in \vc_n + \cE_n\) and the inner and outer ellipsoids are not updated, so the invariants clearly hold.
Otherwise, we apply \Cref{claim:i_gen_us} and \Cref{claim:i_gen_v} setting \(\mA = \mA_{t-1}, \vc^{\circ} = \vc_{t-1}, \vz^{\circ} = \vz_{t}, \alpha=\alpha_{t-1}\).
Using (\ref{eqn:alg_mat_update}), \(\cE_{\mA_t}\) is the same as  \(\cE'\) in \Cref{claim:i_gen_us}; and clearly \(\alpha_{t} = \alpha'\).
This establishes the inner ellipsoid invariant \(\vc_{t} + \alpha_{t} \cE_{t} \subseteq \conv{(\vc_{t-1} + \alpha_{t-1} \cE_{t-1}) \cup \{\vz_{t}\}}\) directly.
To show \(\conv{(\vc_{t-1} + \cE_{t-1}) \cup \{\vz_{t}\}} \subseteq \vc_{t} + \cE_{t}\),
observe that we have \(\vc_{t-1} + \cE_{t-1} \subseteq \vc_{t} + \cE_{t}\) from \Cref{claim:i_gen_us},
and \(\vz_{t} \in \vc_{t} + \cE_{t}\) from \Cref{claim:i_gen_v}.
Then the outer ellipsoid invariant follows as \(\vc_{t} + \cE_{t}\) is a convex set.
\end{proof}

Finally, we bound the relevant quantities that will be used in the analysis of the full algorithm's approximation factor.
In particular, we show that \(\exp(\gamma^{\star}_t)\) gives a lower bound on the increase in volume at each iteration \(t\).
If \(\|\vu\| \leq 1\), and the ellipsoids are not updated, in that iteration we think of \(\gamma^{\star}_t = 0\).
\begin{claim}\label{claim:vol_t}
For any input given to \Cref{alg:full_update}, we have  \(\vol(\cE_i) \geq \exp(\gamma^{\star}_t) \vol(\cE_{t-1})\).
\end{claim}
\begin{proof}
This formula is clearly true when the ellipsoids are not updated because \(\gamma^{\star}_t = 0\), so we consider the nontrivial case.
Recall the formula \(\vol(\cE_\mA) = \det(\mA^{-1}) \vol(B_2^d)\) from \Cref{claim:prelim_vol}. Then we have
\[\vol(\cE_{\mA_i}) = \det(\mA_i^{-1}) \vol(B_2^d) = \det(\hat{\mA}^{-1}) \cdot \det(\mA_{t-1}^{-1}) \cdot \vol(B_2^d) = \det(\hat{\mA}^{-1}) \vol(\cE_{\mA_{t-1}})\]
where we use the definition of \(\hat{\mA}\) from Line \ref{alg_line:ellipsoid_update} on the \(t\)-th iteration.
Then \begin{align*}
    \det(\hat{\mA}^{-1}) &= a(\gamma^{\star}_t) \cdot b(\gamma^{\star}_t)^{d-1} & \text{using (\ref{eqn:alg_mat_update})}\\
    &\geq a(\gamma^{\star}_t) & \text{by \Cref{claim:update_params}-(\ref{item:claim_update_params_2})} \\
    &= \exp(\gamma^{\star}_t) & \text{by definition of \(a\) in (\ref{eqn:update_params})}
\end{align*}
and using \(\vol(\cE_{\mA_i}) = \det(\hat{\mA}^{-1}) \cdot \vol(\cE_{\mA_{t-1}})\) completes the proof.
\end{proof}

We are now ready to present the complete algorithm in \Cref{alg:main}.
The algorithm is explicitly given \(\vc_0 + r_0 \cdot B_2^d \subseteq Z\),
for simplicity here we say \(r = r_0\).
Let \(R = R(Z)\); while the final approximation factor depends on this quantity,
the algorithm is not given it.
Note that \(\kappa(Z) \leq \nfrac{R}{r}\), so the quality of the approximation depends not only on \(\kappa(Z)\), but also on how well the given ball \(\vc_0 + r \cdot B_2^d\) is centered within \(Z\).
This algorithm proceeds in two phases. It begins with a `local' first phase, where the inner ellipsoid is a ball kept at radius \(r\), and the outer ellipsoid is a ball scaled to contain all the points. For readability, the variables of the algorithm in this phase are annotated with a superscript \(^{(l)}\). The second phase starts if the approximation factor of the first phase ever reaches \(\alpha^{(l)} \leq \frac{1}{d \log d}\),
at which point the algorithm uses the `full' update that was just described in \Cref{alg:full_update}.
We use two phases because while the full update reaches a near-optimal approximation factor when \(\nfrac{R}{r} \geq d \log d\), the local phase using balls does better when \(\nfrac{R}{r} \leq d \log d\).
While we cannot tell when to switch phases exactly (this would require knowing \(\nfrac{R}{r}\)), we show that it is enough to approximate the aspect ratio during the first phase up to a constant factor.

\begin{algorithm}[h]
\caption{Streaming ellipsoid rounding -- complete algorithm}\label{alg:main}

\textbf{input}: \(\vc_0+r B_2^d \subseteq Z\) \\
\textbf{output}: \(\vc_n + \cE_n, \vc_n + \alpha_n \cdot \cE_n\)
\begin{algorithmic}[1]
\State Initialize \(\mA^{(l)}_0 = \frac{1}{r} \mI_d, \vc_0^{(l)} = \vc_0, \alpha_0^{(l)} = 1\)
\State \({t^{(l)}} = 0, R_0 = 0\)
\While{\({t^{(l)}} \leq n\)}
\Comment{Phase I: Local update step that maintains a ball}
    \State Receive point \(\vz_{t^{(l)}}\)
    \If{\(\|\vz_{t^{(l)}} - \vc_0\| \leq r \cdot d \log d\)}\label{alg_line:phaseI_if}
        \If{\(\|\vz_{t^{(l)}} - \vc_0\| > R_{{t^{(l)}}-1}\)}
            \State \(\mA_{t^{(l)}}^{(l)} = \frac{1}{\|\vz_{t^{(l)}} - \vc_0\|} \cdot \mI_d, \vc_{t^{(l)}}^{(l)} = \vc_{{t^{(l)}}-1}^{(l)}, \alpha_{t^{(l)}}^{(l)} = \frac{r}{\|\vz_{t^{(l)}} - \vc_0\|}\)
            \Comment{Grow the ball to contain \(\vz_{t^{(l)}}\)}
            \State \(R_{t^{(l)}} = \frac{\|\vz_{t^{(l)}} - \vc_0\|}{r}\)
        \Else
        \State \(\mA_{t^{(l)}}^{(l)} = \mA_{{t^{(l)}}-1}, \vc_i^{(l)} = \vc_{{t^{(l)}}-1}^{(l)}, \alpha_{t^{(l)}}^{(l)} = \alpha_{{t^{(l)}}-1}\)
        \State \(R_{t^{(l)}} = R_{{t^{(l)}}-1}\)
        \EndIf
    \Else{}
        \State \textbf{break} 
        \Comment{Break the loop and jump to Line \ref{alg_line:end_phase1}}
    \EndIf
    \State \({t^{(l)}} = {t^{(l)}} + 1\)
\EndWhile
\If{\({t^{(l)}} > n\)}\label{alg_line:end_phase1}
    \Comment{If we stayed in Phase I for the entire execution of the algorithm}
    \State \Return \(\mA_n^{(l)}, \vc_n^{(l)}, \alpha_n^{(l)}\)
\EndIf
\State \(t_{s} = {t^{(l)}}\) \label{alg_line:tswitch}
\Comment{Point \(\vz_{t_{s}}\) has not yet been processed}
\State \(\mA_{t_{s}-1} =\frac{1}{r d \log d} \cdot \mI_d, \vc_{t_{s}-1} = \vc_{t_{s}-1}^{(l)}, \alpha_{t_{s}-1} = \frac{1}{d \log d}\)\label{alg_line:phase_trans}
\Comment{Transition: grow the ball to maximium size}
\For{\(t \in \{t_{s}, t_{s}+1, \ldots, n\}\)}
    \Comment{Phase II: full update for the remaining points}
    \State Receive point \(\vz_i\)
    \State \(\mA_i, \vc_i, \alpha_i = \cA^{\mathsf{full}}(\mA_{t-1}, \vc_{t-1}, \alpha_{t-1}, \vz_i)\)\label{alg_line:phaseII_update}
\EndFor
\State \Return \(\mA_n, \vc_n, \alpha_n\)
\end{algorithmic}
\end{algorithm}

Before Line \ref{alg_line:end_phase1}, the algorithm executes the first phase that has the outer and inner ellipsoids as balls.
In Line \ref{alg_line:end_phase1}, we have \({t^{(l)}} > n\) if the algorithm stayed in Phase I for every point, i.e. we had \(\max_{1 \leq {t^{(l)}} \leq n} \| \vz_{t^{(l)}} - \vc_0 \| \leq r \cdot d \log d\). In this case, the algorithm returns the approximation maintained by Phase I.
Otherwise we must have come across a point where \(\|\vz_{t^{(l)}} - \vc_0 \| > r \cdot d \log d\), and the algorithm proceeds with Phase II.
We let \(t_{s}\) in Line \ref{alg_line:tswitch} mark the point received that causes the algorithm to proceed to Phase II.
We then perform a `transition' on Line \ref{alg_line:phase_trans} that grows the ball of Phase I to its maximum size. This transition step makes the analysis of the complete algorithm easier, as then the starting approximation for the second phase is exactly \(\alpha_{t_s - 1} = \frac{1}{d \log d}\).
Then the algorithm runs the full update \(\cA^{\mathsf{full}}\) for the rest of the points, including \(\vz_{t_{s}}\).
For simplicity, we write our algorithm so that it `receives' \(\vz_{t_{s}}\) twice, once for each phase.
However, the first phase does not commit to an update for this point, and the ellipsoids in Line \ref{alg_line:phase_trans} are not committed either; the algorithm does not commit to an update for this point until Line \ref{alg_line:phaseII_update}.

Recall the approximation guarantee stated in \Cref{thm:main_one}:
\begin{equation}\label{eq:main_approx_guarantee}
    \frac{1}{\alpha_n} \leq O(\min\inparen{\nfrac{R}{r}, d \log\inparen{\nfrac{R}{r}}})
\end{equation}

We can interpret the approximation guarantee \eqref{eq:main_approx_guarantee} by cases depending on if \(\nfrac{R}{r} \geq d \log d\) (i.e. if the algorithm ever enters the second phase):
\begin{claim}
We have for all \(d \geq 2\) that
\[\min\inparen{\nfrac{R}{r}, d \log\inparen{\nfrac{R}{r}}} = \Theta\inparen{\begin{cases}
    d \log\inparen{\nfrac{R}{r}} & \text{if } \nfrac{R}{r} > d \log d \\
    \nfrac{R}{r} & \text{if } \nfrac{R}{r} \leq d \log d
\end{cases}}\]
\end{claim}

Now, we claim a straightforward geometric fact, that the distance of the furthest \(\vz_{t}\) to \(\vc_0\)
approximates the circumradius of \(Z\) up to a constant factor.
We will use this to show that Line \ref{alg_line:phaseI_if} will be able to properly detect when
\(\nfrac{R}{r} > d \log d\) (again, up to a constant factor).
\begin{claim}\label{claim:alg_r_est}
    Let \(\vc_0 + r_0 B_2^d \subseteq Z\), and 
    \(R = R(Z)\). Then
    \[R \leq \max_{1 \leq t^{(l)} \leq n} \| \vc_0 - \vz_{t^{(l)}} \| \leq 2 \cdot R\]
\end{claim}
\begin{proof}
    For the left inequality, observe that if we let \(r_{\max} = \max_{1 \leq t^{(l)} \leq n} \| \vc_0 - \vz_{t^{(l)}} \|\),
    then \(Z \subseteq \vc_0 + r_{\max} \cdot B_2^d\).
    For the right inequality, observe that for any containing ball \(\vc' + R' \cdot B_2^d \supseteq Z\),
    its diameter is \(2 R'\).
    But as \(\vc' + R' \cdot B_2^d\) contains \(\vc_0\) and \(\vz_{1}, \ldots, \vz_{n}\), we must have \(\diam{\vc' + R' \cdot B_2^d} \geq \diam{\{\vc_0\} \cup \{\vz_1, \ldots, \vz_n\}}\) and so \(2R' \geq r_{\max}\).
\end{proof}

Next, we discuss the approximation guarantee that the algorithm achieves, depending on the phase that it terminates with.
We start with if the algorithm only stays in the local phase, in which case we can readily apply the previous claim.
\begin{claim}\label{claim:approx_local}
    If \Cref{alg:main} never enters Phase II, then its approximation guarantee satisfies
    \(\frac{1}{\alpha_n} \leq \nfrac{2 R}{r}\).
\end{claim}
\begin{proof}
    At the termination of Phase I, the algorithm produces approximation \(\alpha_n^{(l)} = \max_{1 \leq {t^{(l)}} \leq n} \frac{\|\vz_{t^{(l)}} - \vc_0\|}{r}\).
    Using \Cref{claim:alg_r_est} we obtain \[
    \frac{1}{\alpha_n^{(l)}} = \max_{1 \leq {t^{(l)}} \leq n} \frac{\|\vz_{t^{(l)}} - \vc_0\|}{r} \leq \frac{2 R}{r}
    \]
\end{proof}

The analysis in the case where the algorithm enters the full phase is more involved.
We use \Cref{claim:vol_t}, which shows that the increase in approximation factor each iteration is not too large compared to the increase in volume, to bound \(\frac{1}{\alpha_n}\).
We know that the volume of the final ellipsoid \(\vc_n + \cE_n\) must be bounded relative to \(R \cdot B_2^d\), as the algorithm produces \(\vc_n + \alpha_n \cdot \cE_n \subseteq Z\);
however, this does lead to an upper bound that is still a function of \(\frac{1}{\alpha_n}\).
\begin{claim}\label{claim:alpha_rel}
    If \Cref{alg:main} enters Phase II, the approximation guarantee satisfies
    \[\frac{1}{\alpha_n} \leq 2 \left(d \log\left(\frac{1}{\alpha_n}\right) + d \log\left(\frac{R}{r}\right)\right)\]
\end{claim}
\begin{proof}
The algorithm transitions to Phase II at Line \ref{alg_line:tswitch}, starting at iteration \(t_{s}\).
At each subsequent iteration, we claim that \Cref{alg:full_update} guarantees \(\frac{1}{\alpha_{t}} = \frac{1}{\alpha_{t-1}} + 2 \gamma^{\star}_t\).
By \Cref{claim:update_params}-(\ref{item:claim_update_params_1}), we have for all \(t_{s} \leq t \leq n-1\) where the ellipsoids were updated that \(\frac{1}{\alpha_{t}} = \frac{1}{\alpha_{t-1}} + 2 \gamma^{\star}_t\).
When the ellipsoids are not updated, this still holds, as in that case \(\gamma^{\star}_t = 0\). 

As in Phase II the algorithm begins with \(\alpha_{t_{s} - 1} = \frac{1}{d \log d}\), we have
\begin{equation}\label{eqn:alpha_rel_t}
\frac{1}{\alpha_n} = d \log d + 2 \sum_{t=t_{s}}^{n-1} \gamma^{\star}_t
\end{equation}
Now applying \Cref{claim:vol_t} for each \(t\), we have \(\vol(\cE_n) \geq \exp\left(\sum_{t=t_{s}}^{n-1} \gamma^{\star}_t\right) \cdot \vol(\cE_{t_{s}})\). Taking logarithms gives
\begin{equation}\label{eq:ub_t_vol}
\log\inparen{\frac{\vol \cE_n}{\vol \cE_{t_{s} - 1}}} \geq \sum_{t=t_{s}}^{n-1} \gamma^{\star}_t
\end{equation}

Recall that \(\vc_0 + r \cdot B_2^d \subseteq Z\),
and by \Cref{def:circumradius},
\(Z \subseteq \vc_c + R \cdot B_2^d\) for some center \(\vc_c\).
By \Cref{claim:alg_correctness}, we have \(\vc_n + \alpha_n \cdot \cE_n \subseteq Z\), so
that \(\vol(\cE_n) \leq \frac{1}{\alpha_n^d} \cdot \vol(R \cdot B_2^d)\).
As in Phase II we start with \(\cE_{t_s - 1} = \vc_0 + r d \log d \cdot B_2^d\), this yields
\begin{align*}\sum_{t=t_{s}}^{n-1} \gamma^{\star}_t &\leq \log\inparen{\frac{\vol(\cE_n)}{\vol (\cE_{t_{s}-1})}} &\text{by \eqref{eqn:alpha_rel_t}} \\
&\leq d \log \inparen{\frac{1}{\alpha_n}} + \log\inparen{\frac{\vol(R \cdot B_2^d)}{\vol(r d \log d \cdot B_2^d)}} & \text{by } \vol(\cE_n) \leq \frac{1}{\alpha_n^d} \vol(R \cdot B_2^d)\\
&= d \log\inparen{\frac{1}{\alpha_n}} + d\log\inparen{\frac{R}{rd \log d}}\\
&\leq d \log \inparen{\frac{1}{\alpha_n}} + d \log \inparen{\frac{R}{r}} - d \log d
\end{align*}
and plugging into (\ref{eqn:alpha_rel_t}) finishes the claim.
\end{proof}

Intuitively, \(x \leq a + b \cdot \log x\) for some constants \(a, b > 0\) can only be true for bounded \(x\), as \(x = \omega(\log x)\).
As we showed \(1/\alpha_n\) satisfies a relation like this in \Cref{claim:alpha_rel}, we develop this intuition to give a quantitative upper bound on \(1/\alpha_n\).
\begin{claim}\label{claim:approx_full}
If \Cref{alg:main} enters Phase II, then we have \[\frac{1}{\alpha_n} \leq 8 d(\log d + \log \nfrac{R}{r})\]
\end{claim}
\begin{proof}
Assume towards contradiction that \(\frac{1}{\alpha_n} > 8 d (\log d + \log \nfrac{R}{r})\).
Observe then that \(\frac{1}{\alpha_n} - \frac{3}{4} \cdot\frac{1}{\alpha_n} > 2 d (\log d + \log \nfrac{R}{r})\). Using \Cref{claim:alpha_rel}, we have
\[2 (d \log \nfrac{1}{\alpha_n} + d \log \nfrac{R}{r}) \geq \frac{1}{\alpha_n} > 2 (d \log d + d \log \nfrac{R}{r}) + \frac{3}{4} \cdot \frac{1}{\alpha_n}\]
Simplifying the above inequality gives \(2d \log \nfrac{1}{d \cdot \alpha_n} > \frac{3}{4} \cdot \frac{1}{\alpha_n}\), i.e. 
\(2\log \nfrac{1}{d \cdot \alpha_n} > \frac{7}{8} \cdot \frac{1}{d \cdot \alpha_n}\).
It is clear that this is impossible by looking at the graph of the function \(x \mapsto 2 \log x - \frac{3}{4} x\), which is concave with a maximum of \(2(\log(\nfrac{8}{3}) - 1) < 0\).
\end{proof}

Now we combine the previous claims to prove the guarantees of \Cref{alg:main}.
\begin{proof}[Proof of \Cref{thm:main_one}.] We first discuss the approximation guarantee and correctness, then the memory and runtime complexity of \Cref{alg:main}.

\paragraph{Approximation guarantee}
We break the analysis of the approximation guarantee by cases, depending on the aspect ratio.
If \(\nfrac{R}{r} \leq \frac{1}{2} d \log d\),
then by \Cref{claim:alg_r_est} we have \(\max_{1 \leq t^{(l)} \leq n} \| \vc_0 - \vz_{t^{(l)}} \| \leq r d \log d\),
and the algorithm never enters Phase II.
By \Cref{claim:approx_local}, the final approximation factor is \(\nfrac{2R}{r}\).
If \(\nfrac{R}{r} > d \log d\), then by \Cref{claim:alg_r_est} we have \(\max_{1 \leq t^{(l)} \leq n} \| \vc_0 - \vz_{t^{(l)}} \| > r d \log d\), and the algorithm must enter Phase II.
Then \Cref{claim:approx_full} applies, and the final approximation factor is \(O(d (\log d + \log \nfrac{R}{r}) = O(d \log \nfrac{R}{r})\).

If \(\frac{1}{2} d \log d < \nfrac{R}{r} \leq d \log d\), then it is possible for the algorithm to never enter Phase II or for it to enter Phase II. Either way, we argue that the final approximation factor is \(\frac{1}{\alpha_n} \leq O\inparen{\nfrac{R}{r}}\).
If it does not enter Phase II, then by \Cref{claim:approx_local}, the approximation guarantee we get is \(\frac{1}{\alpha_n} \leq O(\nfrac{R}{r})\).
If it does enter Phase II, then by \Cref{claim:approx_full}
we have \[
\frac{1}{\alpha_n} \leq O(d \log d + d \log \nfrac{R}{r})
\]
Due to the assumption that \(\frac{1}{2} d \log d < \nfrac{R}{r} \leq d \log d\), we also have in this case that \(\frac{1}{\alpha_n} \leq O(\nfrac{R}{r})\).

\paragraph{Correctness} By \Cref{claim:alg_invariant}, to argue that the algorithm solves \Cref{prob:main} it is enough to show that it is monotone, i.e. it satisfies the invariants of \Cref{def:alg-invariant}.
It is clear that the local update in Phase I satisfies the invariants, as the outer ellipsoid is a ball of growing radius
and the inner ellipsoid is kept to the ball of radius \(r\). 
It is also clear that after the algorithm transitions to Phase II, all the full updates are monotone by \Cref{claim:alg_correctness} and the fact that the starting approximation factor for this phase is is \(\alpha_{t_s - 1} = \frac{1}{d \log d} \leq \frac{1}{2}\).
As algorithm transitions to Phase II,
observe that on Line \ref{alg_line:phase_trans} the radius of the outer ellipsoid grows again to \(r d \log d\) before applying the full update,
so the first first full update of Phase II is also monotone.

\paragraph{Memory and runtime complexity}
The memory complexity of the algorithm is \(O(d^2)\).
Observe that \Cref{alg:full_update} only stores a constant number of matrices in \(\R^{d \times d}\), vectors in \(\R^d\), or constants,
so its memory complexity is \(O(d^2)\). It is only instantiated once for each point received in Phase II, so the memory complexity in this phase \(O(d^2)\). Finally, the memory complexity in the first phase is also \(O(d^2)\) because it stores the same kind of quantities as \Cref{alg:full_update}.

To show the runtime of the algorithm is \(\widetilde{O}(nd^2)\), we show that the runtime to process each next point is at most \(\widetilde{O}(d^2)\). This is clear in Phase I, and during the transition to Phase II.
For the full update this is less clear, as \Cref{alg:full_update} uses both \(\mA_{t-1}\) and \(\mA_{t-1}^{-1}\)
which naively would require inverting a matrix on each iteration.
However, if we represent \(\mA\) using the SVD (see the next section and \Cref{claim:full_update_eff}), we can implement the update in \(\widetilde{O}(d^2)\) time.
This would require that \(\mA_{t_s - 1}\) be given in SVD form as well for the first full update, but it is already in that form as a scaled identity matrix.

\end{proof}

\subsubsection{Efficient implementation of the full update step}

In this section, we use a method similar to that in Algorithm 2 from \cite{mmo22} to show that the full update step can be implemented in \(\widetilde{O}(d^2)\) time.
In particular, we use the same subroutine \(\textsc{SVDRankOneUpdate}\) with signature
\begin{align}
    (\mU', \mSigma', \mV') = \textsc{SVDRankOneUpdate}((\mU, \mSigma, \mV), \vy_1, \vy_2)\label{eq:stange_update}
\end{align}
where the result \(\mU' \mSigma' (\mV')^\top\) is the SVD of the matrix \(\mU \mSigma \mV^\top + \vy_1 \vy_2^\top\).
\citet{stange08} shows that this procedure be done in \(O(d^2 \log d)\) time.
We rewrite \Cref{alg:full_update} in \Cref{alg:full_update_eff} to make it clear how to use the SVD representation and the efficient rank-1 update to efficiently implement the full update.
One can readily see that \Cref{alg:full_update_eff} has the exact same behavior as \Cref{alg:full_update}, and so gives the same approximation and correctness guarantees.

\begin{algorithm}[h]
\caption{Efficient full update step \(\cA^{\mathsf{full}}\)}\label{alg:full_update_eff}

\textbf{input}: \((\mU_{t-1}, \mSigma_{t-1}, \mV_{t-1}) \in \R^{d \times d}, \vc_{t-1} \in \R^{d}, \alpha_{t-1} \in [0, \frac{1}{2}], \vz_{t} \in \R^{d}\) \\
\textbf{output}: \((\mU_{t}, \mSigma_{t}, \mV_{t}) \in \R^{d \times d}, \vc_{t} \in \R^{d}, \alpha_{t} \in [0, \alpha_{t}]\)
\begin{algorithmic}[1]
    \State Let \(\vu = \mU_{t-1} \mSigma_{t-1} \mV_{t-1}^\top (\vz_{t} - \vc_{t-1})\), \(\vw = \frac{\vu}{\|\vu\|}\)
    \If{\(\|\vu\| > 1\)}
        \State Let \(\gamma^{\star}_t\) be such that \(a(\gamma^{\star}_t) + c(\gamma^{\star}_t) = \|\vu\|\)\label{alg_line:compute_update_t_eff}
        \State \(\vy_1 = \left(\frac{1}{a(\gamma^{\star}_t)} - \frac{1}{b(\gamma^{\star}_t)}\right) \vw, \vy_2 = \mV_{t-1} \mSigma_{t-1} \mU_{t-1}^\top \vw\)
        \State \((\mU_{t}, \mSigma_{t}, \mV_{t}) = \textsc{SVDRankOneUpdate}((\mU_{t-1}, \frac{1}{b(\gamma^{\star}_t)} \mSigma_{t-1}, \mV_{t-1}), \vy_1, \vy_2)\)
        \State \Return \((\mU_{t}, \mSigma_{t}, \mV_{t}), \ \vc_{t} = \vc_{t-1} + c(\gamma^{\star}_t) \mV_{t-1} \mSigma_{t-1}^{-1} \mU_{t-1}^\top \vw, \ \alpha_t = \alpha'(\gamma^{\star}_t)\)
    \Else
        \State \Return \((\mU_{t}, \mSigma_{t}, \mV_{t}) = (\mU_{t-1}, \mSigma_{t-1}, \mV_{t-1}), \vc_t = \vc_{t-1}, \alpha_t = \alpha_{t-1}\)
    \EndIf
\end{algorithmic}
\end{algorithm}

\begin{remark}\label{remark:compute_update_t_eff}
    We briefly explain why Line \ref{alg_line:compute_update_t_eff}, finding \(\gamma^{\star}\) such that \(a(\gamma^{\star}) + c(\gamma^{\star}) = \|\vu\|\)  can be implemented efficiently. This is a one-dimensional optimization problem, and \(\gamma \mapsto a(\gamma) + c(\gamma)\) using \(a, c\) as defined in \eqref{eqn:update_params} is monotone increasing, 
    so finding an approximate \(\gamma^{\star}\) can be done efficiently with binary search.
    In particular, we can choose \(\gamma^{\star}\) to be a slight overestimate so the update is still monotone after slightly increasing \(\alpha_t\).
    This does not affect the final approximation guarantee beyond constant factors.
\end{remark}

This algorithm performs a constant number of taking norms of vectors, matrix-vector products, and algebraic operations; as well as one rank-one SVD update. As explained in \Cref{remark:compute_update_t_eff}, finding \(\gamma_i^*\) can also be done in effectively constant time. Thus for our runtime guarantee, we have:
\begin{claim}\label{claim:full_update_eff}
    \Cref{alg:full_update_eff} runs in time \(O(d^2 \log d)\).
\end{claim}

\subsection{Fully-online asymmetric ellipsoidal rounding algorithm}

In this subsection, we prove Theorem \ref{thm:main_two}. See Algorithm \ref{alg:fully_online_approx}.

\begin{algorithm}
\caption{Fully online asymmetric ellipsoidal rounding}\label{alg:fully_online_approx}
\begin{algorithmic}[1]
\State \textbf{Input:} Stream of points $\vz_t$; 
 monotone update rule $\cA$ (Definition \ref{def:alg-invariant}) that takes as input the previous ellipsoid matrix $\mA$, center $\vc$, approximation factor $\alpha$, and update point $\vz$ and outputs the next ellipsoid matrix $\mA'$, center $\vc'$, and approximation factor $\alpha'$.
\State \textbf{Output:} Ellipsoid $\cE$, center $\vc$, and scale $\alpha \in (0,1)$ such that $\vc + \alpha \cdot \cE \subseteq \mathsf{conv}\inparen{\inbraces{\vz_1,\dots,\vz_n}} \subseteq \vc +\cE$.
\State Receive $\vz_1$; set $\mA = \mI_d$, $d_1 = 1$, $\vc_1 = \vz_1$, $\alpha_1 = 1$.
\For{$t = 2,\dots,n$}
    \State Receive $\vz_{t}$.
    \If{$\vz_{t} - \vc_{t-1} \notin \vspan{\vz_{1} - \vc_{t-1},\dots,\vz_{t-1}-\vc_{t-1}}$}\Comment{Irregular update step.}\label{line:irregular_step}
        \State Let $\vv_1,\dots,\vv_{d_{t-1}}$ be the singular vectors of $\mA$ corresponding to the semiaxes of $\cE_{t-1}$.
        \State Let $d_t = d_{t-1}+1$.
        \State Let $\vz_{d_{t}}' \coloneqq \frac{\vz_t - \sum_{i=1}^{d_{t-1}} \vv_i\ip{\vv_i,\vz_t}}{\norm{\vz_t - \sum_{i=1}^{d_{t-1}} \vv_i\ip{\vv_i,\vz_t}}_2}$.
        \State Let $\mM \coloneqq \mI_d - \frac{1}{\ip{\vv_{d_{t}}',\vz}}\cdot\inparen{\vz_t - \sqrt{1+2\alpha_{t-1}} \cdot \vv_{d_{t}}'}(\vv_{d_{t}}')^T$.
        \State Update $\mA_t \gets \mA_{t-1}\mM$.\Comment{Use \eqref{eq:stange_update} of \citet{stange08} to update $\vv_1,\dots,\vv_d$.}
        \State Update $\vc_t = \frac{\alpha_{t-1}}{1+2\alpha_{t-1}}\cdot\vz_t + \inparen{1-\frac{\alpha_{t-1}}{1+2\alpha_{t-1}}}\cdot\vc_{t-1}$.
        \State Update $\nfrac{1}{\alpha_t} \gets \nfrac{1}{\alpha_{t-1}}+1$.
    \Else  
        \State $\mA_t, \vc_t, \alpha_t = \cA(\mA_{t-1}, \vc_{t-1}, \alpha_{t-1}, \vz_{t})$
        \State $d_t \gets d_{t-1}$.
    \EndIf
\EndFor
\State \textbf{Output:} $(\vc_n, \cE_n, \alpha_n)$.
\end{algorithmic}
\end{algorithm}


To prove Theorem \ref{thm:main_two}, we need to show that our \textit{irregular update step} (a timestep $t$ when we have to update the dimensionality of our ellipsoid $\cE_{t-1}$ -- see Line \ref{line:irregular_step} of Algorithm \ref{alg:fully_online_approx}) still maintains the invariants we desire (Definition \ref{def:alg-invariant}).

Our plan is to first consider the special case of the irregular update where the new point to cover is conveniently located with respect to our current ellipsoids. We will see later that this special case is nearly enough for us to conclude the proof.

\begin{claim}
\label{claim:irregular_update_easy}
Let $Z \subset \R^{d}$ be a convex body where $Z$ lies in $\vspan{\vv_1,\dots,\vv_{d'}}$ for $d' < d$. For $0 < \alpha \le 1$, suppose we have
\begin{align*}
    \alpha \cdot \inbraces{\vz \in \vspan{\vv_1,\dots,\vv_{d'}} \suchthat \norm{\vz}_2 \le 1} \subseteq Z \subseteq \inbraces{\vz \in \vspan{\vv_1,\dots,\vv_{d'}} \suchthat \norm{\vz}_2 \le 1}.
\end{align*}
Then, for any $\vv_{d'+1}$ such that $\ip{\vv_i, \vv_{d'+1}} = 0$ for all $i \in [d']$ and for which
\begin{align*}
    \cE' &\coloneqq \inbraces{\vz \in \vspan{\vv_1,\dots,\vv_{d'+1}} \suchthat \norm{\vz}_2 \le \frac{1+\alpha}{\sqrt{1+2\alpha}}} \\
    \vc &\coloneqq \frac{\alpha}{\sqrt{1+2\alpha}} \cdot \vv_{d'+1}
\end{align*}
we have
\begin{align*}
    \vc + \frac{1}{1 + \nfrac{1}{\alpha}} \cdot \cE' \subseteq \mathsf{conv}\inparen{Z \cup \inbraces{\sqrt{1+2\alpha} \cdot \vv_{d'+1}}} \subseteq \vc + \cE'.
\end{align*}
\end{claim}
\begin{proof}[Proof of Claim \ref{claim:irregular_update_easy}]
We will show that the pair of ellipsoids given below satisfy the conditions promised by the statement of Claim \ref{claim:irregular_update_easy}.
\begin{align}
    &\inbraces{\vz \in \vspan{\vv_1,\dots,\vv_{d'+1}} \suchthat \norm{\vz}_2 \le \frac{1+\alpha}{\sqrt{1+2\alpha}}} +  \frac{\alpha}{\sqrt{1+2\alpha}} \cdot \vv_{d'+1}\label{eq:irregular_outer} \\
    &\inbraces{\vz \in \vspan{\vv_1,\dots,\vv_{d'+1}} \suchthat \norm{\vz}_2 \le \frac{1+\alpha}{\sqrt{1+2\alpha}}} \cdot \frac{\alpha}{1+\alpha} +  \frac{\alpha}{\sqrt{1+2\alpha}} \cdot \vv_{d'+1}\label{eq:irregular_inner}
\end{align}
Clearly, the two ellipsoids given above are apart by a factor of $\nfrac{1+\alpha}{\alpha} = \nfrac{1}{\alpha}+1$, which means the approximation factor increases by exactly $1$ as a result of this update. It now suffices to show that the ellipsoid described by (\ref{eq:irregular_outer}) contains $\mathsf{conv}\inparen{B_2^{d'} \cup \inbraces{\sqrt{1+2\alpha} \cdot \vv_{d'+1}}}$ and that the ellipsoid described by (\ref{eq:irregular_inner}) is contained by the cone whose base is $\alpha \cdot B_2^{d'}$ and whose apex is $\sqrt{1+2\alpha}\cdot\vv_{d'+1}$.

For the first part, it suffices to verify that every point $\vz \in Z$ and $\sqrt{1+2\alpha} \cdot \vv_{d'+1}$ is contained by (\ref{eq:irregular_outer}). We give both the calculations below, from which the result for (\ref{eq:irregular_outer}) follows.
\begin{align*}
    \vz \in Z :&\quad\quad \norm{\vz - \frac{\alpha}{\sqrt{1+2\alpha}} \cdot \vv_{d'+1}}_2 = \sqrt{\norm{\vz}_2^2 + \frac{\alpha^2}{1+2\alpha}} \le \frac{1+\alpha}{\sqrt{1+2\alpha}} \\
    \vz = \sqrt{1+2\alpha} \cdot \vv_{d'+1} :&\quad\quad \norm{\vz - \frac{\alpha}{\sqrt{1+2\alpha}} \cdot \vv_{d'+1}}_2 = \sqrt{1+2\alpha} - \frac{\alpha}{\sqrt{1+2\alpha}} = \frac{1+\alpha}{\sqrt{1+2\alpha}}
\end{align*}

We now analyze (\ref{eq:irregular_inner}). Our task is to show the below inclusion.
\begin{align*}
    &\inbraces{\vz \in \vspan{\vv_1,\dots,\vv_{d'+1}} \suchthat \norm{\vz - \frac{\alpha}{\sqrt{1+2\alpha}}\cdot \vv_{d'+1}}_2 \le \frac{\alpha}{\sqrt{1+2\alpha}}}\\
    \subseteq &\mathsf{conv}\inparen{\alpha \cdot \inbraces{\vz \in \vspan{\vv_1,\dots,\vv_d} \suchthat \norm{\vz}_2 \le 1} \cup \inbraces{\sqrt{1+2\alpha} \cdot \vv_{d'+1}}}
\end{align*}
Let $\vw$ be an arbitrarily chosen unit vector in $\vspan{\vv_1,\dots,\vv_{d'}}$. Observe that it is enough to show
\begin{align*}
    \inbraces{\vz \in \vspan{\vw, \vv_{d'+1}} \suchthat \norm{\vz - \frac{\alpha}{\sqrt{1+2\alpha}} \cdot \vv_{d'+1}}_2 \le \frac{\alpha}{\sqrt{1+2\alpha}}} \subseteq \mathsf{conv}\inparen{\pm \alpha \cdot \vw, \sqrt{1+2\alpha} \cdot \vv_{d'+1}}.
\end{align*}
Since the above is a two-dimensional problem and that $\ip{\vw, \vv_{d'+1}} = 0$, it is equivalent to show that the inradius of the triangle with vertices $(-\alpha, 0)$, $(\alpha, 0)$, and $(0, \sqrt{1+2\alpha})$ is $\nfrac{\alpha}{\sqrt{1+2\alpha}}$ and that its incenter is $(0, \nfrac{\alpha}{\sqrt{1+2\alpha}})$.

Recall that the inradius of a triangle can be written as $\nfrac{K}{s}$ where $K$ is the area of the triangle (in this case, $\alpha\sqrt{1+2\alpha}$) and $s$ is the semiperimeter of the triangle (in this case, $1+2\alpha$). This implies that the inradius is indeed $\nfrac{\alpha}{\sqrt{1+2\alpha}}$. Finally, since the triangle in question is isosceles with its apex being the $y$-axis, the $x$-coordinate of its incenter must be $0$. These observations imply that the incenter is $(0, \nfrac{\alpha}{\sqrt{1+2\alpha}})$.

This is sufficient for us to conclude the proof of Claim \ref{claim:irregular_update_easy}.
\end{proof}

We will now see that the analysis for the convenient update that we gave in Claim \ref{claim:irregular_update_easy} is nearly enough for us to fully analyze the irregular update step. See Claim \ref{claim:irregular_update_hard}, where we analyze the irregular update step in full generality (up to translating by $\vc_{t-1}$).

\begin{claim}
\label{claim:irregular_update_hard}
Let $Z \subset \R^d$ be a convex body such that $Z$ lies in a subspace $H$ of dimension $d' < d$. Let $\cE$ be an ellipsoid and let $0 < \alpha \le 1$ be such that
\begin{align*}
    \alpha \cdot \cE \subseteq Z \subseteq \cE.
\end{align*}
Let $\vz \notin H$. Then, there exists a center $\vc$ and an ellipsoid $\cE'$ such that
\begin{align*}
    \vc + \frac{1}{1+\nfrac{1}{\alpha}} \cdot \cE' \subseteq \mathsf{conv}\inparen{Z \cup \inbraces{\vz}} \subseteq \vc + \cE'.
\end{align*}
\end{claim}
\begin{proof}[Proof of Claim \ref{claim:irregular_update_hard}]
Recall that $\vv_1,\dots,\vv_{d'} \in \R^{d}$ are the unit vectors corresponding to the semiaxes of $\cE$; notice that these form a basis for $H$. Observe that $\vv_{d'+1}$ is a unit vector orthogonal to $\vv_1,\dots,\vv_{d'}$ such that $\vz$ can be expressed as $\sum_{i=1}^{d'+1} \vv_{i}\ip{\vv_i, \vz}$.

As stated in Algorithm \ref{alg:fully_online_approx}, let
\begin{align*}
    \mM \coloneqq \mI_d - \frac{1}{\ip{\vv_{d_t}',\vz_t}}\cdot\inparen{\vz_t - \sqrt{1+2\alpha} \cdot \vv_{d_t}'}(\vv_{d_t}')^T.
\end{align*}
We calculate
\begin{align*}
    \mM\vz_t = \vz_t - \frac{1}{\ip{\vv_{d_t}',\vz_t}}\cdot\inparen{\vz_t - \sqrt{1+2\alpha} \cdot \vv_{d_t}'}(\vv_{d_t}')^T\vz_t = \vz_t - \vz_t + \sqrt{1+2\alpha} \cdot \vv_{d_t}'  = \sqrt{1+2\alpha} \cdot \vv_{d_t}'.
\end{align*}
By the definition of $\mA_{t-1}$, we have
\begin{align*}
    \mA_{t-1}\mM\vz_t = \sqrt{1+2\alpha} \cdot \mA_{t-1}\vv_{d_t}' = \sqrt{1+2\alpha} \cdot \vv_{d_t}'.
\end{align*}
Next, for any $\vz \in Z$, we have $\vz \in H_{t-1}$. This means that
\begin{align*}
    \mM\vz = \vz - \frac{1}{\ip{\vv_{d_t}',\vz}}\cdot\inparen{\vz - \sqrt{1+2\alpha} \cdot \vv_{d_t}'}(\vv_{d_t}')^T\vz = \vz - 0 = \vz.
\end{align*}
By Claim \ref{claim:irregular_update_easy}, we know for
\begin{align*}
    \mA_{t-1}\mM\vc_t &= \frac{\alpha}{\sqrt{1+2\alpha}}\cdot\vv_{d_t}'\\
    \mA_{t-1}\mM\cE_t &= \inbraces{\vz \in \vspan{\vv_1,\dots,\vv_{d_{t-1}},\vv_{d_t}'} \suchthat \norm{\vz}_2 \le 1}\\
\end{align*}
that
\begin{align*}
    \mA_{t-1}\mM\vc_t + \frac{1}{1+\nfrac{1}{\alpha_{t-1}}} \cdot \mA_{t-1}\mM\cE_t \subseteq \mathsf{conv}\inparen{\mA_{t-1}\mM\cdot Z \cup \inbraces{\mA_{t-1}\mM\vz_t}} \subseteq \mA_{t-1}\mM\vc_t + \mA_{t-1}\mM\cE_t
\end{align*}
and, since $\mA_{t-1}\mM$ is invertible (owing to the invertibility of $\mA_{t-1}$ and $\mM$),
\begin{align*}
    \vc_t + \frac{1}{1+\nfrac{1}{\alpha_{t-1}}} \cdot \cE_t \subseteq \mathsf{conv}\inparen{Z \cup \inbraces{\vz_t}} \subseteq \vc_t + \cE_t.
\end{align*}
Finally, note that
\begin{align*}
    \vc_t &= \frac{\alpha}{\sqrt{1+2\alpha}} \cdot \mM^{-1}\mA_{t-1}^{-1}\vv_{d_t}' = \frac{\alpha}{1+2\alpha} \cdot \vz_t \\
    \cE_t &= \inbraces{\vz \in \vspan{\vz_1,\dots,\vz_t} \suchthat \norm{\mA_{t-1}\mM\vz}_2 \le 1}
\end{align*}
and then translate by $\vc_{t-1}$, which concludes the proof of Claim \ref{claim:irregular_update_hard}.
\end{proof}

We are now ready to prove \Cref{thm:main_two}.

\begin{proof}[Proof of \Cref{thm:main_two}]
Using \Cref{claim:irregular_update_hard}, we have that the ellipsoids maintain our desired invariants (Definition \ref{def:alg-invariant}) throughout the process. Hence, \Cref{alg:fully_online_approx} maintains an ellipsoidal approximation to $\conv{\inbraces{\vz_1,\dots,\vz_t}}$ for all $t$.

It remains to verify the approximation factor $\alpha_t$ of Algorithm \ref{alg:fully_online_approx}.

Consider a timestep $t$. For every $t' \leq t$, let $H_{t'} = \vspan{\vz_1,\dots,\vz_{t'}}$,  $r_{t'} = r(Z_{t'})$ be the inradius of $Z_{t'} = \conv{\vz_1,\dots, \vz_{t'}}$, and $R_{t'} = R(Z_{t'})$ be the circumradius of $Z_t$. 
Let ${\hat r} = \min_{t'\leq t} r_{t'}$. Consider the $d$-dimensional ellipsoid $T(\cE_{t'})$ which is exactly equal to $\cE_{t'}$ in the space $H_{t'}$ and whose remaining semiaxes orthogonal to $H_{t'}$ are equal and have length $\hat r$. Observe that for a regular update step $t'$ (with $d_{t'} = d_{t'-1}$), we have
$$\frac{\mathsf{vol}_{d_{t'}}(\cE_{t'})}{\mathsf{vol}_{d_t}(\cE_{t'-1})} = \frac{\mathsf{vol}_d(T(\cE_{t'}))}{\mathsf{vol}_d(T(\cE_{t'-1}))}.$$
Now applying the evolution condition (\ref{eq:overview_evolution}) to the update restricted to $H_{t'}$, we get 
$$
\frac{1}{\alpha_{t'}} - \frac{1}{\alpha_{t'-1}} \leq C\log \frac{\mathsf{vol}_{d_{t'}}(\cE_t)}{\mathsf{vol}_{d_{t'}}(\cE_{t'-1})} = C\log \frac{\mathsf{vol}_d(T(\cE_{t'}))}{\mathsf{vol}_d(T(\cE_{t'-1}))}.
$$
We have obtained the following upper bound on the approximation-factor increase:
\begin{align}
    \frac{1}{\alpha_{t'}} - \frac{1}{\alpha_{t'-1}} \leq \begin{cases} 1 & \text{if $t'$ is an irregular update step} \\ C\logv{\frac{\mathsf{vol}_{d}(T(\cE_{t'}))}{\mathsf{vol}_{d}(T(\cE_{t'-1}))}} & \text{otherwise}\end{cases}\label{eq:fully_online_invariant}
\end{align}
Let $T_{\mathsf{reg}}$ consist of all the timesteps $t' \le t$ where we perform a regular update. Then we have,
\[
    \frac{1}{\alpha_t} - \frac{1}{\alpha_0}= \alpha_0 + \sum_{t'=1}^t  \left(\frac{1}{\alpha_{t'}} - \frac{1}{\alpha_{t'-1}} \right)\leq d_t + C\sum_{t'\in T_{\mathsf{reg}}} \logv{\frac{\mathsf{vol}_{d}(T(\cE_{t'}))}{\mathsf{vol}_{d}(T(\cE_{t'-1}))}}.
\]
Now we show that $\logv{\frac{\mathsf{vol}_{d}(T(\cE_{t'}))}{\mathsf{vol}_{d_t}(T(\cE_{t'-1}))}} \geq 0$ for an irregular step: let $\sigma_1\geq\dots\geq\sigma_d$ and $\sigma_1'\geq \dots\geq\sigma_d'$ be the lengths of semi-axes of $T(\cE_{t'})$ and $T(\cE_{t'-1})$, respectively. Then $\sigma_i \geq \sigma'_i$ for $1\leq i \leq d_{t'}-1$, since $\cE_{t'-1} \subset \cE_{t'}$; $\sigma_{d_{t'}} \geq r_{t'} \geq \hat r = \sigma'_{d_{t'}}$; and $\sigma_{i} = \hat r = \sigma'_{i}$ for $i > d_{t'}$. Therefore,
\[\logv{\frac{\mathsf{vol}_{d}(T(\cE_{t'}))}{\mathsf{vol}_{d_t}(T(\cE_{t'-1}))}}=
\logv{\frac{\sigma_1\cdot \ldots \cdot \sigma_d}{\sigma_1'\cdot \ldots \cdot \sigma_d'}} \geq \log 1 = 0.
\]
Using this inequality and plugging in $\alpha_0 = 1$, we get
\begin{align*}
    \frac{1}{\alpha_t} &= 1 + d_t + C\sum_{t'\in T_{\mathsf{reg}}} \logv{\frac{\mathsf{vol}_{d}(T(\cE_{t'}))}{\mathsf{vol}_{d}(T(\cE_{t'-1}))}}
    \leq 1+ d_t + C\sum_{t'=1}^t \logv{\frac{\mathsf{vol}_{d}(T(\cE_{t'}))}{\mathsf{vol}_{d}(T(\cE_{t'-1}))}} \\
    &\lesssim d_t + \logv{\frac{\mathsf{vol}_{d}(T(\cE_{t}))}{\mathsf{vol}_{d}(T(\cE_{0}))}} 
    \lesssim d_t + \logv{\frac{(R_t/\alpha_t)^{d_t}{\hat r}^{d-d_t}}{{\hat r}^d}} \lesssim d_t +  d_t \logv{\frac{R_t}{\alpha_t\hat r}}.
\end{align*}
We conclude that
\begin{align*}
    \frac{1}{\alpha_t} \lesssim d_t + d_t \logv{\frac{R_t}{\hat r}} + d_t\log d_t.
\end{align*}
This concludes the proof of \Cref{thm:main_two}.
\end{proof}

\subsection{Aspect ratio-independent bounds and proof of \texorpdfstring{Theorem~\ref{thm:main_two_ints}}{Theorem 3}}

To prove \Cref{thm:main_two_ints}, we first establish \Cref{claim:irregular_volume_increase}.

\begin{claim}
\label{claim:irregular_volume_increase}
Let $t$ be an iteration corresponding to an irregular update step in \Cref{alg:fully_online_approx}. Then,
\begin{align*}
    \frac{\vol_{d_{t-1}}\inparen{B_2^{d_{t-1}}}}{\vol_{d_{t}}\inparen{B_2^{d_{t}}}} \cdot\frac{\vol_{d_t}\inparen{\cE_t}}{\vol_{d_t-1}\inparen{\cE_{t-1}}} \ge \frac{\norm{\vz_{t}^{\perp}}_2}{2}
\end{align*}
where $\norm{\vz_{t}^{\perp}}_2$ is the length of the component of $\vz_t$ in the orthogonal complement of $\vspan{\vz_1,\dots,\vz_{t-1}}$.
\end{claim}
\begin{proof}[Proof of \Cref{claim:irregular_volume_increase}]
By affine invariance, we can apply an affine transformation to map $\vz_t$ and $\cE_{t-1}$ to a convenient position. Hence, following the proof of \Cref{claim:irregular_update_easy}, without loss of generality, suppose we have $\cE_{t-1} = B_{2}^{d_{t-1}}$ and $\vz_t = \sqrt{1+2\alpha} \cdot \ve_{d_t}$. By \Cref{claim:irregular_update_easy}, the ellipsoid $\cE_t$ is a ball of radius $\nfrac{(1+\alpha)}{\sqrt{1+2\alpha}}$. Let $z \coloneqq \sqrt{1+2\alpha}$.
We now have
\begin{align*}
    \frac{\vol_{d_{t-1}}\inparen{B_2^{d_{t-1}}}}{\vol_{d_{t}}\inparen{B_2^{d_{t}}}} \cdot \frac{\vol_{d_t}\inparen{\cE_t}}{\vol_{d_{t-1}}\inparen{\cE_{t-1}}} = \inparen{\frac{1+\alpha}{\sqrt{1+2\alpha}}}^{d_t} \geq 1 > \frac{\|\vz_t\|_2}{2}
\end{align*}
since $\|\vz_t\|_2 =\sqrt{1+2\alpha} \leq \sqrt{3} < 2$. This concludes the Proof of \Cref{claim:irregular_volume_increase}.
\end{proof}

We will also need \Cref{claim:det_prod_identity}, which we take from \citet{gk10}.

\begin{claim}
\label{claim:det_prod_identity}
Let $\mM \in \R^{r \times d}$ have linearly independent rows $\vm_1,\dots,\vm_r$. Then,
\begin{align*}
    \prod_{i=1}^r \norm{\vm_i}_2 = \sqrt{\detv{\mM\mM^T}}.
\end{align*}
\end{claim}

We are now ready to prove \Cref{thm:main_two_ints}.

\begin{proof}[Proof of \Cref{thm:main_two_ints}]
Our approach is reminiscent of that used in the proof of Theorem 1.5 in \citet{woodruff2022high}.

By applying a translation to all points, we may assume without loss of generality that $\vz_1 = 0$. We will prove the guarantee for the last timestamp $t=n$ to simplify the notation. By replacing $n$ with $n'$, we can get a proof for any time stamp $t=n'$.

Let $S$ be the set of timestamps of irregular update steps excluding the first step. Since the update rule satisfies the evolution condition \eqref{eq:overview_evolution}, we have for all $t \notin S$ (recall that $d_t=d_{t-1}$ for $t\notin S$)
\begin{align*}
    \frac{\vol_{d_t}\inparen{\cE_{t}}}{\vol_{d_{t-1}}\inparen{\cE_{t-1}}} \ge \expv{\frac{1}{\alpha_t}-\frac{1}{\alpha_{t-1}}}.
\end{align*}
Next, by \Cref{claim:irregular_volume_increase}, we have for every irregular update step $t> 1$ 
\begin{align*}
    \frac{\vol_{d_{t-1}}\inparen{B_2^{d_{t-1}}}}{\vol_{d_{t}}\inparen{B_2^{d_{t}}}} \cdot\frac{\vol_{d_t}\inparen{\cE_t}}{\vol_{d_{t-1}}\inparen{\cE_{t-1}}} \ge \frac{\norm{\vz_{t}^{\perp}}_2}{2}.
\end{align*}
Here, we assume that $\vol_0(\{0\}) = 1$ and define $\norm{\vz_2^{\perp}}_2 = \norm{\vz_2}_2$.
Inductively combining the above for all $t > 1$ gives
\begin{align} 
    \vol_{d_n}\inparen{\cE_n} &\ge \prod_{t \notin S} \expv{\frac{1}{\alpha_t}-\frac{1}{\alpha_{t-1}}} \cdot \prod_{t \in S} \frac{\norm{\vz_t^{\perp}}_2}{2}\cdot \prod_{j=1}^{d_n} \frac{\vol_{j}\inparen{B_2^{j}}}{\vol_{j-1}\inparen{B_2^{j-1}}}\nonumber \\
    &= \prod_{t \notin S} \expv{\frac{1}{\alpha_t}-\frac{1}{\alpha_{t-1}}} \cdot \prod_{t \in S} \frac{\norm{\vz_t^{\perp}}_2}{2}\cdot \vol_{d_n}\inparen{B_2^{d_n}}\label{eq:int_potential}
\end{align}
Here we used that $\vol_0(\cE_{0}) = \vol_0(B_2^0) = 1$. Now invoking~\Cref{claim:det_prod_identity}, we get
\begin{align*}
    \prod_{t \in S} \frac{\norm{\vz_t^{\perp}}_2}{2} \ge 2^{-\abs{S}}\sqrt{\detv{\mZ\vert_S\mZ\vert_S^T}} \ge 2^{-\abs{S}} = 2^{-d_n},
\end{align*}
where we used that $\detv{\mZ\vert_S\mZ\vert_S^T} \ge 1$ because all the vectors $\vz_t$ have integer coordinates. Moreover, since all coordinates are at most $N$ in absolute value, all the vectors  $\vz_t$ have length at most $N\sqrt{d}$. Therefore, $\frac{\vol\inparen{\cE_n}}{\vol\inparen{B_2^{d_n}}} \le \inparen{N\sqrt{d}}^{d_n}$. We plug these bounds back into \eqref{eq:int_potential}, rearrange, and take the logarithm of both sides, yielding
\begin{align*}
    \sum_{t \notin S} \frac{1}{\alpha_t}-\frac{1}{\alpha_{t-1}} \lesssim d_n\logv{d N}.
\end{align*}
Finally, by \eqref{eq:fully_online_invariant}, we have $\frac{1}{\alpha_t} - \frac{1}{\alpha_{t-1}} = 1$ for every $t \in S$. Combining everything gives
\begin{align*}
    \sum_{t \le n} \frac{1}{\alpha_t}-\frac{1}{\alpha_{t-1}} \lesssim d_n\logv{d N} + |S| \lesssim d_n\logv{d N},
\end{align*}
thereby concluding the proof of \Cref{thm:main_two_ints}.
\end{proof}
\section{Forming Small Coresets for Convex Bodies (Proof of Theorem~\ref{thm:main_three})}
\label{sec:coreset}

In this section, we prove Theorem \ref{thm:main_three}. See Algorithm \ref{alg:ellipse_to_coreset}.

\begin{algorithm}[H]
\caption{Streaming coreset for convex hull}\label{alg:ellipse_to_coreset}
\begin{algorithmic}[1]
\State \textbf{Input:} Stream of points $\vz_t$; Update rule for Algorithm \ref{alg:fully_online_approx} $\cA$.
\State \textbf{Output:} Set $S \subseteq [n]$.
\For{$t = 1,\dots,n$}
    \State Receive $\vz_{t}$.
    \State Let $\cE_{\mathsf{test}} = \cA(\vc_{t-1}, \cE_{t-1}, \vz_t)$.
    \State Let $d_t = \mathsf{dim}\inparen{\vspan{\vz_1-\vc_{t-1},\dots,\vz_t-\vc_{t-1}}}$.
    \If{$d_t > d_{t-1}$ or $\frac{\mathsf{Vol}_{d_t}(\cE_{\mathsf{test}})}{\mathsf{Vol}_{d_t}(\cE_{t-1})} \ge e$}
        \State Let $\vc_t, \cE_{t} = \cA(\vc_{t-1},\cE_{t-1},\vz_{t})$.
        \State Update $S_t = S_{t-1} \cup \inbraces{\vz_{t}}$.
    \Else
        \State Let $\vc_t,\cE_{t}=\vc_{t-1},\cE_{t-1}$.
        \State Let $S_{t-1} = S_t$.
    \EndIf
\EndFor
\State \textbf{Output:} $S_n$
\end{algorithmic}
\end{algorithm}

For a sketch of the intuition and the argument we will use for the proof, see Section \ref{sec:overview_coreset}.

\begin{proof}[Proof of Theorem \ref{thm:main_three}]

We prove two properties of Algorithm \ref{alg:ellipse_to_coreset}. First, we show $\abs{S_t} \le O\inparen{d_t \cdot \logv{d_t \cdot \max_{t' \le t} \nfrac{R_t}{r_{t'}}}}$ and, further, $\abs{S_t} \le O\inparen{d_t \cdot \logv{dN}}$ if points $\vz_t$ have integer coordinates between $-N$ and $N$. Second, we  show that $\conv{Z\vert_{S_t}} \subseteq \conv{Z\vert_{[t]}} \subseteq O\inparen{d_t \cdot \logv{d_t \cdot \max_{t' \le t} \nfrac{R_t}{r_{t'}}}}\cdot \conv{Z\vert_{S_t}}$ and $\conv{Z\vert_{S_t}} \subseteq \conv{Z\vert_{[t]}} \subseteq O\inparen{d_t \cdot \logv{dN}} \cdot \conv{Z\vert_{S_t}}$.

\paragraph{Bounding $\abs{S_t}$.}

It is enough to count the number of steps $t$ for which we have $\frac{\mathsf{Vol}_{d_t}(\cE_{\mathsf{test}})}{\mathsf{Vol}_{d_t}(\cE_{t-1})} \ge e$.

It is easy to see that for all $t$, we have  $r(Z\vert_{[t]}) \cdot \inparen{B_2^d \cap \vspan{\vz_1-\vc_t,\dots,\vz_t-\vc_t}} \subseteq \vc_{t} + \cE_t$. Additionally, by the definition of $R(Z)$, we always have $Z\vert_{[t]} \subseteq R(Z) \cdot \inparen{B_2^d \cap \vspan{\vz_1-\vc_t,\dots,\vz_t-\vc_t}}$. These are enough to give volume lower and upper bounds in each step. Next, for each step in which we add an element to $S_{t-1}$ to obtain $S_t$, the volume must increase by a factor of $e$. It easily follows that the number of elements in $S_t$ satisfies
\begin{align*}
    \abs{S_t} \le \logv{\max_{t' \le t} \frac{\prod_{i=1}^{d_t} R(Z\vert_{[t]})}{\prod_{i=1}^{d_t} r(Z\vert_{[t']})}} = d_t\logv{\max_{t' \le t} \frac{R(Z\vert_{[t]})}{r(Z\vert_{[t']})}}.
\end{align*}
We now give an upper bound for the case when all coordinated of $\vz_t$ are integers not exceeding $N$ in absolute value. It is easy to see that the update rule in \Cref{alg:ellipse_to_coreset} exactly corresponds to the steps where we have
\begin{align*}
    \frac{1}{\alpha_{t}} - \frac{1}{\alpha_{t-1}} \gtrsim 1,
\end{align*}
and in the same way as in the proof of \Cref{thm:main_two_ints}, we have for all $t$ that
\begin{align*}
    \sum_{t \ge 1} \frac{1}{\alpha_{t}} - \frac{1}{\alpha_{t-1}} \lesssim d_t\logv{dN}.
\end{align*}
It therefore follows that $\abs{S} \lesssim d_t\logv{dN}$, as desired.

\paragraph{Bounding the distortion of the chosen points.}

Consider some iteration $t' \le t$. Without loss of generality, let $\vc_{t'-1} = 0$. Suppose $\vz_{t'}$ does not result in an update to $S_{t'-1}$. This implies that $\vz_{t'} \in 2e \cdot \cE_{t'-1}$. Next, observe that $0 \in \vc_t + \cE_t$. Putting these together, we have $\vz_{t'} \in \inparen{\vc_t + \cE_t} + 2e \cdot \cE_{t'-1}$. Since $\cA$ is monotone, we must have $2e \cdot \cE_{t'-1} \subseteq \vc_t + e\cdot\cE_t$; hence, we may write $\vz_{t'} \in \vc_t + \inparen{2e+1}\cE_t$.

The inner ellipsoid $\vc_t + \alpha_t \cdot \cE_t$ will still be an inner ellipsoid for the points determined by $S_t$. Stitching together all our inclusions, we have
\begin{align}
    \vc_t + \alpha_t \cdot \cE_t \subseteq Z\vert_{S_t} \subseteq Z \subseteq \vc_t + (2e+1)\cE_t \subseteq \frac{2e+1}{\alpha_t} \cdot Z\vert_{S_t}\label{eq:ellipsoid_coreset_inclusions}.
\end{align}
which means that
\begin{align*}
    Z\vert_{S_t} \subseteq Z \subseteq O\inparen{d_t \cdot \logv{d_t \cdot \max_{t' \le t} \frac{R(Z\vert_{([t] \cap S_t)})}{r(Z\vert_{([t'] \cap S_t)})}}} \cdot Z\vert_{S_t}.
\end{align*}
Notice that this is nearly what we want, except that the aspect ratio term is in terms of the subset body $Z\vert_{S_t}$. To obtain the final guarantee in terms of the aspect ratio of $Z\vert_{[t]}$, observe that the above guarantee readily implies that
\begin{align*}
    O\inparen{d_t \cdot \logv{d_t \cdot \max_{t' \le t} \frac{R(Z\vert_{([t] \cap S_t)})}{r(Z\vert_{([t'] \cap S_t)})}}} \le O\inparen{d_t \cdot \logv{d_t \cdot \max_{t' \le t} \frac{R(Z\vert_{[t]})}{r(Z\vert_{[t']})}}}.
\end{align*}
We now give the corresponding improvement when the $\vz_t$ are integer-valued. As before, \eqref{eq:ellipsoid_coreset_inclusions} holds. From this, we get
\begin{align*}
    Z\vert_{S_t} \subseteq Z \subseteq O\inparen{d_t \cdot \logv{dN}} \cdot Z\vert_{S_t},
\end{align*}
as desired. This concludes the proof of Theorem \ref{thm:main_three}.
\end{proof}


\section{Lower Bound}

In this section, we show \Cref{thm:main_four}.

\subsection{Lower bound adversary}

Our proof of \Cref{thm:main_four} constructs an adversary, which given a monotone algorithm \(\cA\) and \(\kappa \geq 1\),
constructs a sequence of points \(\vz_1, \ldots, \vz_n\) satisfying \(\kappa(\conv{\vz_1, \ldots, \vz_n}) \leq \kappa\) to witness that the algorithm does not produce an approximation better than \(\widetilde{\Omega}(d \log \kappa)\).
While by definition \(\kappa = \frac{R}{r}\), our construction keeps \(r = 1\) (notice that any lower bound construction must be scale-invariant),
and for simplicity we use \(R = \kappa\).

Let \(\vz^\Delta_1, \vz^\Delta_2, \ldots, \vz^\Delta_{d+1} \in \R^{d}\) be the \(d+1\) vertices of a regular simplex \(\Delta_{d}\) that circumscribes \(B_2^d\).
Our adversary is described in \cref{alg:lower_bound}.
It uses a first phase that feeds \(\cA\) the vertices of \(\Delta_d\),
then a second phase that repeatedly feeds \(\cA\) points at a constant distance from the previous ellipsoid.
Specifically, every new point \(\vz_t\) in the second phase is in \(\vc_{t-1} + 2 \cdot \cE_{t-1}\), i.e. its distance is 2 from \(\vc_{t-1}\) in the norm that is the gauge of \(\cE_{t-1}\).

\begin{algorithm}[h]
\caption{Lower bound adversary}\label{alg:lower_bound}
\textbf{Input}: Monotone algorithm \(\cA\), \(R \geq 1\)
\begin{algorithmic}[1]
\State \((\vc_0 + \cE_0, \alpha_0) = (0 + B_2^d, 1)\)
\Comment{Initialize to the unit ball}
\For{\(t \in \{1, 2, \ldots, d+1\}\)}
    \Comment{Phase I: feed \(\cA\) the vertices of a simplex}
    \State \((\vc_{t} + \cE_{t}, \alpha_{t}) = \cA(\vc_{t-1} + \cE_{t-1}, \alpha_{t-1}, \vz^\Delta_{t})\)
\EndFor
\State \(t \leftarrow d+2\)
\While{\(\vol(\cE_{t-1}) \leq \vol\inparen{\frac{R}{2} \cdot B_2^d}\)}\label{alg_line:adv_stopvol}
\Comment{Phase II: feed \(\cA\) points outside the previous ellipsoid}
    \State Let \(F_{t-1} = \partial(\vc_{t-1} + 2 \cE_{t-1}) \cap (R \cdot B_2^d)\)
    \If{\(F_{t-1} = \varnothing\)}
        \State \textbf{stop}\label{alg_line:adv_stopdisj}
    \EndIf
    \State Let arbitrary \(\vz_{t} \in F_{t-1}\)
    \State \((\vc_{t} + \cE_{t}, \alpha_{t}) = \cA(\vc_{t-1} + \cE_{t-1}, \alpha_{t-1}, \vz_{t})\)
    \State \(t \leftarrow t + 1\)
\EndWhile
\end{algorithmic}
\end{algorithm}

\begin{remark}
 This particular construction we give of the hard case is adaptive, meaning that the adversary's choice of points depend on the previous ellipsoids the algorithm outputs.
   However, this adversary can be made non-adaptive by taking an \(\varepsilon\)-net \(S\) of \(B_2^d\) for sufficiently small \(\varepsilon\),
   then feeding \(\cA\) the sequence of points in sets \(S, 2 \cdot S, 4 \cdot S, \ldots, 2^{\log_2 R - 1}, 2^{\log_2 R} \cdot S\).
   In consequence, this means that randomization on the part of the monotone algorithm does not help, unlike some other online settings.
\end{remark}

Let \(T\) be the largest value of \(t -1\) before the adversary halts.
We first show that the adversary only gives finitely many points before halting.
\begin{claim}
    \(T \leq O(d \log R)\)
\end{claim}
\begin{proof}
    We argue that the volume of \(\cE_{t}\) increases by at least a constant factor on each iteration.
    This is sufficient to bound the number of iterations by \(O(d \log R)\), 
    as \(\cE_0 = B_2^d\),
    and Line \ref{alg_line:adv_stopvol} is no longer true when the volume of \(\cE_t\) exceeds \(\inparen{\frac{R}{2}}^d \cdot \vol(B_2^d)\).

    We claim that for all \(t \geq d + 2\), \(\vol(\cE_{t}) \geq \frac{3}{2} \cdot \vol(\cE_{t-1})\).
    By applying a nonsingular affine transformation, we can assume without loss of generality that \(\cE_{t-1} = B_2^d\).
    With a further rotation, we can assume the newly received point is \(\vz_t = 2 \ve_1\).
    From monotonicity of \(\cA\) we must have that \(\vc_{t} + \cE_{t} \supseteq B_2^d \cup \{2 \ve_1\}\).
    Clearly every semi-axis of \(\cE_{t}\) must have length 
    at least 1 in order to contain \(\cE_{t-1}\).
    Observe that \(\cE_{t}\) must also contain the segment connecting
    \(-1 \ve_1\) and \(2 \ve_1\), and so at least one semi-axis must have length at least \(\frac{3}{2}\)
    (if not, the diameter of \(\cE_{t}\) would be strictly less than \(3\)).
    Hence as \(\frac{\vol(\cE_{t})}{\vol(B_2^d)}\) equals the product of the length of the semi-axes of \(\cE_{t}\), we have \(\vol(\cE_{t}) \geq \frac{3}{2} \vol(B_2^d)\).
\end{proof}

For the analysis we define quantities \(A_t, P_t\) associated with the sequence of ellipsoids for \(1 \leq t \leq T\):
\[A_t \defeq \frac{1}{\alpha_t}, \quad P_t \defeq \log \inparen{\frac{\vol(\cE_t)}{\vol(B_2^d)}}\]
By the monotonicity of \(\cA\), we have that \(A_t\) and \(P_t\) are both nondecreasing in \(t\).
We first observe that the adversary guarantees that the final volume of the ellipsoid output by \(\cA\) is large:
\begin{claim}\label{claim:lower_bound_vol}
    At the conclusion of \Cref{alg:lower_bound}'s execution, we have
   \[P_T  \geq d \log \frac{R}{2}\]
\end{claim}
\begin{proof}
    There are two ways that the adversary stops:
    if the condition in Line \ref{alg_line:adv_stopvol} is no longer true,
    or if Line \ref{alg_line:adv_stopdisj} is reached.
    If the former occurs, then we have \(\vol(\cE_{T}) > \vol(\frac{R}{2} \cdot B_2^d)\),
    and clearly then \(P_T \geq d \log\inparen{\frac{R}{2}}\).
    
    In the latter stopping condition, the algorithm halts at time \(T\) when \(\partial (c_{T} + 2 \cE_{T}) \cap R \cdot B_2^d = \varnothing\).
    The sets \(\partial (c_{T} + 2 \cE_{T})\) and \(R \cdot B_2^d\) can be disjoint in two cases:
    \(c_T + 2 \cE_T\) and \(R \cdot B_2^d\) are disjoint;
    or \(R \cdot B_2^d \subseteq c_T + 2 \cE_t\) with the boundaries of both ellipsoids disjoint.
    By the monotonicity of \(\cA\), we have \(1 \cdot B_2^d \subseteq c_T + 2 \cE_{T}\),
    and so eliminate the former case.
    But then \(\vol(2 \cdot \cE_T) \geq \vol(R \cdot B_2^d) = R^d \vol(B_2^d)\), and taking logarithms on both sides yields the claim.
\end{proof}

Now in contrast to the upper bound where we essentially gave an algorithm for which \(\frac{\Delta A}{\Delta P}\) was upper bounded by a constant,
here we will show a constant lower bound on the same quantity for any monotone algorithm.
\begin{claim}\label{claim:lb_step_main}
    There exists a constant \(C_{\ref{eqn:lb_step_main}} > 0\) such that if \(A_t \geq d\), we have
    \begin{equation}\label{eqn:lb_step_main}
    A_{t+1} - A_t \geq C_{\ref{eqn:lb_step_main}} (P_{t+1} - P_t)
    \end{equation}
\end{claim}

Observe that this lower bound requires \(A_t \geq d\),
hence necessitating a first phase using the simplex,
whose optimal roundings show tightness for John's theorem for general convex bodies.
In order to prove the lower bound we also need a second property, that \(A_{t}\) is large compared to \(P_t\).
\begin{claim}\label{claim:lb_simplex_symmetry}
    Let \(0 \leq \alpha \leq 1, \vc \in \R^d\), and \(\cE\) be an ellipsoid such that
    \[\vc + \alpha \cdot \cE \subseteq \Delta_d \subseteq \vc + \cE\]
    then we have: \begin{enumerate} 
        \item\label{item:lb_simplex_symmetry_1} \(\alpha \leq \frac{1}{d}\)
        \item \(\log\inparen{\frac{\vol(\cE)}{\vol(B_2^d)}} \leq O\inparen{\log(d) \cdot \frac{1}{\alpha}}\)
    \end{enumerate}
\end{claim}

With the statements of these claims in hand, we are ready to prove the lower bound.
\begin{proof}[Proof of \Cref{thm:main_four}]
 It is clear that \(\kappa(\conv{\vz_1, \ldots, \vz_T)} \leq R\), as for every \(1 \leq t \leq T\) the adversary guarantees \(1 \leq \|\vz_t\| \leq R\). So we focus on showing a lower bound on the quality of the approximation produced by \(\cA\).

As \(\cA\) is monotone, after the end of Phase I
we must have that \begin{equation}\label{eq:lb_ph1_approx}
\vc_{d+1} + \alpha_{d+1} \cdot \cE_{d+1} \subseteq \Delta_d \subseteq \vc_{d+1} + \cE_{d+1}
\end{equation}
Now because \(\cE_{d+1}\) satisfies the conditions of \Cref{claim:lb_simplex_symmetry},
we get using the definition \(A_{d+1} = \frac{1}{\alpha_{d+1}}\) that \(A_{t} \geq A_{d+1} \geq d\) for any \(t \geq d+1\).
Then we can apply \Cref{claim:lb_step_main} for every \(t \geq d + 1\) until termination of the algorithm:
\begin{align*}
A_{d+2} - A_{d+1} &\geq C_{\ref{eqn:lb_step_main}} \inparen{P_{d+2} - P_{d+1}} \\
A_{d+3} - A_{d+2} &\geq C_{\ref{eqn:lb_step_main}} \inparen{P_{d+3} - P_{d+2}} \\
&~\vdots \\
A_{T-1} - A_{T-2} &\geq C_{\ref{eqn:lb_step_main}} \inparen{P_{T-1} - P_{T-2}} \\
A_{T} - A_{T-1} &\geq C_{\ref{eqn:lb_step_main}} \inparen{P_{T} - P_{T-1}}
\end{align*}
Summing these inequalities, we have
\begin{align*}
    \sum_{t=d+1}^{T-1} A_{t+1} - A_{t} \geq C_{\ref{eqn:lb_step_main}} \inparen{\sum_{t=d+1}^{T-1} P_{t+1} - P_t}
\end{align*}
Both sides of this inequality are telescoping sums, so simplifying we get
\begin{equation}\label{eqn:thm_lower_bound_1}
A_T \geq A_{d+1} + C_{\ref{eqn:lb_step_main}} (P_T - P_{d+1})
\end{equation}

Again because we can apply \Cref{claim:lb_simplex_symmetry} for \(\cE_{d+1}\), we have \(P_{d+1} \leq O(\log(d) \cdot A_{d+1})\), which along with (\ref{eqn:thm_lower_bound_1}) yields
\[A_T \geq A_{d+1} + \Omega(P_T - \log(d) \cdot A_{d+1}) \geq  \Omega (P_T - \log(d) \cdot A_{d+1})\]
Thus we have \[
A_T \geq \Omega(\max(A_{d+1}, P_T - \log(d) \cdot A_{d+1})) \geq \Omega\inparen{\frac{P_T}{\log(d)}}
\]
and we get the desired bound using \Cref{claim:lower_bound_vol}.
\end{proof}

Our proof of \Cref{claim:lb_step_main}
relies on a symmetrization argument to a reduced case (essentially two-dimensional, like for our algorithms). We now define this reduced case, and related quantities.
\begin{definition}\label{defn:lb_reduced}
In the reduced case, the previous outer and inner ellipsoids are given by \(B_2^d, \alpha \cdot B_2^d\), and the received point is \(\vz = 2 \ve_1\).
The next outer and inner ellipsoids are given by \(c \cdot \ve_1 + \cE_{\mM}, c \cdot \ve_1 + \alpha' \cdot \cE_{\mM}\)
for \(c \in \R\), and \(\mM = \diag{a, b, b, \ldots, b, b}\) for \(a, b \geq 1\).
We let \(\Delta A = \frac{1}{\alpha'} - \frac{1}{\alpha}\) and \(\Delta P = \log \inparen{\frac{\vol(\cE_{\mM})}{\vol(B_2^d)}}\).
\end{definition}

Note that the update in this reduced case is monotone if
\(B_2^d \cup \{ 2 \ve_1 \} \subseteq c \cdot \ve_1 + \cE_{\mM}\)
and \(c \cdot \ve_1 + \alpha' \cdot \cE_{\mM} \subseteq \conv{(\alpha \cdot B_2^d) \cup \{ 2 \ve_1 \}}\).

Now we state the lower bound on \(\frac{\Delta A}{\Delta P}\) in this setting, which is established in \Cref{sec:lb_2d}.
It is exactly the bound of \Cref{claim:lb_step_main} in this special case.
\begin{claim}\label{claim:lb_reduced}
     In the reduced case, for any monotone update \(c \cdot \ve_1 + \cE_{\mM}, c \cdot \ve_1 + \alpha' \cdot \cE_{\mM}\) when \(\alpha' \leq \frac{1}{d}\) we have
    \begin{equation}
    \frac{\Delta A}{\Delta P} \geq C_{\ref{eqn:lb_step_main}}
     \end{equation}
\end{claim}

We now give the symmmetrization argument that shows that the above bound in the special case
implies the bound in the general case.
\begin{proof}[Proof of \Cref{claim:lb_step_main}]
    By the monotonicity of \(\cA\), we have \((\vc_t + \cE_t) \cup \{\vz_{t+1}\} \subseteq \vc_{t+1} + \cE_{t+1}\)
    and \(\vc_{t+1} + \alpha_{t+1} \cdot \cE_{t+1} \subseteq \conv{(\vc_t + \alpha_t \cdot \cE_t) \cup \{\vz_{t+1}\}}\).
    Without loss of generality we assume that \(\vc_t + \cE_t = B_2^d\) and \(\vz_{t+1} = 2 \cdot \ve_1\);
    we do this by applying a nonsingular affine transformation that maps \(\vc_t\) to the origin and \(\cE_t\) to \(B_2^d\), then apply a rotation
    that maps \(\vz_{t+1}\) to \(2 \cdot e_1\).
    Let \(\vc = \vc_{t+1}\), \(\cE = \cE_{t+1}\), and \(\alpha = \alpha_{t+1}\).
    Summarizing the conditions guaranteed by the monotonicity of \(\cA\), we have that \(B_2^d \cup \{ 2 \ve_1 \} \subseteq \vc + \cE\)
    and \(\vc + \alpha \cdot \cE \subseteq \conv{(\alpha_t \cdot B_2^d) \cup \{ 2 \ve_1 \}}\).

    To perform the reduction to the two-dimensional case, we apply a sequence of volume-preserving symmetrizations to the new inner and outer ellipsoids; these symmetrizations will also ensure that the update remains monotone.
    We will first apply two Steiner symmetrizations. The first of these Steiner symmetrizations transforms the ellipsoids so that their center lies on the \(\ve_1\)-axis. The second ensures that the ellipsoids have a semi-axis that is parallel to \(\ve_1\).
    Then, by a final symmetrization step we can transform the ellipsoids into bodies of revolution about \(\ve_1\).
    At that point it will suffice to consider the two-dimensional reduced case.

    Let \(\vc'\) be the projection of \(\vc\) onto the \(\ve_1\)-axis.
    The goal of the first symmetrization step is to transform \(\vc + \cE\) to \(\vc' + \cE'\) so that \(\vc'\) lies on the \(\ve_1\) axis. If \(\vc = \vc'\) then we do not need to do anything, otherwise we apply Steiner symmmetrization and consider \(S_{\vc - \vc'}(\vc+ \cE)\).
    By \Cref{claim:steiner_ellipsoid} this is still an ellipsoid,
    and we also have that the center of \(S_{\vc - \vc'}(\vc + \cE)\) is actually \(\vc'\);
    thus we may write \(S_{\vc - \vc'}(\vc + \cE) = \vc' + \cE'\) for some \(\cE'\).
    Further, we have that \(S_{\vc - \vc'}(\vc + \alpha \cdot \cE) = \vc' + \alpha \cdot \cE'\), as the Steiner symmetrization acts similarly on the scaled version of \(\cE\).
    To show that the update is still monotone, we observe that \(\vc' + \cE' = S_{\vc - \vc'}(\vc + \cE) \subseteq S_{\vc - \vc'}( \conv{(\alpha_t \cdot B_2^d) \cup \{ 2 \ve_1 \}})\).
    But by \Cref{claim:steiner_sym} and that \(\vc - \vc' \perp \ve_1\), \(\conv{(\alpha_t \cdot B_2^d) \cup \{ 2 \ve_1 \}}\) is invariant under the symmetrization \(S_{\vc - \vc'}\) and so we still have the inclusion \(\vc' + \alpha \cdot \cE' \subseteq \conv{(\alpha_t \cdot B_2^d) \cup \{ 2 \ve_1 \}}\).
    The `outer' inclusion
    \(B_2^d \cup \{ 2 \ve_1 \} \subseteq \vc' + \cE'\) follows in the same way.

    We now apply the second and final Steiner symmetrization. Let \(\vr\) be the rightmost point of \(\vc' + \cE'\) along \(\ve_1\); i.e. \(\vr = \arg\max_{\vr \in \vc' + \cE'} \inangle{\vr, \ve_1}\).
    Also let \(\vr'\) be its projection along the \(\ve_1\)-axis;
    if \(\vr = \vr'\) we again do not need to perform this symmetrization step, otherwise
    the Steiner symmetrization we apply is \(S_{\vr - \vr'}(\vc' + \cE')\).
    Since \(\vc'\) is at the midpoint of \(\vc' + \R (\vr - \vr')\) the center of the new ellipsoid is still \(\vc'\),
    so we may write \(S_{\vr - \vr'}(\vc' + \cE') = \vc' + \cE''\) and similarly
    \(S_{\vr - \vr'}(\vc' + \alpha \cdot \cE') = \vc' + \alpha \cdot \cE''\).
    Like for the previous symmetrization, the fact that \(\vr - \vr' \perp \ve_1\)
    means that both inclusions of the monotone update are preserved.
    Note finally that \(\vr'\) is the rightmost point of \(\vc' + \cE''\)
    and that the tangent plane of \(\vc + \cE''\) at \(\vr'\) is orthogonal to the line segment \(\overline{\vc' \vr'}\),
    so \(\ve_1\) is a semi-axis of \(\vc' + \cE''\).

    Our last transformation is a symmetrization of a different form, to turn \(\vc' + \cE''\) into a body of revolution.
    Let \(\sigma_1\) be the length of the semi-axis \(\ve_1\) of \(\cE''\), and \(\sigma_2, \ldots, \sigma_d\) be the lengths of the other semi-axes of \(\cE''\).
    We let \(\cE'''\) be the ellipsoid that has a \(\ve_1\) as a semi-axis of length \(\sigma_1\),
    and where every other semi-axis of \(\cE'''\) has length \(\sigma' \defeq \inparen{\prod_{i=2}^d \sigma_i}^{1/(d-1)}\).
    Clearly \(\vc' + \cE'''\) is now a body of revolution about \(\ve_1\) whose volume is the same as that of \(\vc' + \cE''\) (and hence also of \(\vc + \cE\)).
    Note that \(\vc' + \alpha \cE'''\) is also now a body of revolution.
    Since \(\sigma' \ge \min_{2 \le i \le d} \sigma_i\) we have \(B_2^d \cup \{ 2 \ve_1 \} \subseteq \vc' + \cE'''\), and correspondingly since \(\sigma' \le \max_{2 \le i \le d} \sigma_i\) we have that \(\vc' + \alpha \cdot \cE''' \subseteq \conv{(\alpha_t \cdot B_2^d) \cup \{ 2 \ve_1 \}}\).
    
    Clearly \(\vc' + \cE''', \vc + \alpha \cdot \cE'''\) now adhere to the reduced case of \Cref{defn:lb_reduced}.
    Since the update is monotone as well (and still \(\alpha \le 1/d\)) we can apply \Cref{claim:lb_reduced}.
    As \(\vol(\cE''') = \vol(\cE)\), this means we have
    \[
    A_{t + 1} - A_{t} \geq C_{\ref{eqn:lb_step_main}} \cdot (P_{t+1} - P_{t})
    \]
    as desired.
\end{proof}

\begin{proof}[Proof of \Cref{claim:lb_simplex_symmetry}]
For the first property, this is exactly the well-known fact that the best ellipsoidal rounding for the simplex \(\Delta_d\) (see e.g. \cite[Remark 1.1]{howard1997john}) has approximation factor \(d\).

Now we show the second property.
Again because the ball rounds the simplex \(\Delta_d\) with approximation factor \(d\), we have
\[\frac{1}{d} \cdot \Delta_d \subseteq B_2^d \subseteq \Delta_d\]
As a result of this, we have
\begin{align*}
    \log \inparen{\frac{\vol(\vc + \alpha \cdot \cE)}{\vol(B_2^d)}} &\le \log \inparen{\frac{\vol(\Delta)}{\vol(B_2^d)}} \\
    &\le \log \inparen{\frac{\vol(d \cdot B_2^d)}{\vol(B_2^d)}} \\
    &\le O(d \log d) 
\end{align*}
And so
\begin{align*}
\log \inparen{\frac{\vol(\cE)}{\vol(B_2^d)}} &= \log \inparen{\frac{\vol(\vc + \alpha \cdot \cE)}{\vol(B_2^d)}} + d \log\inparen{\frac{1}{\alpha}} \\
&\le O\inparen{d \log \inparen{\frac{1}{\alpha}}} & \text{as } \alpha \le \frac{1}{d}
\end{align*}
To establish the second property, it remains to show \(d \log(1/\alpha) \le O((1/\alpha) \log(d))\).
Observe that \(x \mapsto \frac{x}{\log x}\) is increasing for \(x \ge e\), so we have
\[\frac{d}{\log d} \le O\inparen{\frac{\nfrac{1}{\alpha}}{\log(\nfrac{1}{\alpha})}}\]
for all \(d \ge 2\) as \(d \le \nfrac{1}{\alpha}\).
Rearranging gives the desired inequality and thus the second property.
\end{proof}

\subsection{Analysis of the reduced case}
\label{sec:lb_2d}
In this section, we establish a lower bound on \(\frac{\Delta A}{\Delta P}\),
assuming we are in the `reduced case' defined in \Cref{defn:lb_reduced}.
Observe that in this case all relevant convex bodies \(\cE, \alpha \cE, c \cdot \ve_1 + \cE', c \cdot \ve_1 + \alpha' \cE', \conv{\alpha \cE \cup \{\vz\}}\) are all bodies of revolution about the \(x_1\)-axis, so to analyze the quantities involved we may instead look at any two-dimensional slice.
Accordingly we talk about the ellipses \(\cE, \alpha \cE, c + \cE', c + \alpha' \cE'\) in this two-dimensional slice,
where again \(\cE = B_2^2\), and \(c + \cE'\) and \(c + \alpha' \cE'\) are defined by
\begin{gather*}
c + \cE' = \inbraces{(x, y) \in \R^d \middle| \inparen{\frac{x-c}{a}}^2 + \inparen{\frac{y}{b}}^2 \leq 1}\\
c + \alpha' \cE' = \inbraces{(x, y) \in \R^d \middle| \inparen{\frac{x-c}{a}}^2 + \inparen{\frac{y}{b}}^2 \leq \alpha'^2}
\end{gather*}
for \(a, b > 0, c \in \R\).
We also use for convenience \(A = \frac{1}{\alpha}\) and \(A' = \frac{1}{\alpha'}\) so that \(\Delta A = A' - A\).
Also note in this reduced case we have by symmetry that
\[\Delta P = \log\inparen{\frac{\vol(c \cdot \ve_1 + \alpha' \cE')}{\vol(B_2^d)}} - \log\inparen{\frac{\vol(B_2^d)}{\vol(B_2^d)}} = \log(a \cdot b^{d-1})\]

Our lower bound in this reduced case is the following:
\begin{claim} There exists a fixed constant \(C_{\ref{eq:lb_cpct}} > 0\) such that 
    \[\frac{\Delta A}{\Delta P} \geq \min\inparen{C_{\ref{eq:lb_cpct}}, \frac{1}{10} \frac{A}{d}}\]    
\end{claim}

Clearly this claim yields \Cref{claim:lb_reduced} as a corollary, as by assumption in \Cref{claim:lb_reduced} we have \(A \geq d\) and so we get \(\frac{\Delta A}{\Delta P} \geq \Omega(1)\).

The inner ellipses in this lower bound, and some relevant points used in the proof of this claim, are depicted in \Cref{fig:lb_inner}.
\begin{figure}[h]
\centering
\includegraphics[height=6cm]{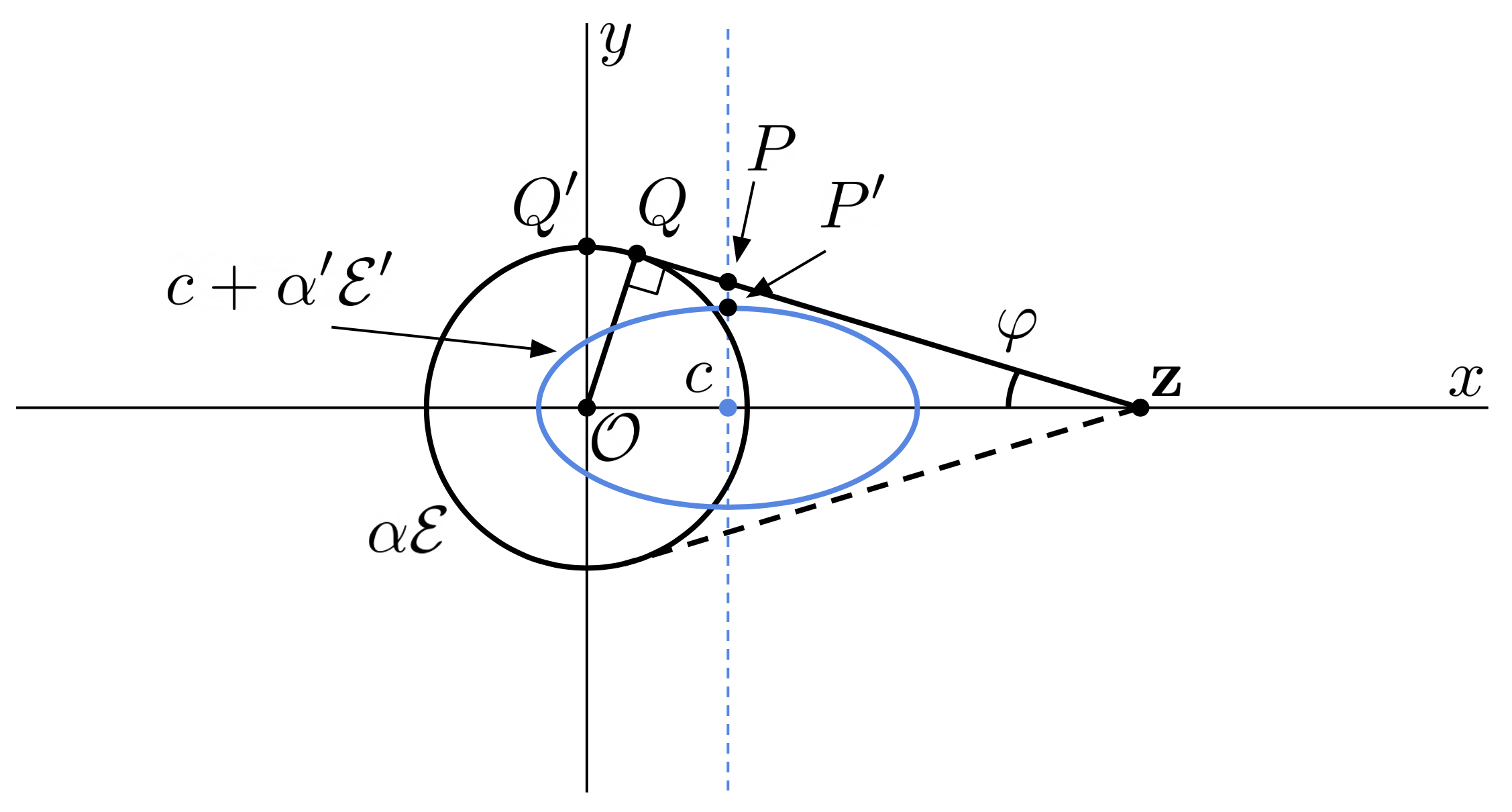}
\caption{The inner ellipses in the two-dimensional lower bound. \(\cO\) is the origin. The black solid circle is the previous inner ellipse \(\alpha \cE\),
and the blue solid circle is the next inner ellipse \(c + \alpha' \cE'\). The vertical dotted blue line \(x=c\) through the center \(c\) marks the location of the next inner ellipse on the \(x\)-axis. The new point is \(\vz = 2 \ve_1\),
and \(\overline{\vz Q}\) is one of the lines through \(\vz\) tangent to \(\alpha \cE\), with \(Q\) the point of tangency.
\(Q'\) is the intersection of \(\alpha \cE\) with the \(y\)-axis on the same side of the \(x\)-axis as \(Q\).
\(P'\) is the intersection of the line \(x=c\) with \(c + \alpha' \cE'\) on the same side as \(Q\),
and \(P\) is the intersection of this line with \(\overline{\vz Q}\).
We denote the angle \(\angle P \vz c\) with \(\varphi\).
}
\label{fig:lb_inner}
\end{figure}

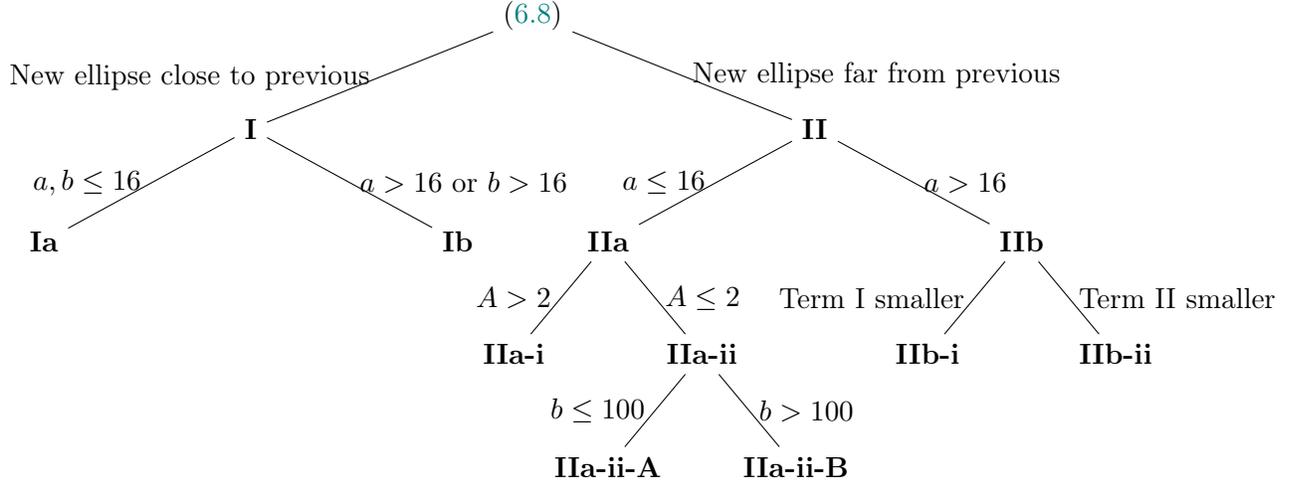
\begin{figure}[h]
\centering
\begin{tikzpicture}[
level 1/.style={sibling distance=75mm},
level 2/.style={sibling distance=55mm},
level 3/.style={sibling distance=25mm},
]
\tikzset{v/.append style={font=\bfseries}}
  \node {\eqref{eq:lb_5}}
    child {
        node[v] {I}
        child {
            node[v] {Ia}
            edge from parent
                node[left] {\(a, b \le 16\)}
        }
        child {
            node[v] {Ib}
            edge from parent
                node[right] {\(a > 16\) or \(b > 16\)}
        }
        edge from parent
            node[left] {New ellipse close to previous}
    }
    child {
        node[v] {II}
        child {
            node[v] {IIa}
            child {
                node[v] {IIa-i}
                edge from parent
                    node[left] {\(A > 2\)}
            }
            child {
                node[v] {IIa-ii}
                child {
                    node[v] {IIa-ii-A}
                    edge from parent
                        node[left] {\(b \le 100\)}
                }
                child {
                    node[v] {IIa-ii-B}
                    edge from parent
                        node[right] {\(b > 100\)}
                }
                edge from parent
                    node[right] {\(A \le 2\)}
            }
            edge from parent
                node[left] {\(a \le 16\)}
        }
        child {
            node[v] {IIb}
            child {
                node[v] {IIb-i}
                 edge from parent
                    node[left] {Term I smaller}
            }
            child {
                node[v] {IIb-ii}
                 edge from parent
                    node[right] {Term II smaller}
            }
            edge from parent
                    node[right] {\(a > 16\)}
        }
        edge from parent
            node[right] {New ellipse far from previous}
    };
\end{tikzpicture}
\caption{Tree of cases in the lower bound}
\label{fig:lb_cases}
\end{figure}

\begin{proof}
We establish this claim through a geometric argument that we break down by cases
(the logical tree of cases is visualized in \Cref{fig:lb_cases}).
First, as the new outer ellipse \(c + \cE'\) contains \(\cE = B_2^d\) we readily have that \(a, b \geq 1\).

As the rightmost point of the new outer ellipse must be to the right of \(\vz\), we have 
\begin{equation}\label{eq:lb_1}
c + a > 2
\end{equation}

As the leftmost point of the new inner ellipse must be to the right of the leftmost point of the previous inner ellipse, we have
\begin{equation}\label{eq:lb_4}
c \geq \alpha' a - \alpha
\end{equation}

\begin{claim}\label{claim:lb_1}
    We have \(\alpha' \cdot b \leq \alpha\), or equivalently \(A' \geq b \cdot A\).
\end{claim}
\begin{proof}
    The geometry of this fact is visualized in \Cref{fig:lb_inner}.
    We overload notation so that \(c\) will also denote the point \((c, 0)\), the center of the new inner ellipse.
We denote \(\cO\) as the origin.
Let \(\overline{\vz Q}\) be one of the lines through \(\vz\) and tangent to \(\alpha \cE\),
with \(Q\) the point of tangency  (the choice of which line is arbitrary, in the figure we choose the one whose intersection with \(\alpha \cE\) is above the \(x\)-axis).
We let \(P\) be the intersection of the vertical line through \((c, 0)\) with \(\overline{\vz Q}\),
and \(P'\) be the intersection of this line with the ellipse \(c + \alpha' \partial \cE'\)
on the same side of the \(x\)-axis as \(Q\).

    Observe that \(\alpha' b = \overline{c P'}\) as the vertical semi-axis of the ellipse \(c + \alpha' \cE\), and \(\alpha = \overline{\cO Q'}\).
    Due to the fact that \(c + \alpha' \cE' \subseteq \conv{\alpha \cE \cup \{\vz\}}\),
    the projection of both sets onto the \(y\)-axis satisfies the same inclusion, and this gives the desired inequality.
\end{proof}

Observe that as \(c + \cE'\) must contain both the points \((-1, 0)\) and \((0, 2)\), we have
\begin{equation}
\label{eq:lb_a_3_2}
a \geq \frac{3}{2}
\end{equation}

Observe that we can split \(\Delta P\) into two terms: \[\Delta P = \underbrace{(d-1) \log b}_{\text{I}} + \underbrace{\log a}_{\text{II}}\]
First we show that if term I is larger, then we have a constant lower bound on \(\frac{\Delta A}{\Delta P}\). 
\begin{claim}\label{claim:lb_term2}
    If \((d-1) \log b \geq \log a\), then \(\frac{\Delta A}{\Delta P} \geq \frac{1}{2} \frac{A} {d}\).
\end{claim}
\begin{proof}
    Under the assumption, we have \(\Delta P \leq 2 (d-1) \log b\).
    Combining this with \Cref{claim:lb_1}, we have
    \begin{align*}
    \frac{\Delta A}{\Delta P} &\geq \frac{A' - A}{2 (d-1) \log b} \\
    &\geq \frac{b \cdot A - A}{2 (d-1) \log b} \\
    &= \frac{b-1}{2 \log b} \frac{A}{d-1} \\
    &\geq \frac{1}{2} \frac{A}{d} \\
    \end{align*}
    where the last line uses that \(\frac{x-1}{2 \log (x)} > \frac{1}{2}\) when \(x > 1\).
\end{proof}

In light of \Cref{claim:lb_term2}, we can then assume in the sequel Term II is larger, meaning that
\begin{equation}\label{eq:lb_5}
\Delta P \leq 2 \log a
\end{equation}

\textbf{Case I }\textit{(New ellipse is close to the previous one)}. Assume that 
\begin{equation}\label{eq:lb_2}
    c + \alpha' \cdot a \leq \frac{11}{10} \alpha
\end{equation}
i.e. that the rightmost point of the new inner ellipse is to the left of \(\frac{11}{10}\).

\begin{claim}\label{claim:lb_close_1}
    In Case I, we have \(\Delta A \geq \frac{4}{11} A\).
\end{claim}
\begin{proof}
    We prove this by cases.
    First, if we assume that \(\alpha' \leq \frac{\alpha}{2}\),
    we get \(A' \geq 2 A\) and \(\Delta A \geq A\).

    In the second case, we have \(\alpha' > \frac{\alpha}{2}\).
    We first use this to show \(c > 0\).
    By (\ref{eq:lb_4}) and (\ref{eq:lb_1}) we have \(\frac{c + \alpha}{\alpha'} \geq a > 2 - c\), so
    \(c (1 + \alpha') > 2 \alpha' - a > 0\) and so \(c > 0\).

    Using (\ref{eq:lb_2}) and that \(c > 0\), we have \(\alpha' a \leq \frac{11}{10} \alpha\).
    Thus \(A' \geq \frac{10}{11} a A\).
    By (\ref{eq:lb_a_3_2}) we get \(A' \geq \frac{15}{11} A\), and finally
    \(\Delta A \geq \frac{4}{11} A\).
\end{proof}

\begin{claim}\label{claim:lb_2}
    In Case I, we have \(A' \geq \frac{10}{21} a \cdot A\).
\end{claim}
\begin{proof}
From (\ref{eq:lb_4}) we also get the weaker lower bound \(c \geq - \alpha\).
Combined with (\ref{eq:lb_2}), this gives \(\alpha' a \leq \frac{21}{10} \alpha\), which is equivalent to the desired inequality.
\end{proof}

We divide Case I into two sub-cases.

\textbf{Case Ia} \textit{(\(a, b \leq 16\))}.
First, assume that \(a, b \leq 16\).
Then \(\Delta P = \log(a \cdot b^{d-1}) \leq d \log(16)\),
and by \Cref{claim:lb_close_1} we get \[\frac{\Delta A}{\Delta P} \geq \frac{4}{11 \log(16)} \frac{A}{d} \geq \frac{1}{10} \frac{A}{d}\]

\textbf{Case Ib} \textit{(\(a > 16\) or \(b > 16\))}. Now assume that either \(a\) or \(b\) is greater than \(16\).

Combining \Cref{claim:lb_1} and \Cref{claim:lb_2} together, we get \(A'^2 \geq \frac{10}{21} ab A^2\), or \(A' \geq \sqrt{\frac{10}{21} a b} \cdot A\).
We have \(\Delta A = A' - A \geq \inparen{\sqrt{\frac{10}{21} a b} - 1} A\). As \(ab \geq 16\) we get \(\sqrt{ab} \geq 2 \sqrt{\frac{21}{10}}\), so
\(\Delta A = \inparen{\sqrt{\frac{10}{21} a b} - 1} A \geq \sqrt{\frac{5}{21} ab} \cdot A\).

Using the analytic inequality that \(\log x \leq \sqrt{x}\) for all \(x > 0\), we get
\(\Delta P = \log a + (d-1) \log b \leq \sqrt{a} + (d-1) \sqrt{b} \leq d \cdot \sqrt{ab}\).

Combining these inequalites for \(\Delta A\) and \(\Delta P\), we obtain
\[\frac{\Delta A}{\Delta P} \geq \sqrt{\frac{5}{21}} \frac{A}{d} \geq \frac{4}{10} \frac{A}{d}\]

\textbf{Case II }\textit{(New ellipse is far from the previous one)}. Assume that
\begin{equation}\label{eq:lb_3}
    c + \alpha' \cdot a > \frac{11}{10} \alpha
\end{equation}

\begin{claim}\label{claim:lb_3}
    We have 
    \begin{equation}\label{eq:lb_trig_b}
    \alpha' b \leq (2 - c) \cdot \frac{\frac{\alpha}{2}}{\sqrt{1 - \inparen{\frac{\alpha}{2}}^2}}
    \end{equation}
\end{claim}
\begin{proof}
Again, the proof of this claim is pictured in \Cref{fig:lb_inner},
where we construct the points in the same way as in the proof of \Cref{claim:lb_1}.
Let \(\angle P \vz c\) be denoted by \(\varphi\).
Note that the angle \(\angle P c \vz\) is a right angle, and so \(\tan \varphi = \frac{\overline{c P}}{\overline{c v}}\). The line \(\overline{c \vz}\) has length \(2 - c\), so we get \(\overline{c P} = (2 - c) \tan \varphi\).
The segment \(\overline{c P'}\), of length \(\alpha' b\), is contained within the segment \(\overline{c P}\), and so
\(\alpha' b \leq (2 - c ) \tan \varphi\).
    
Observe that the angle \(\angle \cO Q \vz\) is also a right angle.
Further, clearly the length of \(\overline{\cO Q}\) is \(\alpha\) and the length of \(\overline{0 v}\) is \(2\).
Since we have that \(\varphi\) is also the angle \(\angle Q \vz \cO\), we get \(\sin \varphi = \frac{\alpha}{2}\).
Now using the standard trigonometric identity that \(\tan \varphi = \frac{\sin \varphi}{\sqrt{1 - \sin^2 \varphi}}\) for $\varphi\in[-\pi/2,\pi/2]$,
we get the desired inequality.
\end{proof}

We split Case II into several sub-cases, as for Case I. 

\textbf{Case IIa} \textit{(\(a \leq 16\))}. First, we look at the case where \(a \leq 16\).

\textbf{Case IIa-i} \textit{(\(A \geq 2\))}. Assume \(A \geq 2\).
\begin{claim}\label{claim:lb_4}
    When \(A \geq 2\), we have \(\Delta A \geq \frac{1}{2}\).
\end{claim}
\begin{proof}
Adding (\ref{eq:lb_4}) and (\ref{eq:lb_3}) together gives \(c > \frac{1}{20} \alpha\).
Using this in (\ref{eq:lb_trig_b}) and rearranging using the definitions of \(A\) and \(A'\) yields
\[A' \geq b A \cdot \frac{1}{1-\frac{1}{40} \frac{1}{A}} \inparen{1 - \frac{1}{A^2}}\]
and therefore we get the inequality
\[\Delta A \geq b A \cdot \frac{1}{1-\frac{1}{40} \frac{1}{A}} \inparen{1 - \frac{1}{A^2}} - 1\]
and using \(b \geq 1\), we obtain
\[\Delta A \geq A \cdot \frac{1}{1-\frac{1}{40} \frac{1}{A}} \inparen{1 - \frac{1}{A^2}} - 1\]
To prove the claim, it suffices to show the right hand side exceeds \(\frac{1}{2}\) when \(A \geq 2\). Upon rearranging, this is equivalent to the inequality \(A - \frac{77}{80} \frac{1}{A} \geq \frac{3}{2}\) when \(A \geq 2\).
\end{proof}
Combining the assumption that \(a \leq 16\) with (\ref{eq:lb_5}) and \Cref{claim:lb_4} yields \(\frac{\Delta A}{\Delta P} \geq \frac{1}{4 \log 16} \geq \frac{1}{20}\).

\textbf{Case IIa-ii} \textit{(\(A \leq 2\))}. Next, we look at the other case where \(A \leq 2\).

\textbf{Case IIa-ii-A} \textit{(\(b \leq 100\))}. Now we look at the case where \(b \leq 100\).
\begin{claim}
    If \(a \leq 16, A \leq 2, b \leq 100\), then there is \(C_{\ref{eq:lb_cpct}} > 0\) such that
    \begin{equation}\label{eq:lb_cpct}
        \frac{\Delta A}{\Delta P} \geq C_{\ref{eq:lb_cpct}}
    \end{equation}
\end{claim}
\begin{proof}
    We show this by a compactness argument.
    By (\ref{eq:lb_5}) and the assumption that \(a \leq 16\) we have \(\frac{\Delta A}{\Delta P} \geq \frac{\Delta A}{2 \log 16}\).
    Now, observe that next outer and inner ellipsoids \(c + \cE'\) and \(c + \alpha' \cE'\) are fully determined by the parameters
    \(a, b, c, A, A'\).
    Further, we assume without loss of generality that \(A'\) is a function of the other parameters.
    This is because when \(A'\) is decreased as much as possible while preserving the monotonicity of the update, \(\Delta A = A' - A\) only decreases.
    To show a lower bound on \(\Delta A\) it then suffices to only do so in this hardest case.
    
    Note that we have \(1 \leq a \leq 16, 1 \leq b \leq 100, -1 \leq c \leq 2\), and \(1 \leq A \leq 2\); thus all the parameters defining the next inner and outer ellipsoids are bounded.
    Observe that \(\Delta A\) is a continuous function of these parameters, and as a continuous function of a compact set it attains its minimum. Finally, we argue that it is impossible for the minimum of \(\Delta A\) to be zero, and so the minimum is some strictly positive constant \(\frac{C_{\ref{eq:lb_cpct}}}{2 \log 16}\), which suffices to prove the claim.
    
    The following argument only concerns the inner ellipsoids, and can be pictured in \Cref{fig:lb_inner}. If \(c = 0\), then as the leftmost point of \(c + \alpha' \cE'\) must be to the right of the leftmost point of \(\alpha \cE'\),
    we have \(\alpha \geq \alpha' a\). But by (\ref{eq:lb_a_3_2}), we have \(\alpha \geq \frac{3}{2} \alpha'\), so \(\alpha' < \alpha\) and \(\Delta A > 0\).
    If \(c \neq 0\) then the vertical semi-axis of \(c + \alpha' \cE'\) must have length strictly less than \(\alpha\),
    and so \(\alpha > \alpha' b\). As \(b \geq 1\), this also gives \(\alpha > \alpha'\) and again \(\Delta A > 0\).
\end{proof}

\textbf{Case IIa-ii-B} \textit{(\(b > 100\))}.
Observe that the horizontal axis of the next inner ellipsoid \(c + \alpha' \cE'\) must be contained within the interval \([-\alpha, 2]\), thus we have that \(2 + \alpha \geq 2 b \alpha' > 200 \alpha'\).
Using the definitions of \(A, A'\) this is equivalent to \(2 + \frac{1}{A} > \frac{200}{A'}\), i.e. \(A' > \frac{200}{1 + \frac{1}{A}}\). As \(A \geq 1\) we get \(2 + \frac{1}{A} \leq 3\), and so \(A' \geq \frac{200}{3}\).

As \(A \leq 2\), we obtain \(\Delta A = A' - A \geq \frac{194}{3}\).
Now by \eqref{eq:lb_5} and that \(a \leq 16\) we have 
\[\frac{\Delta A}{\Delta P} \geq \frac{\Delta A}{2 \log 16} \geq \frac{194}{6 \log 16} \geq 11\]

\textbf{Case IIb} \textit{(\(a > 16\))}. Now, we examine the case where \(a > 16\).
Scaling (\ref{eq:lb_4}) by \(\frac{11}{10}\) and adding it to (\ref{eq:lb_3}), we have \(\frac{21}{10} c - \frac{1}{10} \alpha' a \geq 0\), i.e. \(c > \frac{1}{21} \alpha' a\).
Using this in (\ref{eq:lb_trig_b}), using \(b \geq 1\), using the definitions of \(A\) and \(A'\) and rearranging, we obtain
\[A' \geq A \cdot \frac{1}{1 - \frac{1}{42} \cdot \frac{a}{A'}} \cdot \sqrt{1 - \inparen{\frac{\alpha}{2}}^2}\]

Using the inequalities \(\sqrt{1-\inparen{\frac{x}{2}}^2} \geq 1 - \frac{x^2}{7}\) for \(0 \leq x \leq 1\) and \(\frac{1}{1-x} \geq 1 + x\) for \(0 \leq x \leq 1\), we have
\[A' \geq A \inparen{1- \frac{\alpha^2}{7}} \inparen{1 + \frac{1}{42} \frac{a}{A'}}\]
and thus
\(A'^2 - A \cdot A' \inparen{1 - \frac{\alpha^2}{7}} - \frac{1}{42}  \inparen{1 - \frac{\alpha^2}{7}} \geq 0\),
which implies by the quadratic formula that \begin{align*}
A' &\geq \frac{A \inparen{1 - \frac{\alpha^2}{7}} + \sqrt{A^2 \inparen{1 - \frac{\alpha^2}{7}} + \frac{4}{42} a A \inparen{1 - \frac{\alpha^2}{7}}}}{2} \\
&= A \inparen{1 - \frac{\alpha^2}{7}} \cdot \frac{1 + \sqrt{1 + \frac{2}{21} \frac{a}{A \inparen{1 - \frac{\alpha^2}{7}}}}}{2} \\
&= A \inparen{1 - \frac{\alpha^2}{7}} \cdot \inparen{1 + \frac{\sqrt{1 + \frac{2}{21} \frac{a}{A \inparen{1 - \frac{\alpha^2}{7}}}} - 1}{2}}
\end{align*}
Using the inequality \(\sqrt{1+x}-1 \geq \frac{2}{5} \min(x, \sqrt{x})\) for all \(x \geq 0\), we get that
\begin{equation}\label{eq:lb_2b}
A' \geq A \inparen{1 - \frac{\alpha^2}{7}} \inparen{1+ \frac{1}{5}\min\inparen{\smash[b]{\underbrace{\frac{2}{21} \frac{a}{A \inparen{1 - \frac{\alpha^2}{7}}}}_{\text{I}}, \underbrace{\sqrt{\frac{2}{21} \frac{a}{A \inparen{1 - \frac{\alpha^2}{7}}}}}_{\text{II}}}}}
\end{equation}
\\ 

To finish this case, we show the lower bound in the case where either term in the \(\min\) of (\ref{eq:lb_2b}) is the smaller term.

\textbf{Case IIb-i} \textit{(Term I in (\ref{eq:lb_2b}) is smaller)}. In this case, (\ref{eq:lb_2b}) is equivalent to
\begin{align*}
    A' &\geq A \inparen{1 - \frac{\alpha^2}{7}} + \frac{2}{105} a
\end{align*}
Using the definition of \(A\), we have 
\(A' \geq A - \frac{1}{7A} + \frac{2}{105} a\), and so \(\Delta A \geq - \frac{1}{7A} + \frac{2}{105} a\).
As \(A \geq 1\), we have \(\Delta A \geq \frac{2}{105} a - \frac{1}{7}\).

Now by \eqref{eq:lb_5}, we get
\[\frac{\Delta A}{\Delta P} \geq \frac{\frac{2}{105} a - \frac{1}{7}}{2 \log a}\]
Now, we complete this case by noticing the right hand side is at least \(\frac{1}{35}\) when \(a > 16\).

\textbf{Case IIb-ii} \textit{(Term II in (\ref{eq:lb_2b}) is smaller)}.  In this case, (\ref{eq:lb_2b}) is equivalent to
\begin{equation}\label{eq:lb_c2b2_1}
    A' \geq A \inparen{1 - \frac{\alpha^2}{7}} + \sqrt{\frac{2}{525} a A \inparen{1 - \frac{\alpha^2}{7}}}
\end{equation}
Using the definition of \(A\) and that \(A \geq 1\), we have \(A \inparen{1 - \frac{\alpha^2}{7}} = A - \frac{1}{7A} \geq - \frac{1}{7}\). Using this and the definition of \(\Delta A\) in \eqref{eq:lb_c2b2_1}, we have \[
\Delta A \geq - \frac{1}{7} + \sqrt{\frac{2}{525} a A \inparen{1 - \frac{\alpha^2}{7}}}
\]
Further, as \(0 \leq \alpha \leq 1\) we get \(A \inparen{1 - \frac{\alpha^2}{7}} = \frac{1}{\alpha} - \frac{\alpha}{7} \geq \frac{6}{7}\), so
\[\Delta A \geq - \frac{1}{7} + \sqrt{\frac{12}{3675} a}
\]
Now by \eqref{eq:lb_5}, we get
\[\frac{\Delta A}{\Delta P} \geq \frac{- \frac{1}{7} + \sqrt{\frac{12}{3675} a}}{2 \log a}\]
We finish with the fact that the right hand side is at least \(\frac{1}{65}\) when \(a > 16\).
\end{proof}

\section{Details of Analysis in \texorpdfstring{\Cref{sec:two_dim_update}}{Section~\ref{sec:two_dim_update}}}
Here, we give the details for the outstanding claims in \Cref{sec:two_dim_update}, where we also use the notation from that section (specifically, the definition of parameters in (\ref{eqn:update_params})). We first give some well-known bounds on \(e^x\).

\begin{claim}\label{claim:wk_e} Well-known inequalities on \(e^x\).
\begin{enumerate}
\item \label{item:wk_1} \(1 + x \leq e^x\) for all \(x \in \R\)
\item \label{item:wk_2} \(1 + x + \frac{x^2}{2} \leq e^x\) for \(x \geq 0\)
\end{enumerate}
\end{claim}

We will also use a more specialized upper bound on \(e^x\).
\begin{claim}\label{claim:custom_exp_upper_bound} For \(0 \leq x \leq \frac{4}{3}\), we have
\(e^x \leq 1 + x + \frac{x^2}{2} + \frac{x^3}{4}\).
\end{claim}
\begin{proof}
Using the Taylor series for \(e^x\) about \(0\),
we get that \[
\left(1 + x + \frac{x^2}{2} + \frac{x^3}{4}\right) - e^x = \frac{x^3}{12} - \sum_{k=4}^\infty \frac{x^k}{k!}
= x^3 \left(\frac{1}{12} - \sum_{k=4}^\infty \frac{x^{k-3}}{k!}\right)\]
Clearly \(x^3 \geq 0\) for \(x \geq 0\), so it remains to show \(\frac{1}{12} - \sum_{k=4}^\infty \frac{x^{k-3}}{k!} \geq 0\) for \(0 \leq x \leq \frac{4}{3}\). 
\(\sum_{k=4}^\infty \frac{x^{k-3}}{k!}\) is increasing (the derivative is clearly positive when \(x \ge 0\)),
and we finish by noting: \[\frac{1}{12} - \left.\sum_{k=4}^\infty \frac{x^{k-3}}{k!} \right|_{x = \frac{4}{3}} = \left.\frac{1 + x + \frac{x^2}{2}+\frac{x^3}{4} - e^x}{x^3} \right|_{x = \frac{4}{3}} > 0\]
\end{proof}

Now we show some facts used in \Cref{claim:pty_outer}, which \Cref{claim:update_step_outer} reduces to.
The proof of \Cref{claim:pty_outer} will reduce to the following analytic inequality.
\begin{claim}\label{claim:outer_t_ineq}
  For all \(\gamma \geq 0\),
  \[\frac{(e^\gamma - 1)^2}{e^{2\gamma} - (1 + \frac{\gamma}{4})^2} \leq \frac{3}{2} \gamma\]
\end{claim}
\begin{proof}
For the numerator of the left hand side, we have \((e^\gamma - 1)^2 = e^{2\gamma} - 2 e^\gamma + 1 \leq e^{2\gamma} - 2\gamma - 1\) using \Cref{claim:wk_e}-(\ref{item:wk_1}), \(1 + x \leq e^x\).
Further, \(e^\gamma \geq 1 + \frac{\gamma}{4}\) implies \(e^{2 \gamma} - (1 + \frac{\gamma}{4})^2 \geq 0\), so after multiplying both sides by \(e^{2 \gamma} - (1 + \frac{\gamma}{4})^2\) and rearranging it suffices to show
\[\frac{3}{2} \gamma \left(1 + \frac{\gamma}{4}\right)^2 - 1 - 2\gamma \leq \left(\frac{3}{2} \gamma - 1\right) e^{2 \gamma}\]
We split this into two cases, based on the value of \(\gamma\).
If \(\gamma \geq \frac{2}{3}\), then the right hand side is at least \(\left(\frac{3}{2} \gamma - 1\right) ( 1 + 2 \gamma + 2 \gamma^2)\) using \Cref{claim:wk_e}-(\ref{item:wk_2}), \(1 + x + \frac{x^2}{2} \leq e^x\),
so it is sufficient to show \(\left(\frac{3}{2} \gamma - 1\right) ( 1 + 2 \gamma + 2 \gamma^2) \geq \frac{3}{2} \gamma \left(1 + \frac{\gamma}{4}\right)^2 - 1 - 2 \gamma\).
Expanding both sides, this is equivalent to showing \(\frac{\gamma^2}{4} + \frac{93}{32} \gamma^3 \geq 0\), which is clearly true for \(\gamma \geq 0\).

If \(\gamma < \frac{2}{3}\), then we use \Cref{claim:custom_exp_upper_bound} to lower bound the right hand side with \(\left(\frac{3}{2} \gamma - 1 \right) (1 + 2 \gamma + 2 \gamma^2 + 2 \gamma ^3)\),
so it is sufficient to show  \(\left(\frac{3}{2} \gamma - 1\right) ( 1 + 2 \gamma + 2 \gamma^2 + 2 \gamma^3) \geq \frac{3}{2} \gamma \left(1 + \frac{\gamma}{4}\right)^2 - 1 - 2\gamma\).
Similar to before, after expanding both sides this is equivalent to showing \(\frac{\gamma^2}{4} + \frac{93}{32} \gamma^3 + \frac{3}{2} \gamma^4 \geq 0\), which is true for \(\gamma \geq 0\).
\end{proof}

We also show some relations between the parameters in the update step.
Recall that we assumed \(\alpha \leq \frac{1}{2}\).
\begin{claim}\label{claim:outer_int}
We have
\begin{enumerate}
\item \(b \leq 1 + \frac{\gamma}{4}\) \label{item:outer_int_1}
\item \(b \leq a\) \label{item:outer_int_2}
\item \(\frac{(a-1)^2}{a^2-b^2} \leq 1\) \label{item:outer_int_3}
\item \(b^2 \geq 1 + \alpha - \alpha'\) \label{item:outer_int_4}
\end{enumerate}
\end{claim}
\begin{proof}
We start by showing (\ref{item:outer_int_1}). As \(\frac{1}{\alpha'} = \frac{1}{\alpha} + 2\gamma\), we have \(\alpha = \alpha' + 2 \gamma \alpha \alpha'\).
Thus \(b = 1 + \gamma \alpha \alpha' \leq 1 + \gamma \alpha^2 \leq 1 + \gamma/4\) as \(\alpha' \leq \alpha \leq \frac{1}{2}\).

For (\ref{item:outer_int_2}), observe that \(b \leq 1 + \frac{\gamma}{4} \leq 1 + \gamma \leq e^\gamma = a\) using \Cref{claim:wk_e}-(\ref{item:wk_1}), \(1+x \leq e^x\), so \(a \geq b\).

To show (\ref{item:outer_int_3}), we first argue it is sufficient to show \(1 + b^2 \leq 2 a\). As a consequence \(-2a + 1 \leq - b^2\), so \((a-1)^2 \leq a^2 - b^2\).
Because \(b \geq 1\) by \Cref{claim:update_params}-(\ref{item:claim_update_params_2}), from (\ref{item:outer_int_2}) we can say that \(a^2 - b^2 \geq 0\), so that \(\frac{(a-1)^2}{a^2 - b^2} \leq 1\).

Now to show \(1 + b^2 \leq 2a\), we write as a series in terms of \(\gamma\).
On the left hand side using (\ref{item:outer_int_1}), we have \(1 + b^2 \leq 1 + \left(1 + \frac{\gamma}{4}\right)^2 = 2 + \frac{\gamma}{2} + \frac{\gamma^2}{4}\).
Further, by \Cref{claim:wk_e}-(\ref{item:wk_2}), \(e^x \geq 1 + x + \frac{x^2}{2}\), we have that \(2a \geq 2 + 2\gamma + \gamma^2\).
Clearly \(2 + \frac{\gamma}{2} + \frac{\gamma^2}{4} \leq 2 + 2\gamma + \gamma^2\) when \(\gamma \geq 0\), so we are finished.

For (\ref{item:outer_int_4}), we have by definition that \(b^2 = 1 + \alpha - \alpha' + \frac{(\alpha - \alpha')^2}{4}\), so \(b^2 \geq 1 + \alpha - \alpha'\).
\end{proof}

As the proof of \Cref{claim:update_step_outer} shows, that claim reduces to the following inequality.
\begin{claim}\label{claim:pty_outer}
    \begin{equation}\label{eqn:pty_outer_eqn} 
        c^2 \leq \frac{b^2 - 1}{b^2} \cdot (a^2 - b^2)
    \end{equation}
\end{claim}
\begin{proof}
We first upper bound \(c\) to reduce the number of variables in (\ref{eqn:pty_outer_eqn}).
As \(b = 1 + \frac{\alpha - \alpha'}{2}\), we have \(2 (b - 1) = \alpha - \alpha'\) and so
\(\alpha = \alpha' + 2(b-1)\).
Thus \[c = - \alpha + \alpha' \cdot a = -(\alpha' + 2(b-1)) + \alpha' \cdot a = \alpha' \cdot (a - 1) + 2(1-b) a\]
As \(b \geq 1\) by \Cref{claim:update_params}-(\ref{item:claim_update_params_2}), we have that \(2(1-b) \leq 0\) and therefore
\begin{equation}\label{eqn:pty_outer_uc} 
    c \leq \alpha' \cdot (a-1) \nonumber
\end{equation}
Using this in (\ref{eqn:pty_outer_eqn}), it suffices to show \(\alpha'^2 (a-1)^2 \leq \frac{b^2 - 1}{b^2} \cdot (a^2 - b^2)\),
which rearranges to
\begin{equation}\label{eqn:pty_outer_eqn1} 
\frac{b^2 - 1}{\alpha'^2} \geq \frac{(a-1)^2 b^2}{a^2 - b^2} \nonumber
\end{equation}

Using \Cref{claim:outer_int}-(\ref{item:outer_int_4}), this reduces to \begin{equation}\label{eqn:pty_outer_eqn2}
    \frac{\alpha - \alpha'}{\alpha'^2} \geq \frac{(a-1)^2 b^2}{a^2 - b^2}
\end{equation}
The left hand side of (\ref{eqn:pty_outer_eqn2}) equals \(\frac{1}{\alpha'} \left(\frac{\alpha}{\alpha'} - 1\right)\).
    Because \(\frac{1}{\alpha'} = \frac{1}{\alpha} + 2 \gamma\), we have \(\frac{\alpha}{\alpha'} - 1 = 2 \gamma \alpha\),
so \(\frac{1}{\alpha'}\left(\frac{\alpha}{\alpha'} - 1\right) = 2\gamma \cdot \frac{\alpha}{\alpha'} = 2\gamma \cdot (1 + 2 \gamma \alpha)\).
So it is sufficient to show
\begin{equation}\label{eqn:pty_outer_eqn3}
2\gamma (1 + 2 \gamma \alpha)\geq \frac{(a-1)^2 b^2}{a^2 - b^2}
\end{equation}

Now we will eliminate the other variables in this inequality to transform it
into a statement involving only \(\gamma\).
We have
\begin{align*}
\frac{(a-1)^2 b^2}{a^2 - b^2} &= \frac{(a-1)^2}{a^2 - b^2} \left(1 + 2 \alpha \alpha' \gamma + \alpha'^2 \alpha^2 \gamma^2\right)\\
&\leq \frac{(a-1)^2}{a^2 - b^2} + \frac{\gamma}{2} + \alpha \frac{\gamma^2}{8}
\end{align*}
where the first line uses that \(b = 1 + \alpha \alpha' \gamma\), and
the second line inequality follows from \Cref{claim:outer_int}-(\ref{item:outer_int_3}) and the fact that \(\alpha' \leq \alpha \leq \frac{1}{2}\).
Thus we can reduce (\ref{eqn:pty_outer_eqn3}) to \(\frac{(a-1)^2}{a^2-b^2} \leq \frac{3}{2} \gamma + \frac{31}{8} \alpha \gamma^2\), or further to

\begin{equation}\label{eqn:pty_outer_eqn4}
\frac{(a-1)^2}{a^2 - b^2} \le \frac{3}{2} \gamma
\end{equation}

Using \Cref{claim:outer_int}-(\ref{item:outer_int_1}) and that \(a = e^\gamma\), we have \(\frac{(a-1)^2}{a^2 - b^2} \leq \frac{(e^\gamma - 1)^2}{e^{2\gamma} - (1 + \frac{\gamma}{4})^2}\), so finally (\ref{eqn:pty_outer_eqn4}) reduces to
\[\frac{(e^\gamma - 1)^2}{e^{2\gamma} - (1 + \frac{\gamma}{4})^2} \le \frac{3}{2} \gamma\]
which is proved in \Cref{claim:outer_t_ineq}.
\end{proof}

Recall in the proof of \Cref{claim:translate_ellipse_angle} we defined \(\ell_1 = \frac{1}{c+a}, \ell_2 = \sqrt{\frac{1}{\alpha^2} - \frac{1}{(c+a)^2}}, r = \frac{a^2 \ell_1^2}{b^2 \ell_2^2}\).
That claim reduces to the following.
\begin{claim}\label{claim:pty_inner}
\[a - \alpha' \cdot a \sqrt{\frac{1+r}{r}} \geq 0\]
\end{claim}
\begin{proof}
As by definition \(a \geq 0\), it suffices to show
\begin{equation}\label{eqn:pty_inner_3}
    \alpha'^2 \cdot \left(\frac{1}{r}+1\right) \leq 1
\end{equation}
Observe that \(\ell_2^2 = \frac{1}{\alpha^2} - \ell_1^2\),
so we can write \(\frac{1}{r} = \frac{b^2}{a^2} \left(\frac{1}{\alpha^2 \ell_1^2} - 1\right)\), and hence rewrite (\ref{eqn:pty_inner_3}) as
\[\alpha'^2 \left(1 + \frac{b^2}{a^2} \left(\left(\frac{c+a}{\alpha}\right)^2 - 1\right)\right) \leq 1\]

Multiplying both sides by \(\frac{\alpha^2}{\alpha'^2}\) and rearranging, this is equivalent to
\begin{equation}\label{eqn:pty_inner_4}
\frac{b^2}{a^2} \left((c+a)^2 - \alpha^2\right) \leq \frac{\alpha^2}{\alpha'^2} - \alpha^2
\end{equation}
Now, by definition of \(c\) we can write \(c + a = a(1+\alpha') - \alpha\), so that \((c+a)^2 - \alpha^2 = a^2 (1+\alpha')^2 - 2 \alpha a (1 + \alpha')\).
Thus, (\ref{eqn:pty_inner_4}) is equivalent to
\[\frac{b^2}{a^2} (a^2 ( 1 + \alpha')^2 - 2 \alpha a (1+\alpha')) \leq \frac{\alpha^2}{\alpha'^2} (1 + \alpha')(1 - \alpha')\]
Dividing by \(1 + \alpha'\) and simplifying the left hand side, this is equivalent to
\[
b^2 \left(1 + \alpha' - \frac{2 \alpha}{a}\right) \leq \frac{\alpha^2}{\alpha'^2}(1-\alpha')
\]
which we show in \Cref{claim:pty_inner_main}.
\end{proof}

\begin{claim}\label{claim:pty_inner_main}
    \[
b^2 \left(1 + \alpha' - \frac{2 \alpha}{a}\right) \leq \frac{\alpha^2}{\alpha'^2}(1-\alpha')
\]
\end{claim}
\begin{proof}
Using \Cref{claim:wk_e}-(\ref{item:wk_1}), \(e^{-x} \geq 1 - x\); and the fact that by definition \(\frac{1}{a} = e^{-\gamma}\), it suffices to show
\[
b^2 \left(1 + \alpha' - 2 \alpha (1-\gamma)\right) \leq \frac{\alpha^2}{\alpha'^2}(1-\alpha') \nonumber
\]
Using \Cref{claim:outer_int}-(\ref{item:outer_int_4}), this reduces further to
\begin{equation}\label{eqn:pim_1}
    (1 + \alpha - \alpha')(1 + \alpha' - 2 \alpha ( 1 -  \gamma)) \leq \frac{\alpha^2}{\alpha'^2} (1 - \alpha')
\end{equation}
We expand both sides of this inequality into polynomials involving \(\gamma\) and \(\alpha\), and then analyze the resulting expression.
Using the definition of \(\alpha'\), we have \(\alpha ' = \frac{\alpha}{1+2\gamma \alpha}\), and thus
\(1 + \alpha' = \frac{1 + 2 \gamma \alpha + \alpha}{1 + 2 \gamma \alpha}\) and \(1 - \alpha' = \frac{1 + 2 \gamma \alpha - \alpha}{1 + 2 \gamma \alpha}\).
We also have \(\frac{\alpha}{\alpha'} = 1 + 2 \gamma \alpha\), and finally \(\alpha - \alpha' = \frac{2 \gamma \alpha^2}{1 + 2 \gamma \alpha}\).
Substituting these equalities into (\ref{eqn:pim_1}), we obtain the equivalent inequality
\[
\left(\frac{1+2 \gamma \alpha + \gamma \alpha^2}{1 + 2 \gamma \alpha}\right) \left(\frac{1 + 2 \gamma \alpha + \alpha}{1 + 2 \gamma \alpha} - 2 \alpha (1 - \gamma)\right) \leq (1 + 2 \gamma \alpha)^2 \left(\frac{1 + 2 \gamma \alpha - \alpha}{1 + 2 \gamma \alpha}\right)
\]
Multiplying both sides by \((1 + 2 \gamma \alpha)^2\) and rearranging the terms so that they are all on the same side, we get
\[(1 + 2 \gamma \alpha)^3 (1+2 \gamma \alpha - \alpha)-\left(1 + 2 \gamma \alpha + \gamma \alpha^2\right) (1+2 \gamma \alpha + \alpha - 2 \alpha (1- \gamma) (1 + 2 \gamma \alpha)) \geq 0\]
Next, we expand this inequality:
\begin{align*}
    16 \alpha ^4 \gamma^4-16 \alpha ^4 \gamma^3+8 \alpha ^4 \gamma^2+24 \alpha ^3 \gamma^3-12 \alpha ^3 \gamma^2+12
   \alpha ^2 \gamma^2+2 \alpha ^3 \gamma-2 \alpha ^2 \gamma+2 \alpha  \gamma \geq 0
\end{align*}
As \(\gamma \alpha \geq 0\), we can divide both sides of this inequality by \(2 \gamma \alpha\).
Grouping by powers of \(\alpha\), we obtain:
\[4 \alpha ^3 \gamma \left(2 \gamma^2-2 \gamma+1\right)+\alpha ^2 \left(12 \gamma^2-6 \gamma+1\right)+\alpha  (6
   \gamma-1)+1 \geq 0\]
Upon inspection, both quadratics \(2 \gamma^2 - 2\gamma + 1\) and \(12 \gamma^2 - 6 \gamma + 1\) are positive for all \(\gamma\).
Thus we only need to show \(\alpha (6\gamma-1)+1 \geq 0\), but this is clear from writing it as \(1 - \alpha + 6 \gamma \alpha \geq 0\) and using that \(\alpha \leq 1\).

\end{proof}

\newpage
\printbibliography[
    heading=bibintoc
]

@article{gk10,
title = {Determinants and the volumes of parallelotopes and zonotopes},
journal = {Linear Algebra and its Applications},
volume = {433},
number = {1},
pages = {28-40},
year = {2010},
author = {Eugene Gover and Nishan Krikorian},
}

@article{clarkson10,
author = {Clarkson, Kenneth L.},
title = {Coresets, Sparse Greedy Approximation, and the Frank-Wolfe Algorithm},
year = {2010},
issue_date = {August 2010},
volume = {6},
number = {4},
issn = {1549-6325},
journal = {ACM Trans. Algorithms},
month = {09},
articleno = {63},
numpages = {30}
}

@misc{bmv23,
      title={Tight Bounds for Volumetric Spanners and Applications}, 
      author={Aditya Bhaskara and Sepideh Mahabadi and Ali Vakilian},
      year={2023},
      eprint={2310.00175},
      archivePrefix={arXiv},
      primaryClass={cs.DS}
}

@article{ty07,
author = {Todd, Michael J. and Yildirim, E. Alper},
title = {On Khachiyan's Algorithm for the Computation of Minimum-Volume Enclosing Ellipsoids},
year = {2007},
issue_date = {August, 2007},
volume = {155},
number = {13},
journal = {Discrete Appl. Math.},
month = {08},
pages = {1731–1744},
numpages = {14},
}

@article{ky05,
author = {Kumar, P. and Yildirim, E. A.},
title = {Minimum-Volume Enclosing Ellipsoids and Core Sets},
year = {2005},
issue_date = {July      2005},
publisher = {Plenum Press},
address = {USA},
volume = {126},
number = {1},
journal = {J. Optim. Theory Appl.},
month = {07},
pages = {1–21},
numpages = {21}
}

@InProceedings{blum2017approximate,
  author =	{Avrim Blum and Vladimir Braverman and Ananya Kumar and Harry Lang and Lin F. Yang},
  title =	{Approximate Convex Hull of Data Streams},
  booktitle =	{Proceedings of the International Colloquium on Automata, Languages, and  Programming (ICALP)},
  pages =	{21:1--21:13},
  year =	{2018},
  volume =	{107},
}

@article{agarwal2005geometric,
  title={Geometric approximation via coresets},
  author={Agarwal, Pankaj K and Har-Peled, Sariel and Varadarajan, Kasturi R},
  journal={Combinatorial and computational geometry},
  volume={52},
  number={1},
  pages={1--30},
  year={2005}
}

@inproceedings{mmo22,
	title        = {Streaming Algorithms for Ellipsoidal Approximation of Convex Polytopes},
	author       = {Makarychev, Yury and Manoj, Naren Sarayu and Ovsiankin, Max},
	year         = 2022,
	booktitle    = {Proceedings of the Conference on Learning Theory}
}

@inproceedings{woodruff2022high,
  title={High-dimensional geometric streaming in polynomial space},
  author={Woodruff, David P and Yasuda, Taisuke},
  booktitle={Proceedings of the Symposium on Foundations of Computer Science},
  pages={732--743},
  year={2022}
}

@book{todd16,
	title        = {Minimum-Volume Ellipsoids: Theory and Algorithms},
	author       = {Todd, Michael J.},
	year         = 2016,
	publisher    = {SIAM-Society for Industrial and Applied Mathematics},
	address      = {Philadelphia, PA, USA},
	isbn         = 1611974372
}

@article{stange08,
	title        = {On the efficient update of the Singular Value Decomposition},
	author       = {Stange, Peter},
	year         = 2008,
	journal      = {PAMM},
	volume       = 8,
	number       = 1,
	pages        = {10827--10828}
}

@article{boyd97,
	title        = {Obstacle Collision Detection Using Best Ellipsoid Fit},
	author       = {Rimon, Elon and Boyd, Stephen P.},
	year         = 1997,
	journal      = {J. Intell. Robotics Syst.},
	publisher    = {Kluwer Academic Publishers},
	address      = {USA},
	volume       = 18,
	number       = 2,
	pages        = {105–126},
	issue_date   = {February 1997},
	numpages     = 22
}

@inproceedings{he21,
	title        = {Reducing Isotropy and Volume to KLS: An {$O^{*}(n^3\psi^2)$} Volume Algorithm},
	author       = {Jia, He and Laddha, Aditi and Lee, Yin Tat and Vempala, Santosh},
	year         = 2021,
	booktitle    = {Proceedings of the Symposium on Theory of Computing},
	pages        = {961–974},
	numpages     = 14
}

@incollection{john1948,
	title        = {Extremum problems with inequalities as subsidiary conditions},
	author       = {John, Fritz},
	year         = 1948,
	booktitle    = {Studies and Essays Presented to R. Courant on his 60th Birthday},
	publisher    = {Interscience Publishers, Inc},
	pages        = {187--204}
}

@inproceedings{mukhopadhyay2010approximate,
  title={Approximate ellipsoid in the streaming model},
  author={Mukhopadhyay, Asish and Sarker, Animesh and Switzer, Tom},
  booktitle={International Conference on Combinatorial Optimization and Applications},
  pages={401--413},
  year={2010}
}

@article{mukhopadhyayapproximate,
  title={Approximate minimum spanning ellipse in the streaming model},
  author={Mukhopadhyay, Asish and Greene, Eugene and Sarker, Animesh and Switzer, Tom}
}

@article{nesterov2008rounding,
  title={Rounding of convex sets and efficient gradient methods for linear programming problems},
  author={Nesterov, Yurii},
  journal={Optimisation Methods and Software},
  volume={23},
  number={1},
  pages={109--128},
  year={2008},
  publisher={Taylor \& Francis}
}

@article{howard1997john,
  title={The {John} ellipsoid theorem},
  author={Howard, Ralph},
  journal={University of South Carolina},
  year={1997}
}

@inproceedings{agarwal2010streaming,
  title={Streaming algorithms for extent problems in high dimensions},
  author={Agarwal, Pankaj K and Sharathkumar, R},
  booktitle={Proceedings of the  Symposium on Discrete Algorithms},
  pages={1481--1489},
  year={2010},
}

@book{artstein2015asymptotic,
  title={Asymptotic Geometric Analysis, Part I},
  author={Artstein-Avidan, Shiri and Giannopoulos, Apostolos and Milman, Vitali D},
  volume={202},
  year={2015},
  publisher={American Mathematical Soc.}
}

@inproceedings{bourgain2006estimates,
  title={Estimates related to Steiner symmetrizations},
  author={Bourgain, J and Lindenstrauss, J and Milman, V},
  booktitle={Geometric Aspects of Functional Analysis: Israel Seminar (GAFA) 1987--88},
  pages={264--273},
  year={2006},
  organization={Springer}
}

\end{document}